\documentclass[prx,
            aps,
            10pt,
            superscriptaddress,
            nofootinbib,
            a4paper,
            twocolumn,
            longbibliography,
            floatfix]{revtex4-2}

\usepackage{xcolor,
            amsmath,
            amsthm,
            amssymb,
            graphicx,
            mathtools,
            booktabs,
            tabularx,
            adjustbox,
            physics,
            bm}

\usepackage[colorlinks=true, linkcolor=blue, citecolor=gray]{hyperref}

\usepackage{orcidlink}

\usepackage{enumitem}

\usepackage[english]{babel}

\newtheorem{proposition}{Proposition}

\theoremstyle{remark}

\definecolor{myblue}{HTML}{a5bfd2}
\definecolor{myred}{HTML}{E64B35}

\usepackage[T1]{fontenc}
\usepackage{lmodern}

\usepackage[caption=false]{subfig}

\usepackage{tikz}
\usetikzlibrary{quantikz2}

\begin{document}

\title{Task-aligned spectral inductive biases in parameterized quantum circuits}

\author{Felix Frohnert \orcidlink{0000-0003-3717-6352}}
\email[E-mail: ]{f.frohnert@liacs.leidenuniv.nl}
\affiliation{$\langle a Q a^L \rangle$ Applied Quantum Algorithms, Universiteit Leiden, Leiden, The Netherlands}

\author{Adrián Pérez-Salinas\orcidlink{0000-0001-5430-6468}}  
\affiliation{$\langle a Q a^L \rangle$ Applied Quantum Algorithms, Universiteit Leiden, Leiden, The Netherlands}
\affiliation{Dahlem Center for Complex Quantum Systems, Freie Universität Berlin, Germany}
\affiliation{Institute for Theoretical Physics, ETH Zürich, Switzerland}
\affiliation{Instituut-Lorentz, Universiteit Leiden, Leiden, The Netherlands}

\author{Sebastian Lange\orcidlink{0009-0008-5357-1065}}
\affiliation{Airbus Defence and Space GmbH, Immenstaad am Bodensee, Germany}

\author{Marcus Schwaebisch\orcidlink{0009-0003-8529-3462}}
\affiliation{Airbus Defence and Space GmbH, Immenstaad am Bodensee, Germany}

\author{Evert van Nieuwenburg\orcidlink{0000-0003-0323-0031}}
\affiliation{$\langle a Q a^L \rangle$ Applied Quantum Algorithms, Universiteit Leiden, Leiden, The Netherlands}
\affiliation{Instituut-Lorentz, Universiteit Leiden, Leiden, The Netherlands}

\author{Vedran Dunjko\orcidlink{0000-0002-2632-7955}}
\affiliation{$\langle a Q a^L \rangle$ Applied Quantum Algorithms, Universiteit Leiden, Leiden, The Netherlands}

\begin{abstract}
Parameterized quantum circuits are widely investigated in quantum machine learning, yet how their architecture can provide task-aligned inductive biases remains poorly understood.
We study data reuploading quantum circuits with two distinct inputs: a signal that varies between samples and a fixed, task-defined spectral query that specifies how the signal should be probed. 
We show that encoding the query can change the circuit’s hypothesis class over the signal by modulating how its existing signal frequencies are combined, without introducing new sample-dependent information. 
This mechanism matches tasks built on nonuniform Fourier transforms, whose evaluation frequencies specify the query and whose sampling positions guide its encoding. 
We examine this connection in a classification task inspired by synthetic aperture radar, where acquisition geometry determines a query for classifying unfocused spectral patches before image formation. 
Across small-scale spectral and radar-like tasks, aligned queries improve learning from limited data over parameter-matched circuits with mismatched or absent queries. 
These results provide a route from known spectral task structure to the design and validation of inductive biases in parameterized quantum circuits.
\end{abstract}

\maketitle

\section{Introduction}
In quantum machine learning, parameterized quantum circuits (PQCs) are widely studied as trainable models for data-driven tasks~\cite{biamonte_quantum_2017,benedettiParameterizedQuantumCircuits2019,cerezoVariational2021}. 
They encode classical inputs into a quantum circuit, process them with parameterized unitaries, and produce measurement outcomes used by a classical optimizer to update the parameters~\cite{preskill_quantum_2018,bharti2022,cerezo2022challenges}.
While their empirical performance has been explored across a range of tasks~\cite{bowles2024better, huang2021power,kublerInductiveBiasQuantum2021,SchuldAdvantage2022}, we still lack a principled understanding of which practically relevant learning problems are naturally matched to the functional structure of PQCs~\cite{bowles_contextuality_2023,gili2024inductive}.

This correspondence between a learning problem and the functional structure of a model constitutes a task-aligned \emph{inductive bias}: the architecture emphasizes input-output relations with the structure expected of the target relation~\cite{mitchell1980,battaglia2018}.
Architectures in classical machine learning are often designed around such correspondences.
Convolutional neural networks exploit the spatial locality and approximate translation symmetry commonly found in images, while transformers use attention to model dependencies between potentially distant elements, as commonly encountered in language~\cite{lecun_gradient-based_1998,cohen2016,zeiler_visualizing_2013,vaswani_attention_2023,mauricioComparingVisionTransformers2023,dosovitskiy2021}.
In this sense, when these architectural properties reflect the structure of the data and target relation, they bias learning toward compatible functions, potentially reducing the data and model capacity required to approximate the target and improving generalization~\cite{prince2023understanding,elesedy2021,bronstein2021,maier2019,monga2021,karniadakis2021}.

PQCs have been studied extensively through the lenses of universality and trainability~\cite{mcclean2018barren,larocca2025barren}, whereas their task-aligned inductive biases remain comparatively less understood.
Symmetry-based biases have been developed systematically within geometric quantum machine learning, which provides a framework for building known task symmetries into PQCs~\cite{larocca2022group,meyer2023,nguyen2024,schatzki2024}.
Data reuploading circuits offer a complementary structure: their expectation values are generalized trigonometric polynomials (GTPs) of the input~\cite{schuldEffectDataEncoding2021,perez-salinasDataReuploadingUniversal2020,caro2021encoding}. 
This structure is well characterized, but it has not yet provided a comparably principled route to task-aligned design. 
Because the encoding fixes the accessible frequencies while the trainable parameters set their coefficients, prior work has focused on which frequencies a circuit can access and how they are learned, including engineered, trainable, and input-conditioned spectra~\cite{shin2023,jaderberg2024,bartheGradientsFrequencyProfiles2024,duffy2025,qsm2026}.
Many tasks, however, prescribe how measurements must be combined rather than which frequencies are present.
This structure acts on the coefficients, and how to impose it through the architecture remains open.

\begin{figure*}[ht!]
    \centering
    \includegraphics[width=1\textwidth]{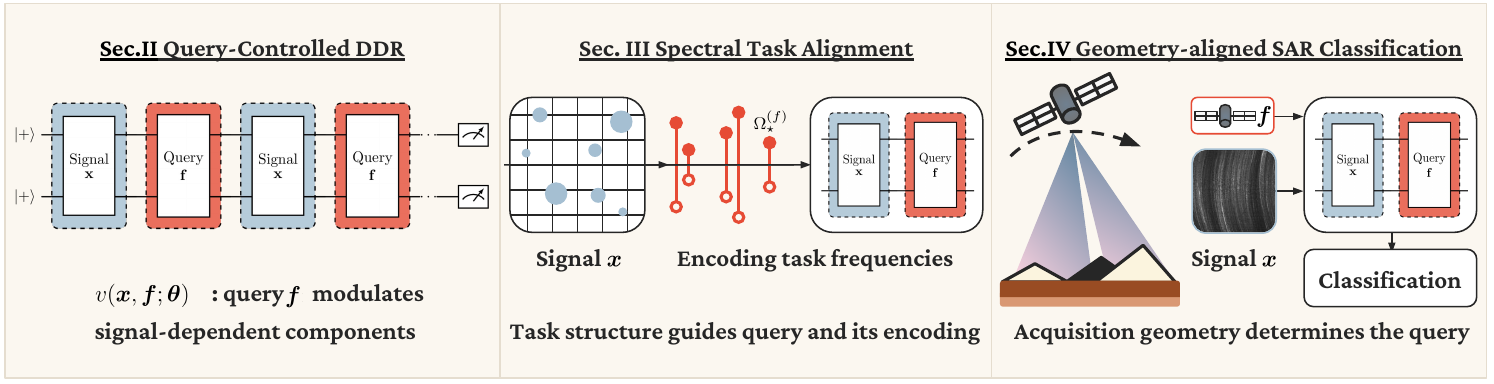}
    \caption{\textbf{From query-controlled circuits to geometry-aligned
    classification.}
    \textbf{(II)} In a DDR circuit, a fixed query controls how
    signal-dependent Fourier components are combined, modifying the
    represented function without providing sample-specific information
    (Sec.~\ref{sec:I}).
    \textbf{(III)} A spectral classification task supplies a fixed query
    $\boldsymbol f^\star$ and, through its sampling positions and score structure, a target spectrum $\Omega_\star^{(f)}$ that guides the query encoding
    (Sec.~\ref{sec:II}).
    \textbf{(IV)} In SAR, the acquisition geometry supplies this structure,
    enabling geometry-aligned classification directly from raw measurements
    without first reconstructing a focused image
    (Sec.~\ref{sec:III}).}
    \label{fig:sarOverview}
\end{figure*}

In this work, we encode known task structure as a second, fixed circuit input. 
We study what we denote as \emph{double data reuploading} (DDR) circuits, which encode two distinct inputs in alternating layers: 
a sample-dependent signal vector $\boldsymbol x$ and a query vector $\boldsymbol f$, determined from known task structure and fixed across samples. 
The terminology is motivated by the spectral tasks studied below: 
the signal contains measured sample values, while the query specifies the Fourier frequencies at which their transform is evaluated. 
Changing the query selects different transform evaluations without changing the sample values. 
Both are classical numerical inputs here; \emph{query} does not denote an oracle call.
The query carries no sample-specific information, yet its encoding can change which functions of $\boldsymbol x$ the circuit represents and the distribution over those functions induced by parameter sampling. 
It thus provides a way to incorporate known task structure into the model. 
We illustrate this approach in a task inspired by synthetic aperture radar (SAR), where acquisition geometry determines phase relations that can guide classification directly from unfocused measurements, before image formation.
We develop this idea through three successive questions, each examined through theory and numerical experiments (see Fig.~\ref{fig:sarOverview}).

\phantom{}

\emph{Query-controlled inductive bias: What can the query change?}---
Section~\ref{sec:I} derives the joint Fourier representation of DDR circuits, in which the query modulates the coefficients of a fixed signal spectrum.
A rank criterion bounds this modulation, and absorption conditions distinguish whether it leaves the model unchanged, reweights its function prior, or changes its hypothesis class.
Two-qubit circuits realize all three cases while preserving the accessible frequencies.

\emph{Aligning the query with spectral tasks: When is the query bias task-aligned?}---
Section~\ref{sec:II} connects this mechanism to classification tasks defined by fixed nonuniform Fourier evaluations~\cite{bagchi2012nonuniform}: their frequencies supply the query value $\boldsymbol f^\star$, while their sampling positions supply the modes that guide its encoding.
We bound the approximation error by separating phase alignment from the signal-dependent coefficients the circuit must learn.
In teacher--student experiments, aligned queries improve learning from limited data relative to parameter-matched circuits with mismatched or absent queries, and the advantage persists under moderate query perturbations.

\emph{Geometry-aligned classification of raw SAR-like data: Where does the query come from?}---
Nonuniform Fourier transforms arise when measurements are irregularly sampled or their transforms must be evaluated at nonuniform frequencies, as in synthetic aperture radar, non-Cartesian magnetic resonance imaging, and radio interferometry~\cite{fessler2003nonuniform,greengard2004,thompson2017}.
In SAR, the acquisition geometry and radar frequency band determine the nonuniform spatial-frequency relation used in image formation~\cite{moreira2013tutorial,cumming2005digital}.
Classifying directly from raw SAR data could avoid costly image formation, an attractive possibility when onboard processing or downlink capacity is limited~\cite{floresSyntheticApertureRadar2019,garciaAdvancementsOnboardProcessing2024}; classical methods have begun to explore this approach~\cite{gaburro2025}.
Quantum approaches, meanwhile, have classified formed images~\cite{naik2024,maritimeqkm2025} or targeted image formation itself~\cite{giovagnoli2025,alsalahat2025}.
Section~\ref{sec:III} derives a fixed query from this image-formation relation, before any data or labels are seen, and uses it to classify unfocused SAR-like phase histories.
The empirically preferred query tracks the geometry-derived prediction across acquisitions, and supplying the predicted query improves classification when labeled scenes are scarce.

Section~\ref{sec:conclusion-pqc} discusses the implications and limitations of this approach.

\section{Query-Controlled Inductive Bias} \label{sec:I}
\subsection{Double data reuploading function class}\label{sec:architecture}
\subsubsection{Data reuploading circuits as generalized trigonometric polynomials}
\label{sec:simple-pqc}

We recall the representation of data reuploading PQCs as GTPs~\cite{benedettiParameterizedQuantumCircuits2019,schuldEffectDataEncoding2021,perez-salinasDataReuploadingUniversal2020}.
Consider an $n_q$-qubit circuit initialized in
$\ket{\boldsymbol{0}}:=\ket{0}^{\otimes n_q}$, with an input
$\boldsymbol{x}=(x_1,\ldots,x_M)\in\mathbb{R}^{M}$ and trainable parameters
$\boldsymbol{\theta}$.
Its output is the expectation value of a Hermitian observable $O$,
\begin{equation}
v(\boldsymbol{x};\boldsymbol{\theta})
=
\bra{\boldsymbol{0}}
U^\dagger(\boldsymbol{x};\boldsymbol{\theta})
O
U(\boldsymbol{x};\boldsymbol{\theta})
\ket{\boldsymbol{0}}.
\label{eq:circ-exp}
\end{equation}
The circuit has $L$ data reuploading blocks, each applying an encoding followed by a trainable unitary (Fig.~\ref{fig:encoding1}):
\begin{equation}
U(\boldsymbol{x};\boldsymbol{\theta})
=
V_L(\boldsymbol{\theta})E_L(\boldsymbol{x})
\cdots
V_1(\boldsymbol{\theta})E_1(\boldsymbol{x}).
\label{eq:data-reupl}
\end{equation}
The layerwise encodings and their collection are
\begin{equation}
E_\ell(\boldsymbol{x})
=
\exp\left(
-i\sum_{m=1}^{M}x_mH_{\ell,m}
\right),
\qquad
\mathcal E=\{E_\ell\}_{\ell=1}^{L},
\label{eq:encoding-choice}
\end{equation}
where $H_{\ell,m}$ is the Hermitian generator associated with feature $x_m$ in layer $\ell$.
We assume that the generators within each layer commute pairwise; commuting Pauli strings are a common choice.
Commuting generators therefore admit a joint eigenbasis in each layer, with any input-independent basis changes entering the Fourier coefficients alongside the trainable blocks~\cite{schuldEffectDataEncoding2021}.

Expanding in these joint eigenbases and grouping terms with the same frequency gives~\cite{schuldEffectDataEncoding2021}
\begin{equation}
v(\boldsymbol{x};\boldsymbol{\theta})
=
\sum_{\boldsymbol{\omega}\in\Omega}
c_{\boldsymbol{\omega}}(\boldsymbol{\theta})
\exp\left(i\boldsymbol{\omega}\cdot\boldsymbol{x}\right).
\label{eq:pqcGTP}
\end{equation}
Thus, the output is a GTP in $\boldsymbol{x}$, with possible frequencies $\Omega$ determined by the eigenvalues of the chosen encodings $\mathcal E$.
The full derivation is given in Appx.~\ref{appx:simple-pqc}.

For a fixed encoding $\mathcal E$, trainable architecture, and observable, varying $\boldsymbol{\theta}$ defines the hypothesis class of functions the circuit can represent:
\begin{equation}
\mathcal{H}_{\mathcal{E}}
=
\left\{
\boldsymbol{x}\mapsto
\sum_{\boldsymbol{\omega}\in\Omega}
c_{\boldsymbol{\omega}}(\boldsymbol{\theta})
\exp\left(i\boldsymbol{\omega}\cdot\boldsymbol{x}\right)
:
\boldsymbol{\theta}\in\Theta
\right\}.
\label{eq:hypothesis-class-pqc}
\end{equation}
Here, $\Theta$ denotes the trainable parameter domain.
At finite depth, the encoding fixes the possible Fourier basis functions, while the rest of the circuit constrains which coefficient vectors
$\bigl(c_{\boldsymbol{\omega}}(\boldsymbol{\theta})\bigr)_{\boldsymbol{\omega}\in\Omega}$ can be realized.
Together, these restrictions define a representational inductive bias: 
circuits with the same frequency spectrum need not have the same hypothesis class.

Under suitable conditions, increasingly expressive data reuploading PQCs can approximate arbitrary continuous functions~\cite{perez-salinasDataReuploadingUniversal2020,perez-salinasUniversalApproximationContinuous2025,manzanoApproximationGeneralizationCapacities2025}.
Every finite circuit nevertheless retains the GTP structure of Eq.~\eqref{eq:pqcGTP}.
This raises the question we address next: 
can a second, fixed input change the attainable coefficient vectors and thus the inductive bias of the model?

\begin{figure}[t!]
    \centering

    \subfloat[Single data reuploading circuit%
    \label{fig:encoding1}]{%
        \begin{adjustbox}{width=\columnwidth}
            \begin{quantikz}
                \lstick{$\ket{0}$}
                & \qwbundle{n_q}
                & \qw
                & \qw
                & \qw
                & \qw
                & \gate{E_\ell(\boldsymbol{x})}
                \gategroup[
                    1,
                    steps=2,
                    style={
                        dashed,
                        rounded corners,
                        fill=myblue!80,
                        inner xsep=2pt
                    },
                    background,
                    label style={
                        label position=below,
                        anchor=north,
                        yshift=-0.2cm
                    }
                ]{{\sc $\ell = 1,\ldots,L$ times}}
                & \gate{V_\ell(\boldsymbol{\theta})}
                & \qw
                & \qw
                & \qw
                & \qw
                & \ \ldots\
                & \meter{}
            \end{quantikz}
        \end{adjustbox}
    }

    \vspace{0.6cm}

    \subfloat[Double data reuploading circuit%
    \label{fig:encoding2}]{%
        \begin{adjustbox}{width=\columnwidth}
            \begin{quantikz}
                \lstick{$\ket{0}$}
                & \qwbundle{n_q}
                & \gate{E^{(f)}_\ell(\boldsymbol{f})}
                \gategroup[
                    1,
                    steps=4,
                    style={
                        dashed,
                        rounded corners,
                        fill=myblue!80,
                        inner xsep=2pt
                    },
                    background,
                    label style={
                        label position=below,
                        anchor=north,
                        yshift=-0.2cm
                    }
                ]{{\sc $\ell = 1,\ldots,L/2$ times}}
                & \gate{V^{(f)}_\ell(\boldsymbol{\theta})}
                & \gate{E^{(x)}_\ell(\boldsymbol{x})}
                & \gate{V^{(x)}_\ell(\boldsymbol{\theta})}
                & \ \ldots\
                & \meter{}
            \end{quantikz}
        \end{adjustbox}
    }

    \caption{\textbf{Parameterized quantum circuits
    with data reuploading architectures:}
    \textbf{(a)} The $n_q$-qubit single data reuploading scheme consisting of $L$ layers.
    Classical input features $\boldsymbol{x}$ are re-encoded at each layer
    via $E_\ell(\boldsymbol{x})$ alongside parameterized unitaries
    $V_\ell(\boldsymbol{\theta})$.
    \textbf{(b)} The $n_q$-qubit double data reuploading scheme extended to two input
    streams, $\boldsymbol{x}$ and $\boldsymbol{f}$, structured across $L/2$
    blocks. Each block alternates between encoding unitaries
    $E_\ell^{(f)}(\boldsymbol{f})$ and
    $E_\ell^{(x)}(\boldsymbol{x})$ and parameterized unitaries
    $V^{(f)}_\ell(\boldsymbol{\theta})$ and
    $V^{(x)}_\ell(\boldsymbol{\theta})$.}
    \label{fig:data_reuploading_comparison}
\end{figure}
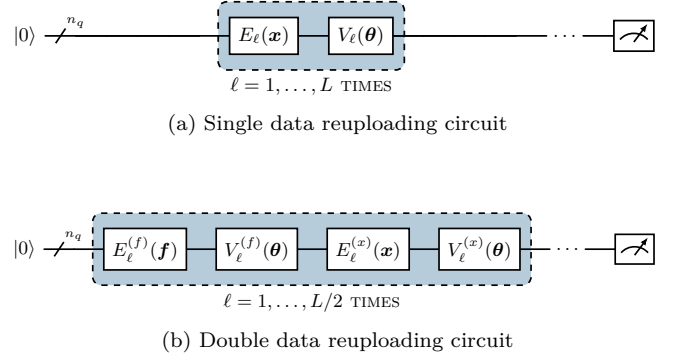

\subsubsection{Double data reuploading and query-modulated coefficients}
\label{sec:expanding-data-rep}

We now introduce a second input with its own encoding. 
We call the sample-dependent input $\boldsymbol{x}\in\mathbb{R}^{M}$ the \emph{signal} and the sample-independent, task-defined input $\boldsymbol{f}\in\mathbb{R}^{N}$ the \emph{query}, anticipating their roles in the spectral learning tasks below.
Each encoding layer acts on one of these inputs, with trainable unitaries between consecutive encodings. 
Figure~\ref{fig:encoding2} shows an example of an alternating gate arrangement. 

Applying the Fourier expansion to both inputs gives
\begin{align}
v(\boldsymbol{x},\boldsymbol{f};\boldsymbol{\theta})
&=
\sum_{\substack{
\boldsymbol{\omega}^{(x)}\in\Omega^{(x)}\\
\boldsymbol{\omega}^{(f)}\in\Omega^{(f)}
}}
c_{\boldsymbol{\omega}^{(x)},\boldsymbol{\omega}^{(f)}}
(\boldsymbol{\theta}) \nonumber\\
&\quad\times
\exp\left(
i\boldsymbol{\omega}^{(x)}\cdot\boldsymbol{x}
+i\boldsymbol{\omega}^{(f)}\cdot\boldsymbol{f}
\right).
\label{eq:joint-ddr-gtp}
\end{align}
Here $\Omega^{(x)}$ and $\Omega^{(f)}$ are the possible frequencies
set by the respective encoding families $\mathcal E_x$ and
$\mathcal E_f$. 
The full derivation of the path expansion is given in
Appx.~\ref{appx:expanding-data-rep}.

To see how the query affects the function of $\boldsymbol{x}$, arrange
the joint coefficients into a matrix and define phase vectors for the
two inputs:
\begin{align}
[\mathbf C(\boldsymbol{\theta})]_{
\boldsymbol{\omega}^{(x)},\boldsymbol{\omega}^{(f)}}
&=
c_{\boldsymbol{\omega}^{(x)},\boldsymbol{\omega}^{(f)}}
(\boldsymbol{\theta}),
\label{eq:encoding-dependent-matrix}\\
[\boldsymbol w(\boldsymbol f)]_{\boldsymbol{\omega}^{(f)}}
&=
\exp(i\boldsymbol{\omega}^{(f)}\cdot\boldsymbol f),
\label{eq:encoding-dependent-phase-vector}\\
[\boldsymbol e(\boldsymbol x)]_{\boldsymbol{\omega}^{(x)}}
&=
\exp (i\boldsymbol{\omega}^{(x)}\cdot\boldsymbol x).
\label{eq:x-fourier-vector}
\end{align}
At fixed $\boldsymbol f$, the coefficient vector over the signal
frequencies is
\begin{equation}
\boldsymbol a(\boldsymbol{\theta},\boldsymbol f)
=
\mathbf C(\boldsymbol{\theta})\boldsymbol w(\boldsymbol f),
\label{eq:encoding-dependent-coefficients}
\end{equation}
so that
\begin{equation}
v(\boldsymbol{x},\boldsymbol f;\boldsymbol{\theta})
=
\boldsymbol a(\boldsymbol{\theta},\boldsymbol f)^{\mathsf T}
\boldsymbol e(\boldsymbol x).
\label{eq:ddr-contraction}
\end{equation}
Thus, changing the query modulates the coefficients of a fixed signal
frequency set. For a fixed query encoding and query value, this defines
the hypothesis class over the signal $\bm x$:
\begin{equation}
\mathcal H_{\mathcal E_f}(\boldsymbol f)
=
\left\{
\boldsymbol x\mapsto
\boldsymbol a(\boldsymbol{\theta},\boldsymbol f)^{\mathsf T}
\boldsymbol e(\boldsymbol x)
:
\boldsymbol{\theta}\in\Theta
\right\}.
\label{eq:ddr-hypothesis-class}
\end{equation}
Here the signal encoding, trainable architecture, and observable are held fixed. 
The query encoding affects the family of matrices
$\mathbf C(\boldsymbol{\theta})$, the parameters select a matrix from that family, and the query value supplies the phases with which it is
contracted. 
We next characterize when this modulation changes the hypothesis class.

\subsection{When does the query change the model?}
\label{sec:query-bias}

Equation~\eqref{eq:ddr-contraction} illustrates that the query changes the coefficients of a Fourier expansion over the signal. 
A change in coefficients at fixed parameters, however, need not change the functions
available to the model: 
another parameter choice might produce the same function. 
In the following, we therefore distinguish three levels: 
the accessible signal frequencies, their modulation at fixed parameters, and changes to the hypothesis class or the function prior.

\subsubsection{Frequency support and coefficient modulation}

First, the query cannot add frequencies to the signal expansion.

\begin{proposition}[Separation of encoding spectra]
\label{thm:frequency-independence}
For disjoint signal- and query-encoding layers with commuting generators within each layer, the joint path-frequency set factorizes as
\begin{equation}
\Omega=\Omega^{(x)}\times\Omega^{(f)}.
\label{eq:factorize}
\end{equation}
In particular, for any query encoding $\mathcal E_f$ and query value
$\boldsymbol f$,
\begin{equation}
\mathcal H_{\mathcal E_f}(\boldsymbol f)
\subseteq
\mathcal V_x:=
\operatorname{span}_{\mathbb C}
\left\{
\exp(i\boldsymbol\omega^{(x)}\cdot\boldsymbol x)
:
\boldsymbol\omega^{(x)}\in\Omega^{(x)}
\right\},
\label{eq:common-x-fourier-space}
\end{equation}
where $\mathcal V_x$ is fixed by the signal encoding.
\end{proposition}

The product in Eq.~\eqref{eq:factorize} describes possible frequencies; 
the corresponding coefficients may vanish for a particular circuit.
The relevant effect of the query must therefore lie in the coefficients in Eq.~\eqref{eq:encoding-dependent-coefficients}.
At fixed parameters, its extent is constrained by the joint coefficient matrix.
The proof is given in Appx.~\ref{appx:proof-spectral-sep}; Appx.~\ref{appx:concat} compares DDR with concatenated-input encoding.

\begin{proposition}[Dimension of query-induced modulation]
\label{prop:query-modulation}
Fix $\mathcal E_f$, $\boldsymbol\theta$, and a query domain
$\mathcal F$, and define the explored phase space
\begin{equation}
\mathcal W=
\operatorname{span}_{\mathbb C}
\{\boldsymbol w(\boldsymbol f):\boldsymbol f\in\mathcal F\}.
\label{eq:query-phase-space}
\end{equation}
As $\boldsymbol f$ varies, the signal coefficient vectors span
$\mathbf C(\boldsymbol\theta)\mathcal W$, with
\begin{equation}
\dim_{\mathbb C}\mathbf C(\boldsymbol\theta)\mathcal W
\leq
\min\!\left\{
\operatorname{rank}\mathbf C(\boldsymbol\theta),
\dim_{\mathbb C}\mathcal W
\right\}.
\label{eq:rank-bound}
\end{equation}
The rank bound is saturated if $\mathcal W$ is the full
query-frequency space.
\end{proposition}

Thus, the rank limits how many independent signal coefficient vectors the query can explore at fixed parameters. 
A one-dimensional image makes these vectors proportional, although their overall scale may still change. 
Even a higher-dimensional image does not establish a change in the hypothesis class: 
varying the trainable parameters might absorb the effect of the query.
The proof is given in Appx.~\ref{appx:proof-modu-bound}.

\subsubsection{Query absorption and three possible effects}

To test for a model-level effect, we compare the hypothesis classes
$\mathcal H_{\mathcal E_f}(\boldsymbol f)$ at different query values.
We also fix a parameter distribution $\mu$; drawing
$\boldsymbol\theta\sim\mu$ induces a distribution
$\mathbb P_{\mathcal E_f,\boldsymbol f}$ on
$\mathcal H_{\mathcal E_f}(\boldsymbol f)$, which we call the function
prior. 
For any measurable
$\mathcal A\subseteq\mathcal H_{\mathcal E_f}(\boldsymbol f)$,
\begin{equation}
\mathbb P_{\mathcal E_f,\boldsymbol f}(\mathcal A)
=
\mu\!\left\{\boldsymbol\theta:
v(\cdot,\boldsymbol f;\boldsymbol\theta)\in\mathcal A\right\}.
\label{eq:function-prior}
\end{equation}
Thus, the class specifies which functions are realizable, while the prior specifies how parameter sampling weights them.
We call a query absorbable when its effect can be removed by reparameterizing the trainable blocks. The following condition establishes absorption into adjacent blocks.

\begin{proposition}[Sufficient condition for query absorption]
\label{prop:absorption}
Suppose each query encoding gate $E_\ell^{(f)}(\boldsymbol f)$ is followed by an
independently parameterized block
$B_\ell(\boldsymbol\theta_\ell)$, whose realizable set is
$\mathcal B_\ell$. If, for every query layer and every
$\boldsymbol f\in\mathcal F$,
\begin{equation}
\mathcal B_\ell E_\ell^{(f)}(\boldsymbol f)=\mathcal B_\ell,
\label{eq:absorption-condition}
\end{equation}
then the hypothesis class is independent of the query value:
\begin{equation}
\mathcal H_{\mathcal E_f}(\boldsymbol f)
=
\mathcal H_{\mathcal E_f}(\boldsymbol f')
\quad
\text{for all }\boldsymbol f,\boldsymbol f'\in\mathcal F.
\label{eq:absorbed-hypothesis-class}
\end{equation}
If these blocks are sampled independently of one another and of the remaining circuit parameters, and each block distribution is invariant under right multiplication by its query encoding gate, 
then the function prior is independent of the query value:
\begin{equation}
\mathbb P_{\mathcal E_f,\boldsymbol f}
=
\mathbb P_{\mathcal E_f,\boldsymbol f'}
\quad
\text{for all }\boldsymbol f,\boldsymbol f'\in\mathcal F.
\label{eq:absorbed-function-prior}
\end{equation}
\end{proposition}

For example, a block spanning a unitary group absorbs query encoding gates in that group; a Haar-distributed block also preserves the function prior.
Product single-qubit blocks can similarly absorb local query encoding gates, whereas generic entangling gates lie outside that product group.
Absorption is a sufficient test: failure of
Eq.~\eqref{eq:absorption-condition} does not rule out a reparameterization involving several layers.
The proof is given in Appx.~\ref{appx:proof-q-absor}.

These results distinguish three possible effects of varying
$\boldsymbol f$ at fixed $\mathcal E_f$:

\begin{enumerate}
\item[(I)] \emph{No model-level change.} The hypothesis class and
function prior are both invariant.

\item[(II)] \emph{Prior reweighting.} The hypothesis class is invariant,
but the function prior changes.

\item[(III)] \emph{Hypothesis-class change.} The set of representable
functions depends on the query value.
\end{enumerate}

All three cases share the signal frequency space $\mathcal V_x$.
Case~(II) changes how functions in a fixed class are weighted, while
Case~(III) changes which functions are available.
We target Case~(III), where task structure supplied through the query can enter the hypothesis class itself.
By Proposition~\ref{prop:absorption}, this requires the absorption condition to fail for at least one layer and query value.
The numerical experiments below distinguish these effects.

\begin{figure*}[ht!]
    \centering
    \includegraphics[width=1\linewidth]{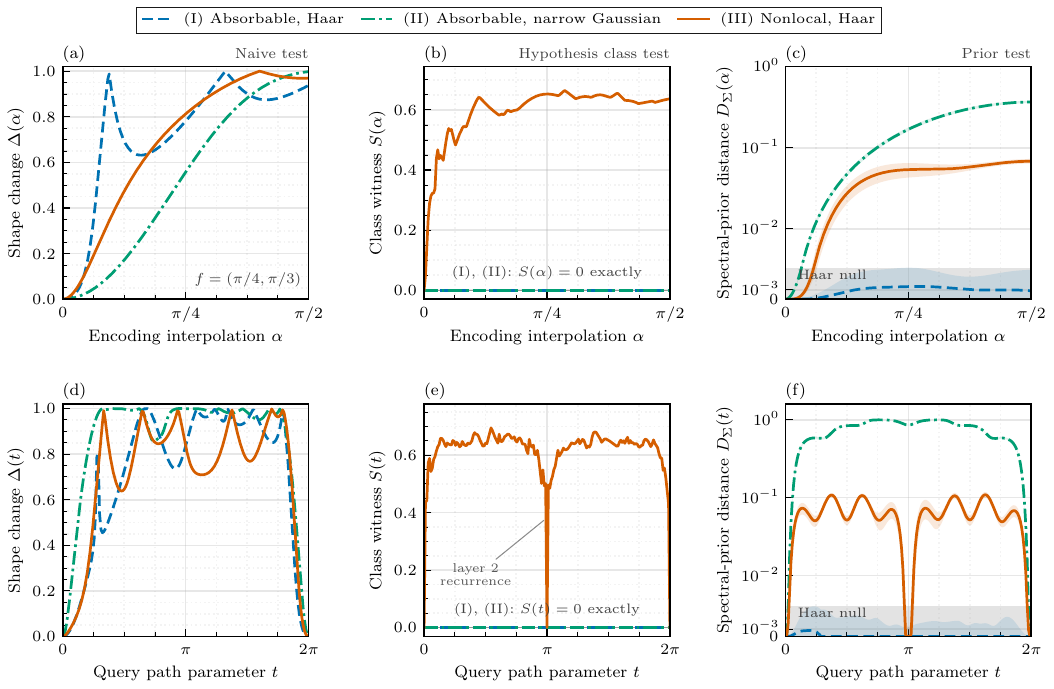}
    \caption{
    \textbf{Changing the query encoding or the query value can alter the hypothesis class and the function prior.}
    The top row varies the encoding at fixed $\boldsymbol f=(\pi/4,\pi/3)$; the bottom row varies $\boldsymbol f(t)=t\boldsymbol 1$ at fixed encoding. In each row, the three series represent (I) an absorbable local encoding with Haar-distributed parameters, (II) the same encoding with a Gaussian parameter prior of width $0.35$, and (III) a nonlocal encoding with Haar-distributed parameters. The comparisons preserve the accessible query-frequency spectrum $\Omega^{(f)}$, isolating changes beyond frequency support.
    \textbf{(a--c) Encoding interpolation.}
    The generators $H_{\ell,n}^{(f)}(\alpha)=\frac{s_\ell}{2}[\cos(\alpha)A_n+\sin(\alpha)B_n]$ start from $(A_0,A_1)=(X_0,X_1)$. The absorbable path ends at $(B_0,B_1)=(Y_0,Y_1)$; the nonlocal path ends at $(Y_0Z_1,Z_0Y_1)$. Anticommutation preserves the generator eigenvalues along both paths.
    \textbf{(d--f) Query variation.}
    We compare the local generators $(X_0,X_1)$ with the nonlocal pair $(Y_0Z_1,Z_0Y_1)$, using layer weights $(s_1,s_2)=(1,2)$ and the same product trainable blocks.
    The columns report complementary diagnostics: \textbf{(a,d)} the normalized change $\Delta$ of the Fourier coefficient vector for a fixed parameter realization; \textbf{(b,e)} the class witness $S$, measuring the largest nonseparable coefficient weight across $x_0|x_1$; and \textbf{(c,f)} the normalized distance $D_\Sigma$ between spectral second moments and their value at $\alpha=0$ or $t=0$.
    Individual functions change in all three cases, so $\Delta$ alone does not identify a change in inductive bias. The remaining diagnostics distinguish the cases: $S$ and $D_\Sigma$ remain at zero or within the Haar sampling null in Case~(I); only $D_\Sigma$ changes in Case~(II); and both change in Case~(III). In (e), $S(t)$ vanishes at $t\in\{0,\pi,2\pi\}$, where the nonlocal query encoding gates reduce to product unitaries. The $D_\Sigma$ axes are linear within the shaded sampling-null region and logarithmic above it. Definitions, circuit details, and sampling procedures are given in Appx.~\ref{appx:encoding-interpolation} and Appx.~\ref{appx:vary-query}.
    }
    \label{fig:encoding-value-interpolation}
\end{figure*}

\subsection{Numerically distinguishing the three cases}\label{sec:controlled-query}

We test the two distinct changes suggested by Eq.~\eqref{eq:encoding-dependent-coefficients}: 
varying the encoding changes the matrix $\mathbf C$ at fixed $\boldsymbol f$, whereas varying $\boldsymbol f$ changes the phase vector $\boldsymbol w$ at fixed encoding. 
Both tests use two-qubit, two-block DDR circuits in which $x_0$ and $x_1$ are encoded by single-qubit generators on separate qubits and the trainable blocks are products of single-qubit unitaries.
They compare query encodings with identical eigenvalues, so an observed difference cannot be attributed to a new frequency set.
The following experiments involve no training; they probe the hypothesis class and the function prior directly.

Figure~\ref{fig:encoding-value-interpolation} reports three
complementary diagnostics, all computed from the signal coefficient
vector $\boldsymbol a$ of Eq.~\eqref{eq:encoding-dependent-coefficients}.
Because $\Omega^{(x)}=\Omega^{(x)}_0\times\Omega^{(x)}_1$, we can reshape
$\boldsymbol a$ into a matrix $\mathbf A$ indexed by the frequencies of
$x_0$ and $x_1$; the represented function is separable across $x_0|x_1$
if and only if $\mathbf A$ has rank at most one. With subscript $0$
denoting the baseline ($\alpha=0$ or $t=0$), we define
\begin{equation}
\begin{aligned}
\Delta &= 1-\frac{|\boldsymbol a^\dagger\boldsymbol a_0|}
{\|\boldsymbol a\|_2\,\|\boldsymbol a_0\|_2},
\qquad
S = \max_{\boldsymbol\theta}\left[1-\frac{\sigma_1^2}{\sum_j\sigma_j^2}\right],\\
D_\Sigma &= 1-\frac{\operatorname{Tr}[\boldsymbol\Sigma\boldsymbol\Sigma_0]}
{\|\boldsymbol\Sigma\|_{\mathrm F}\,\|\boldsymbol\Sigma_0\|_{\mathrm F}},
\qquad
\boldsymbol\Sigma=\mathbb E_{\boldsymbol\theta}\big[\boldsymbol a\boldsymbol a^\dagger\big],
\end{aligned}
\label{eq:diagnostics}
\end{equation}
where $\sigma_1\ge\sigma_2\ge\dots$ are the singular values of
$\mathbf A$. The shape change $\Delta$ detects a change in the function
realized at fixed parameters, independent of overall scale. The class
witness $S$ measures the largest fraction of coefficient power outside
the best rank-one approximation. 
At the baseline, the circuit acts as a product of single-qubit circuits and every function in its class is separable.
Thus, $S>0$ exhibits a function absent from the baseline class
and certifies a class change; as a sampled maximum, it is a lower bound.
A nonzero $D_\Sigma$ certifies a change in the function prior, whereas
$D_\Sigma=0$ alone would not establish that the full priors agree.

\emph{Changing the encoding.} At $\boldsymbol f=(\pi/4,\pi/3)$, we
interpolate the query generators as
$H^{(f)}_{\ell,n}(\alpha)=\frac{s_\ell}{2}[\cos(\alpha)A_n+\sin(\alpha)B_n]$ with
$\alpha\in[0,\pi/2]$ and layer weights $(s_1,s_2)=(1,2)$, starting from
$(A_0,A_1)=(X_0,X_1)$. The local path ends at $(B_0,B_1)=(Y_0,Y_1)$,
the nonlocal path at $(Y_0Z_1,Z_0Y_1)$. Since each $A_n$ anticommutes
with $B_n$, the generators retain the joint eigenvalues
$\frac{s_\ell}{2}(\pm1,\pm1)$ for all $\alpha$. The local path remains absorbable
by adjacent single-qubit blocks, while the nonlocal path can entangle
the qubits. In Figs.~\ref{fig:encoding-value-interpolation}(a--c),
$\Delta$ changes along all three series. On the local path, $S$ remains
exactly zero; Haar sampling also leaves the prior invariant within
sampling error (I), whereas a full-support non-Haar distribution changes
it (II). On the nonlocal path, $S>0$ establishes functions absent from
the baseline class, and $D_\Sigma>0$ also detects a prior change (III).
This separates the effect of encoding eigenspaces from that of frequency
support.

\emph{Changing the query value.} 
We next fix either the local generators $(X_0,X_1)$ or the nonlocal generators $(Y_0Z_1,Z_0Y_1)$, with the same layer weights, and sweep $\boldsymbol f(t)=t\boldsymbol 1$. 
Their joint eigenvalues again agree. 
Figures~\ref{fig:encoding-value-interpolation}(d--f) show $\Delta(t)>0$ even for the local Haar case (I): an individual circuit output changes, although the model class and prior do not. 
With a non-Haar prior the local encoding realizes (II), while the nonlocal encoding yields $S(t)>0$ at generic $t$, establishing (III). 
Its class witness vanishes when the query encoding gates reduce to product unitaries. 
Circuit and statistical details are given in Appx.~\ref{appx:encoding-interpolation} and Appx.~\ref{appx:vary-query}.

Together, these results identify the model-level effects of the query branch within a fixed signal-frequency spectrum, and illustrate that the target Case~(III) requires a query encoding that the trainable blocks cannot absorb. 
To use this mechanism as an inductive bias, the query must be chosen from the structure of the learning task.

\section{Aligning the query with spectral tasks}\label{sec:II}
\subsection{Classification from fixed nonuniform Fourier queries}

Section~\ref{sec:I} showed how a query can change the functions represented over the signal.
We now consider tasks that prescribe how signal samples should be combined, allowing the query and its encoding to be chosen before training.
Learning problems containing fixed nonuniform Fourier evaluations provide such a setting.

\subsubsection{Spectral classification tasks}
\label{sec:mnudft}

Let $\boldsymbol{x}=(x_{\boldsymbol d})_{\boldsymbol d\in\mathcal I}$
be a discretely sampled complex signal indexed by $\mathcal I$, with $x_{\boldsymbol d}$
acquired at position $\boldsymbol p_{\boldsymbol d}\in\mathbb R^n$.
When supplied to the circuit, its complex samples are represented by real parameters: their real and imaginary parts (Appx.~\ref{appx:real}) or, for phase-only signals, their phases.
The multidimensional nonuniform discrete Fourier transform (MNUDFT)
evaluates the signal at frequencies
$\{\boldsymbol f_{\boldsymbol k}\}_{\boldsymbol k\in\mathcal K}$ indexed by $\mathcal K$:
\begin{equation}
X_{\boldsymbol k}(\boldsymbol x,\boldsymbol f_{\boldsymbol k})
=
\sum_{\boldsymbol d\in\mathcal I}x_{\boldsymbol d}
\exp\!\left(
-2\pi i\boldsymbol f_{\boldsymbol k}\cdot
\boldsymbol p_{\boldsymbol d}
\right).
\label{eq:multi-nudft}
\end{equation}
Each evaluation frequency specifies a phase pattern for combining the samples. 
Neither the sampling positions nor the evaluation frequencies need to form a regular Cartesian grid (see Appx.~\ref{appx:mnudft} for more details).

A target score is a real-valued function of the signal used to assign its class label, for example by thresholding one score or comparing several class scores~\cite{prince2023understanding}.
We call a classification task \emph{spectral} when its target scores depend on the signal only through a prescribed finite set of nonuniform Fourier evaluations:
\begin{equation}
s^\star(\boldsymbol x\mid\boldsymbol f^\star)
=
G\!\left(
[X_{\boldsymbol k}(\boldsymbol x,\boldsymbol f^\star_{\boldsymbol k})]
_{\boldsymbol k\in\mathcal K}
\right).
\label{eq:target-mnudft}
\end{equation}
Labels are assigned by a fixed rule from one or more such target scores. 
The map $G$ may be nonlinear, for example through spectral intensities or interactions between evaluations
~\cite{nikias1993signal,boashash2015time}.

For a given acquisition or task instance, the sampling positions and evaluation frequencies are fixed across signals and specified independently of their labels. 
The frequencies $\boldsymbol f^\star=(\boldsymbol f^\star_{\boldsymbol k})_{\boldsymbol k\in\mathcal K}$ therefore supply the query value, while the positions
$\boldsymbol p_{\boldsymbol d}$ determine the phase modes that guide its encoding. 
The classifier must learn how to use these evaluations; 
the query supplies where the task calls for evaluating them.

\subsubsection{Query Fourier structure of spectral scores}
\label{sec:correspondence}

The MNUDFT score and the DDR output share a structure in the query variable: 
query-dependent phases multiply signal-dependent coefficients. 
To make this correspondence explicit, order the
$K=|\mathcal K|$ evaluation frequencies and concatenate them as
$\boldsymbol f=\boldsymbol f_1\mathbin\Vert\cdots
\mathbin\Vert\boldsymbol f_K\in\mathbb R^{Kn}$.
For each pair $(\boldsymbol k,\boldsymbol d)$, define the lifted position
$\boldsymbol q_{\boldsymbol k,\boldsymbol d}
=\boldsymbol e_{\boldsymbol k}\otimes
\boldsymbol p_{\boldsymbol d}$,
where $\boldsymbol e_{\boldsymbol k}$ is the $k$-th standard basis vector, so that
\begin{equation}
\boldsymbol q_{\boldsymbol k,\boldsymbol d}\cdot\boldsymbol f
=
\boldsymbol p_{\boldsymbol d}\cdot
\boldsymbol f_{\boldsymbol k}.
\label{eq:lifted-phase-identity}
\end{equation}
Thus, every MNUDFT phase is a phase in the joint query
$\boldsymbol f$.

\begin{proposition}[Query Fourier structure of MNUDFT scores]
\label{prop:mnudft-structure}
Consider the query-dependent spectral target score
\[
s^\star(\boldsymbol x\mid\boldsymbol f)
=
G\!\left(
[X_{\boldsymbol k}(\boldsymbol x,\boldsymbol f_{\boldsymbol k})]
_{\boldsymbol k\in\mathcal K}
\right).
\]
If $G$ is a finite polynomial in the MNUDFT coefficients and their
complex conjugates, then
\begin{equation}
s^\star(\boldsymbol x\mid\boldsymbol f)
=
\sum_{\boldsymbol\omega\in\Omega_\star^{(f)}}
F^\star_{\boldsymbol\omega}(\boldsymbol x)
\exp(i\boldsymbol\omega\cdot\boldsymbol f)
\label{eq:score-fourier-f}
\end{equation}
with $\Omega_\star^{(f)}$ as the target query-frequency set and $F^\star_{\boldsymbol\omega}$ as the signal-dependent coefficient functions.
The query frequencies are sums and differences of the modes $2\pi\boldsymbol q_{\boldsymbol k,\boldsymbol d}$ appearing in the polynomial, while the coefficients
$F^\star_{\boldsymbol\omega}$ are polynomials in the signal samples and their conjugates. 
If the score is real valued for every
$\boldsymbol f$, then
$F^\star_{-\boldsymbol\omega}
=(F^\star_{\boldsymbol\omega})^*$.
If $G$ is continuous on a compact set containing the MNUDFT coefficient vectors under consideration, scores of the form given in  Eq.~\eqref{eq:score-fourier-f} approximate
$s^\star(\boldsymbol x\mid\boldsymbol f)$ uniformly over the corresponding signal and query domain.
\end{proposition}
The proof is given in Appx.~\ref{appx:mnudft-structure}.
For polynomial spectral scores, Proposition~\ref{prop:mnudft-structure} establishes the same separation of signal and query dependence as the DDR expansion in Eq.~\eqref{eq:joint-ddr-gtp}:
\begin{align*}
v(\boldsymbol x,\boldsymbol f;\boldsymbol\theta)
&=\sum_{\boldsymbol\omega\in\Omega^{(f)}}
F_{\boldsymbol\omega}(\boldsymbol x;\boldsymbol\theta)\,
\exp(i\boldsymbol\omega\cdot\boldsymbol f),\\
s^\star(\boldsymbol x\mid\boldsymbol f)
&=\sum_{\boldsymbol\nu\in\Omega_\star^{(f)}}
F^\star_{\boldsymbol\nu}(\boldsymbol x)\,
\exp(i\boldsymbol\nu\cdot\boldsymbol f).
\end{align*}
Both are GTPs in the query, with all signal dependence in the coefficient functions. 
The query encoding fixes the circuit modes $\Omega^{(f)}$; the sampling positions supply the task's modes, which $G$ combines into $\Omega_\star^{(f)}$.

The two expansions coincide for all signals and queries if
$\Omega_\star^{(f)}\subseteq\Omega^{(f)}$ and there is a parameter setting $\boldsymbol\theta$ for which
$F_{\boldsymbol\nu}(\cdot;\boldsymbol\theta)
=F^\star_{\boldsymbol\nu}$ for every target mode, while the coefficient functions of all unused circuit modes vanish identically.

\subsection{Phase alignment and coefficient error}

Exact reproduction provides a reference point for task alignment, but is not strictly required for a useful inductive bias.
For example, a linear model can generalize well from limited data when the target is approximately linear, despite being unable to represent it exactly.
Likewise, task-aligned query phases may improve learning without making the full target score exactly realizable.
For the polynomial scores of Proposition~\ref{prop:mnudft-structure}, we bound the score error at the fixed task query $\boldsymbol f^\star$ by separating phase mismatch from coefficient mismatch.
Suppose the proposed query encoding admits an injective assignment
$\pi:\Omega_\star^{(f)}\to\Omega^{(f)}$ of target modes to circuit modes, and let $r_{\boldsymbol\theta}(\boldsymbol x)$ collect the contributions of the unused circuit modes at $\boldsymbol f^\star$.

On a compact signal domain $\mathcal X$, write
$\|h\|_\infty=\sup_{\boldsymbol x\in\mathcal X}|h(\boldsymbol x)|$.
Then
\begin{equation}
\|s^\star(\cdot\mid\boldsymbol f^\star)
-v(\cdot,\boldsymbol f^\star;\boldsymbol\theta)\|_\infty
\leq
\varepsilon_{\mathrm{ph}}
+\varepsilon_{\mathrm{coef}}(\boldsymbol\theta),
\label{eq:alignment-error-decomposition}
\end{equation}
where
\begin{align}
\varepsilon_{\mathrm{ph}}
&=
\|\boldsymbol f^\star\|_2
\sum_{\boldsymbol\nu\in\Omega_\star^{(f)}}
\|F^\star_{\boldsymbol\nu}\|_\infty
\|\boldsymbol\nu-\pi(\boldsymbol\nu)\|_2,
\label{eq:phase-support-error}\\
\varepsilon_{\mathrm{coef}}(\boldsymbol\theta)
&=
\sum_{\boldsymbol\nu\in\Omega_\star^{(f)}}
\|F^\star_{\boldsymbol\nu}
-F_{\pi(\boldsymbol\nu)}(\cdot;\boldsymbol\theta)\|_\infty
+\|r_{\boldsymbol\theta}\|_\infty.
\label{eq:coefficient-error}
\end{align}
The first term measures phase mismatch; the second measures
coefficient mismatch and contributions from unused circuit modes.
The derivation is given in Appx.~\ref{appx:realizability}.
Controlling both terms is sufficient for score approximation,
but at a fixed query, the two expansions can agree without
matching term by term.

The appendix gives two complementary realizability results.
Proposition~\ref{thm:realizability} constructs a query
encoding that places each target mode within
$2^{-L_f}\|\boldsymbol R\|_2$ of a circuit mode.
Here, $L_f$ is the number of query-encoding layers and
$R_n=\max_{\boldsymbol\nu\in\Omega_\star^{(f)}}|\nu_n|$.
At sufficient depth, the assignments are injective and the
phase-error bound can be made arbitrarily small.
Figure~\ref{fig:alignment} illustrates this spectral approximation.
Proposition~\ref{prop:exact-intensity-contrast} constructs a circuit
that exactly realizes a phase-only MNUDFT intensity contrast,
including its coefficient functions.
For general targets and restricted architectures, whether the
coefficient error can be made small remains a separate question.

\subsection{Learning hypothesis}
\label{sec:task}

Our learning hypothesis is that a task-aligned query improves
learning from limited data relative to models with the same
trainable architecture and parameter count but mismatched or
absent queries.

For an MNUDFT-based task, the evaluation frequencies
$\boldsymbol f^\star$ specify the query value, while the task modes
$\Omega_\star^{(f)}$ guide the choice of query encoding
$\mathcal E_f^\star$.
The sampling positions determine the query modes; prior knowledge of the score structure, such as intensity dependence or a bound on its polynomial degree, determines which combinations guide the encoding.
Together, these choices define the DDR hypothesis class
$\mathcal H_{\mathcal E_f^\star}(\boldsymbol f^\star)$.
The query and its encoding are shared across samples and specified
independently of the labels: they supply task knowledge without
additional sample-dependent information.

Here, \emph{spectral inductive bias} denotes the effect of
task-prescribed query phases on the hypothesis class or the prior
over functions of the signal.
The query encoding makes these phases available, while training
must still learn signal-dependent coefficients for which
$v(\boldsymbol x,\boldsymbol f^\star;\boldsymbol\theta)
\approx s^\star(\boldsymbol x\mid\boldsymbol f^\star)$.
We next test the role of the query value numerically;
Sec.~\ref{sec:III} derives the query and its encoding from
measurement geometry.

\subsection{Teacher–student test of query alignment}
\label{sec:exp-synthetic}
We compare students that share their teacher's encoding and trainable architecture but differ in their fixed query.
Each teacher is a two-qubit DDR circuit evaluated at $\boldsymbol f^\star=(\pi/4,\pi/3)$, with local signal encodings, nonlocal query encoding gates generated by $Y_0Z_1$ and $Z_0Y_1$, product trainable blocks, and measurement of $Z_0Z_1$.
Inputs $\boldsymbol x$ are sampled uniformly from $[0,2\pi)^2$, and labels are obtained by thresholding the teacher expectation value.
We draw teacher parameters from the student initialization distribution and retain teachers satisfying prespecified class-balance and nonseparability criteria.
The latter ensures that the zero-query separable model cannot represent the continuous teacher score.

The circuits differ only in their fixed query: the correct value
$\boldsymbol f^\star$, the mismatched value
$2\boldsymbol f^\star$, or $\boldsymbol 0$.
The correct-query class contains the teacher by construction.
The mismatched arm retains the nonlocal query encoding gates but combines their
phases differently, while the zero-query circuit factorizes across
the two signal coordinates. 
The comparison therefore tests whether learning benefits from the teacher's query value, rather than from additional sample-dependent information or the mere presence of nonlocal gates.

Figure~\ref{fig:teacher-student} reports the learning curves using the test balanced accuracy, defined as the mean recall across classes.
The correct-query student reaches higher balanced accuracy with fewer
training examples than either control. The mismatched query does not
reproduce this improvement despite retaining the nonlocal query
branch. 
Panel~b reports fits to the continuous teacher score. 
For the correct arm, this is essentially an implementation check, since one start uses the teacher's parameters; for the controls, it reports the best-found approximation error.

\begin{figure}
    \centering
    \includegraphics[width=1\linewidth]{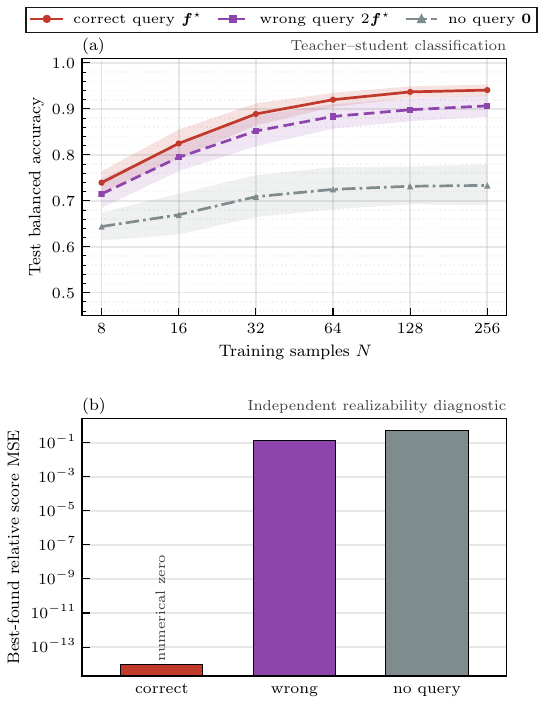}
    \caption{\textbf{A task-aligned fixed query improves learning in a teacher--student problem.}
    Teachers are DDR circuits evaluated at
    $\boldsymbol f^\star=(\pi/4,\pi/3)$; students use the correct query,
    the mismatched query $2\boldsymbol f^\star$, or no query.
    All students share the same architecture, trainable parameters,
    initializations, and optimization budget.
    \textbf{(a)} Test balanced accuracy versus training-set size.
    Lines show averages over ten independently sampled teachers after averaging folds within each teacher; shaded regions are paired $95\%$ bootstrap intervals over teachers.
    \textbf{(b)} Median best-found relative mean-square error when fitting
    the continuous teacher score on an independent grid.
    The correct-query model reaches numerical zero by construction, while
    the other models quantify the residual under the fixed optimization
    budget; the no-query class is separable by definition.
    }
    \label{fig:teacher-student}
\end{figure}

Realizability does not imply perfect accuracy from labels. 
The students see only $N_{\mathrm{train}}$ binary labels, which leave the decision boundary unconstrained between samples.
They also minimize a bounded-logit cross-entropy, whose minimizer within the class can trade a few points near the boundary for larger margins elsewhere and therefore need not coincide with the teacher.
The remaining errors concentrate near the teacher's decision boundary
(see Fig.~\ref{fig:decision-boundary} for a visualization).

\subsection{Sensitivity to query error}
\label{sec:query-sensitivity}

The three-arm comparison established a learning advantage for the
teacher query, but leaves open how precisely that query must be
specified. Using the same teacher ensemble and circuit architecture,
we fix the training-set size at $N_{\mathrm{train}}=256$ and retrain
each student independently for every candidate query. Data,
initializations, and optimization budgets are paired across query
values. 
We first scan $\boldsymbol f(\lambda)=\lambda\boldsymbol f^\star$, from the zero-query model at $\lambda=0$, through the teacher query at $\lambda=1$, to the mismatched query at $\lambda=2$.
This multiplicative mismatch preserves the relative pattern of query components and tests whether performance peaks at their task-prescribed scale.
Additionally, we additively perturb $\boldsymbol f^\star$ in eight directions to assess local sensitivity to the size and direction of additive query errors.
Details are provided in Appx.~\ref{appx:teacher-details}.

Figure~\ref{fig:query-sensitivity}a shows a sharp improvement as the query moves away from zero toward the task-aligned region. 
The matched query attains a mean balanced accuracy of $0.941$, compared with
$0.734$ at zero query and $0.906$ at $\lambda=2$.
The nominal maximum is $0.945$ at $\lambda=0.8125$, but performance forms a broad plateau around the teacher query.
Thus, $\boldsymbol f^\star$ is not a unique optimum in this finite-sample classification experiment. 
Nearby queries yield similar performance, while larger deviations along the scanned path weaken the advantage.

The local scan in Fig.~\ref{fig:query-sensitivity}b shows that this advantage persists under moderate additive query error. 
The mean, over teachers, of each teacher’s minimum accuracy across the eight sampled directions remains within one percentage point of its matched-query value up to $\delta_{1\%}=0.143$ rad, with a paired-bootstrap $95\%$ interval of $[0.106,0.233]$ rad.
At larger perturbations, degradation becomes increasingly dependent on direction. 
Together, the scans establish a region of useful queries around $\boldsymbol f^\star$ along the selected path and local directions; they do not establish an optimum over the full query domain.

\begin{figure}
    \centering
    \includegraphics[width=1\linewidth]{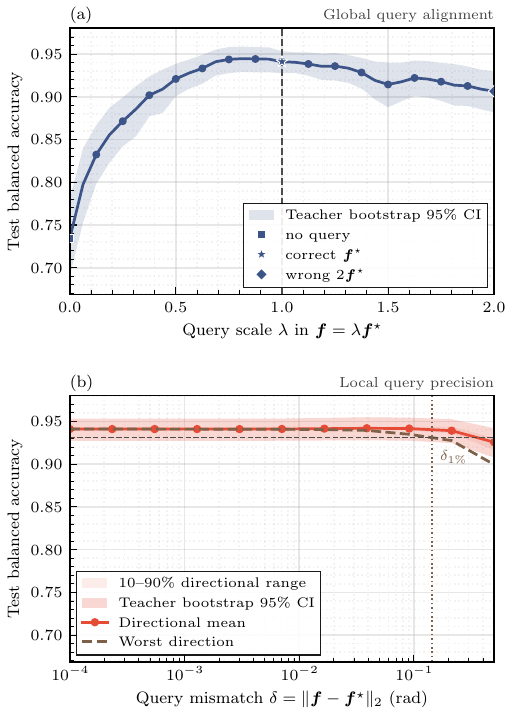}
    \caption{\textbf{Classification sensitivity to the fixed query value.}
    \textbf{(a)} Test balanced accuracy for students trained with
    $\boldsymbol f(\lambda)=\lambda\boldsymbol f^\star$, where
    $\boldsymbol f^\star=(\pi/4,\pi/3)$ is the query used to generate the teacher labels. 
    The values $\lambda=0$, $1$, and $2$ correspond to the no-query, correct-query, and mismatched-query models, respectively.
    \textbf{(b)} Local sensitivity to perturbations
    $\boldsymbol f=\boldsymbol f^\star+\delta(\cos\phi,\sin\phi)$ over eight sampled directions.
    The solid curve denotes the directional mean, the light band the directional
    $10$--$90\%$ range, and the dashed curve the minimum accuracy over directions, averaged over teachers.
    The marked tolerance $\delta_{1\%}$ estimates the largest query mismatch for which the worst-direction accuracy remains within one percentage point of its matched-query value. Parameters are retrained independently for every query;
    shaded confidence intervals are obtained by paired bootstrap resampling over
    teacher instances.}
    \label{fig:query-sensitivity}
\end{figure}

\section{Geometry-aligned classification of raw synthetic aperture radar-like data}
\label{sec:III}

\subsection{Acquisition geometry as a spectral query}
\label{sec:mappingSar}

The numerical study in Sec.~\ref{sec:II} used a teacher to define the aligned query and showed that nearby fixed queries can also support learning. 
In an acquisition task, we seek a query that is specified independently of the classifier and its labels. 
Synthetic aperture radar provides such a specification through the measurement geometry.

SAR forms spatially resolved images by combining complex echoes
recorded as the platform moves. A point scatterer contributes to many
measurement samples; its response becomes localized when those samples
are brought into phase according to the acquisition geometry
\cite{curlander1991synthetic,bamlerSyntheticApertureRadar1998,
cumming2005digital}. 
For scenes acquired under the same geometry, the raw signal echoes vary, but
the geometric phase relation used in focusing is fixed. We use the
echoes as the signal and derive the circuit query from that relation.

In standard $\omega$-$k$ processing, image formation combines
Fourier transformation of the raw measurements, matched filtering, a
geometry-dependent Stolt remapping, and an inverse transform. The
inverse Stolt relation specifies, for each point on the regular
image-frequency grid, where the Fourier transform of the measurements
must be evaluated.
We collect these generally nonuniform evaluation frequencies in
$\boldsymbol f^\star$.
Direct evaluation at these frequencies gives the idealized focusing operator
\begin{equation}
\mathcal R_{\mathrm{SAR}}
=
\mathcal F^{-1}_{\mathrm{img}}\,
\mathcal M\,\mathcal N_{\boldsymbol f^\star}.
\label{eq:sar-operator}
\end{equation}
Here $\mathcal N_{\boldsymbol f^\star}$ evaluates the MNUDFT of the
raw measurements at $\boldsymbol f^\star$, $\mathcal M$ applies the
corresponding fixed filtering and spectral weights, and
$\mathcal F^{-1}_{\mathrm{img}}$ transforms the result from the regular
image-frequency grid to a human-readable image.
Direct evaluation replaces the Stolt interpolation used to
approximate these values in conventional discrete implementations.
The idealized operator and its relation to interpolation are detailed
in Appx.~\ref{appx:sar-pipeline}.

\begin{proposition}[Spectral structure of SAR classification]
\label{cor:sar-spectral-task}
For a fixed acquisition, let $G_{\mathrm{SAR}}$ be a decision map
applied to the focused image. Its score as a function of the raw
measurements $\boldsymbol x$ is
\begin{equation}
\begin{aligned}
s_{\mathrm{SAR}}(\boldsymbol x)
&=
G_{\mathrm{SAR}}\!\left(
\mathcal R_{\mathrm{SAR}}[\boldsymbol x]\right)\\
&=
\widetilde G_{\mathrm{SAR}}\!\left(
\mathcal N_{\boldsymbol f^\star}[\boldsymbol x]\right),
\end{aligned}
\label{eq:sar-decision}
\end{equation}
where
$\widetilde G_{\mathrm{SAR}}
=G_{\mathrm{SAR}}\circ\mathcal F^{-1}_{\mathrm{img}}
\circ\mathcal M$.
Thus, classification after focusing has the spectral task structure
of Eq.~\eqref{eq:target-mnudft}. If $G_{\mathrm{SAR}}$ is polynomial
in the focused complex amplitudes and their conjugates, the raw-data
score has the exact nested Fourier structure of
Proposition~\ref{prop:mnudft-structure}.
\end{proposition}

Equation~\eqref{eq:sar-decision} connects the task-defined query to
classification before image formation. A raw-data classifier can use
geometry-dependent phase relations without explicitly forming every
image pixel, provided it can learn the combinations relevant to the
label. This is attractive when processing or data transfer is
limited~\cite{floresSyntheticApertureRadar2019,
cossuRealtimeSARbasedProcessing2009,
garciaAdvancementsOnboardProcessing2024}. It is also challenging:
information about a localized scene feature is spread across
measurements and can be obscured by noise and other scatterers. The
acquisition geometry supplies the fixed query, but neither the
unknown scene nor its label.

\subsection{A small-scale SAR-like localization task}
\label{sec:sar-experiment}

In $\omega$--$k$ focusing, the nonuniform range-frequency
evaluation points depend on azimuth frequency.
We isolate one local effect of this dependence that can be tested with a small quantum circuit: the coupling between azimuth and range phases in an unfocused spectral patch.
The scalar query $f^\star_{\mathrm{geom}}$ specifies the strength of this coupling and hence the phase correction that compensates it to leading order.
It summarizes one local component of the focusing relation; $\boldsymbol f^\star$ in the preceding subsection denotes the full set of MNUDFT evaluation frequencies.

Each example is an $8\times8$ patch of complex spectral samples
$\boldsymbol x=(x_{ad})$. 
Its coordinates are the azimuth wavenumber $k_x$, associated with measurements along the flight path, and the two-way range wavenumber $k_r$, associated with radar frequency. 
In the propagation geometry considered here, the remaining spatial wavenumber is
\begin{equation}
k_y(k_x,k_r)=\sqrt{k_r^2-k_x^2}.
\label{eq:sar-dispersion}
\end{equation}
A point target at reference slant range $R_0$ contributes a phase
containing $R_0k_y(k_x,k_r)$. 
Because this phase depends jointly on $k_x$ and $k_r$, independent corrections along the two coordinates do not remove all of its variation.

\begin{figure}[ht!]
\centering
\includegraphics[width=\linewidth]{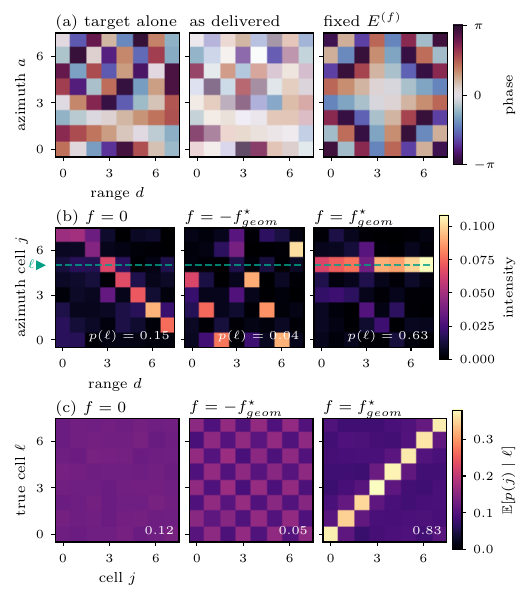}
\caption{\textbf{Geometry-derived phase correction of an unfocused SAR-like measurement.}
    \textbf{(a)} 
    Isolated target phase, delivered patch with clutter,
    noise, and channel phase, and the correction calculated from the acquisition. 
    Hue shows phase; brightness shows amplitude
    ($|x_{ad}|$ for the measured patch).
    \textbf{(b)} 
    Intensity by azimuth cell $j$ and range index $d$
    after the same Fourier readout with zero, wrong-sign, or
    geometry-derived query. The marker indicates the true cell $\ell$.
    The value $p(\ell)$ is the readout probability of the true cell for the scene shown.
    \textbf{(c)} 
    Conditional mean output probabilities $\mathbb{E}[p(j)\mid\ell]$ over independent scenes. The displayed balanced accuracies use the most probable output cell as the predicted label.
    This diagnostic uses the known
    channel calibration and geometry, with no fitted parameters.}
\label{fig:sar-data-structure}
\end{figure}

The patch samples 
$k_{x,a}=k_{x0}+\Delta k_x\widetilde a$ and
$k_{r,d}=k_{r0}+\Delta k_r\widetilde d$, where
$a,d\in\{0,\ldots,7\}$ and
$\widetilde a=a-7/2$, $\widetilde d=d-7/2$.
Here, $k_{x0}$ and $k_{r0}$ are the patch-center wavenumbers, and $\Delta k_x$ and $\Delta k_r$ are the sample spacings.
Expanding around $(k_{x0},k_{r0})$ separates the constant and
single-coordinate phases from the leading mixed term:
\begin{equation}
R_0k_y(k_{x,a},k_{r,d})
=
\mathrm{const}+\phi_a+\psi_d
+f^\star_{\mathrm{geom}}\widetilde a\widetilde d
+\mathcal R_{ad}.
\label{eq:sar-local-expansion}
\end{equation}
Here, $\mathrm{const}$ is a constant phase, $\phi_a$ and $\psi_d$ depend only on their respective coordinates, and  $\mathcal R_{ad}$ contains higher-order mixed terms.
The leading coupling is
\begin{equation}
\begin{aligned}
f^\star_{\mathrm{geom}}
&=
R_0\Delta k_x\Delta k_r
\left.
\frac{\partial^2 k_y}{\partial k_x\,\partial k_r}
\right|_{(k_{x0},k_{r0})}\\
&=
\frac{R_0\Delta k_x\Delta k_r\,k_{x0}k_{r0}}
{(k_{r0}^2-k_{x0}^2)^{3/2}}.
\end{aligned}
\label{eq:sar-geometry-query}
\end{equation}
The full derivation is given in Appx.~\ref{appx:derivation-geom}.
The acquisition fixes the center wavenumbers, sample spacings, and reference range. 
Thus $f^\star_{\mathrm{geom}}$ is calculated before any scene or label is used.
Because $f^\star_{\rm geom}\propto k_{x0}k_{r0}$, the leading coupling vanishes for a patch centered at zero azimuth wavenumber.
All acquisitions considered here therefore use a patch offset in azimuth wavenumber, so that $k_{x0}\neq0$.

The role of $f^\star_{\mathrm{geom}}$ can be seen by fixing the range index $d$. 
The mixed phase $f^\star_{\mathrm{geom}}\widetilde a\widetilde d$ then acts as an
azimuth phase ramp whose slope depends on $d$, shifting different
range slices differently under an azimuth Fourier readout.
A direct correction with query $f$ multiplies each sample by $\exp[-if\widetilde a\widetilde d]$, leaving the leading mixed phase
\[
\exp\!\left[
i(f^\star_{\mathrm{geom}}-f)\widetilde a\widetilde d
\right].
\]
Thus, $f=f^\star_{\mathrm{geom}}$ cancels the leading coupling,
$f=0$ leaves it unchanged, and $f=-f^\star_{\mathrm{geom}}$
doubles it. 
Figure~\ref{fig:sar-data-structure} illustrates the resulting effect on localization.

This correction has the spectral phase structure of
Sec.~\ref{sec:II}: the patch coordinates fix the pattern
$q_{ad}=\widetilde a\widetilde d$, while the acquisition determines
its scale $f^\star_{\mathrm{geom}}$.
These specify the query encoding and query value, respectively.
The query is shared across scenes; their target positions remain
encoded in the measured signal.

Each generated scene contains a point target at a continuous
position within one of eight azimuth cells. Its range position and
complex reflectivity vary between scenes. 
We also include four off-range scatterers and complex Gaussian measurement
noise at a nominal SNR of $10$\,dB.
Clutter introduces competing coherent returns, while noise perturbs
the measured phase pattern; together, they test whether the
geometry-derived query remains useful when  clutter and noise obscure the target response.
A fixed, label-independent channel phase introduces a shared calibration error.
The label is the target's azimuth cell, and the classifier receives the phase of the unfocused spectral patch. 
Scene generation uses the nonseparable part of $R_0k_y$ from the full dispersion relation, Eq.~\eqref{eq:sar-dispersion}, and represents the separable phases by random linear ramps; only the circuit query uses the leading mixed term of Eq.~\eqref{eq:sar-local-expansion}. 
The neglected higher-order mixed terms remain below $0.35$\,rad for every acquisition studied (Appx.~\ref{app:sar}), so the query does not exactly correct every phase in the generated data.

Figure~\ref{fig:sar-data-structure} illustrates that, before training, the geometry-derived correction concentrates the target response more strongly in its azimuth cell than the zero or wrong-sign correction under the same Fourier readout. 
Its balanced accuracy uses known channel calibration; the trained circuits receive uncalibrated phases and learn channel compensation and an effective readout from labels.
Acquisition parameters, scene generation, and training details are given in Appx.~\ref{appx:sar-pipeline}.

\subsection{Circuit with a geometry-derived query encoding gate}

The circuit uses two three-qubit registers, $A$ and $D$, for the azimuth and range indices. 
Across $L/2=2$ reuploading blocks, it encodes the measured phase and geometry query as
\begin{align}
E_\ell^{(x)}(\boldsymbol x)&=\operatorname{diag}_{a,d}\ ( \exp(i\,\arg(x_{ad})/2) ),\\ 
E^{(f)}_\ell(f)&=\exp (-i\,\frac{f}{2}\,\widetilde A\otimes\widetilde D),
\end{align}
where $\widetilde A=\operatorname{diag}_a(\widetilde a)$ and $\widetilde D=\operatorname{diag}_d(\widetilde d)$.
The generator $\widetilde A\otimes\widetilde D$ is diagonal with
eigenvalues $q_{ad}=\widetilde a\,\widetilde d$, so it is the position-diagonal encoding prescribed by the patch coordinates.
Each upload applies the phase $\exp(-i f q_{ad}/2)$ to sample $(a,d)$. 
At the geometry-derived query $f=f^\star_{\mathrm{geom}}$, each query upload applies half of the leading bilinear correction. 
When all trainable rotations are zero, the fixed CZ layers cancel, and the circuit before Fourier readout reduces to the phase-only signal encoding followed by the full correction.
Each upload decomposes into nine commuting cross-register $R_{ZZ}$ rotations.
The trainable blocks are $U_A\otimes U_D$: they can couple qubits within either register, but cannot implement an entangling correction between the registers. 
For the acquisitions studied here, the geometry-derived query couples the azimuth and range registers.
Because each trainable block acts separately on these registers, the query's effect cannot be removed by adjusting the parameters of a neighboring block.
A Fourier transform on $A$, followed by marginalization over $D$, gives probabilities for the eight azimuth cells. 
At each fixed query the circuit has 48 trainable parameters. 
We can therefore compare query values while holding the trainable architecture and optimization resources fixed.

\begin{figure}[t!]
    \centering
    \includegraphics[width=0.9\linewidth]{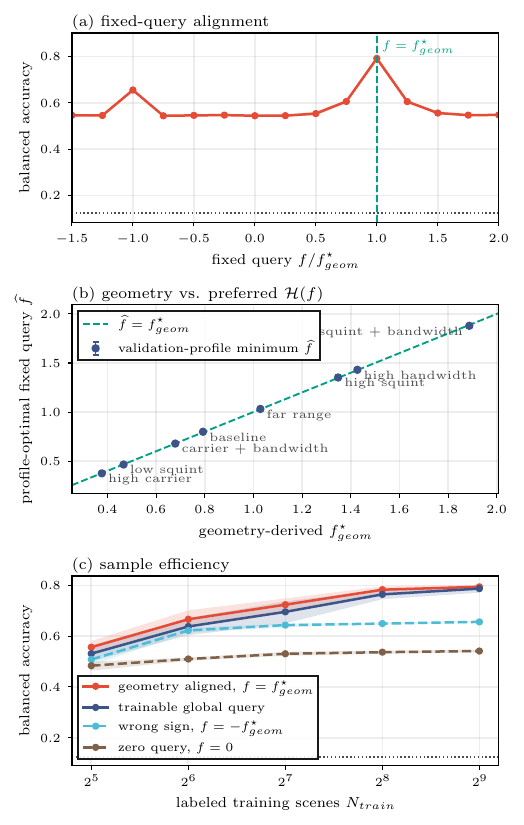}
    \caption{\textbf{Geometry-derived queries improve SAR-like spectral classification.}
    \textbf{(a)} Test balanced accuracy as the fixed query is varied relative to $f^\star_{\mathrm{geom}}$. The variational parameters are retrained independently at every query value using paired datasets and initializations; the vertical dashed line marks the prediction from the acquisition geometry.
    \textbf{(b)} Profile-optimal fixed query $\widehat f$ versus $f^\star_{\mathrm{geom}}$ for acquisitions with different carrier frequencies, squint angles, bandwidths, and reference ranges. For each acquisition, the variational parameters are retrained over a grid of fixed queries and $\widehat f$ minimizes the mean held-out validation cross-entropy. 
    Error bars are $95\%$ intervals obtained by bootstrapping the complete validation-loss profiles and re-estimating the preferred query. 
    They quantify variability across scene draws conditional on the local quadratic interpolation and do not include interpolation bias.
    The dashed line denotes $\widehat f=f^\star_{\mathrm{geom}}$.
    \textbf{(c)} Test balanced accuracy versus the number $N_{\rm train}$ of labeled training scenes for the geometry-aligned query $f=f^\star_{\mathrm{geom}}$, a trainable global query, the equally entangling wrong-sign query $f=-f^\star_{\mathrm{geom}}$, and the zero-query circuit with $f=0$. The fixed-query models share the same phase upload, 48-parameter variational circuit, optimization budget, and paired initializations; the trainable-query model contains one additional global scalar. Lines and shaded regions in \textbf{(a)} and \textbf{(c)} show means and $95\%$ bootstrap intervals over independent paired scene draws. The horizontal dotted line denotes chance accuracy, $1/8$.
    \textbf{(a,b)} use $N_{\rm train}=2^9$.
    }
    \label{fig:sar-geometry-classification}
\end{figure}

\subsection{Results}

We first vary $f$ for one acquisition, keeping the sampled scenes fixed and retraining the circuit at every query value under the same optimization budget. 
Figure~\ref{fig:sar-geometry-classification}a shows that classification performance peaks near the independently
predicted $f^\star_{\mathrm{geom}}$. The wrong-sign query performs substantially worse despite having the same gate count and operator-Schmidt spectrum. 
The benefit therefore depends on the direction of the geometric correction, not only on the presence of an entangling query encoding gate.

We next vary carrier frequency, squint, bandwidth, and reference range~(Table~\ref{tab:sar_acq}).
Across these acquisitions, $f^\star_{\rm geom}$ ranges from $0.38$ to $1.89$.
For each acquisition we retrain the circuit at fixed queries $f/f^\star_{\rm geom}\in\{-1.5,-1.25,\dots,2\}$
and find the query value $\widehat f$ that minimizes the validation loss.
For every acquisition, the minimum validation loss on the sampled grid occurs at the geometry-derived query $f=f^\star_{\rm geom}$.
Local quadratic interpolation through the minimum-loss grid point and its two neighbors (Appendix C4) gives estimates close to the geometry-derived prediction (Fig.~\ref{fig:sar-geometry-classification}b).
Because the scenes are generated from the same dispersion relation, this agreement is a consistency check rather than independent physical evidence. 
It shows that the leading-order coupling of Eq.~\eqref{eq:sar-geometry-query}, despite neglecting higher-order mixed terms, is the query the trained circuit itself prefers under clutter, noise, and an uncalibrated channel phase, and that this preference follows the geometry across a range of $f^\star_{\mathrm{geom}}$. 
Supplying $f^\star_{\mathrm{geom}}$, which uses no labels or fitted scene parameters, therefore fixes in advance a value that would otherwise have to be found from labels, which is probed in panel~c.

Finally, Fig.~\ref{fig:sar-geometry-classification}c tests whether supplying that prediction helps when labels are limited. 
The geometry-aligned circuit outperforms both the equally entangling wrong-sign control and the zero-query control across the studied training-set sizes. 
It also modestly outperforms a circuit that learns one global query from labels at several small and intermediate sizes; 
at the largest size, the paired 95\% bootstrap interval of the accuracy difference includes zero. The acquisition-derived query is most useful where training data offer the least opportunity to infer its effect.

This experiment retains the coupled $\omega$-$k$ phase and includes off-range scatterers, noise, and channel phase, but uses small synthetic point-scatterer patches. 
It omits the transmitted pulse, antenna response, platform errors, extended-scene speckle, and complete focusing pipeline, and assumes ideal phase upload and noiseless circuit evaluation.
Within this small-scale setting, acquisition geometry supplies a query that cannot be absorbed into an adjacent trainable block before training, and that choice improves localization from limited labeled scenes.

\section{Discussion and Outlook}
\label{sec:conclusion-pqc}

We identify coefficient modulation by a fixed query as a mechanism for task-aligned inductive bias in parameterized quantum circuits.
The query changes how existing signal Fourier modes are combined, allowing task-prescribed phase relations to shape the hypothesis class or function prior without adding sample-dependent information.
This connects spectral task structure to circuit design: 
evaluation frequencies specify the query, and sampling geometry guides its encoding.
Small-scale spectral and SAR-like experiments demonstrate that aligned queries can improve learning from limited data, while their practical benefits at larger scales remain open.

\emph{Query-controlled inductive bias}---
A query-induced change in an individual circuit output does not by itself establish a change in inductive bias.
The coefficient-matrix rank bounds modulation at fixed parameters, while absorption provides sufficient conditions for the hypothesis class and function prior to remain invariant.
Two-qubit circuits distinguish all three cases: 
no model-level change, prior reweighting, and hypothesis-class change.
Rank alone does not predict learning benefits, and failure of the absorption test does not establish a class change.
These results motivate circuit design from required coefficient relations: 
choose encodings for the task modes, place query operations where their effects survive reparameterization, and select trainable blocks that make the required coefficient functions accessible.
The remaining challenge is to characterize which functions such architectures can represent and efficiently learn, and when additional reuploading layers are needed.

\emph{Aligning the query with spectral tasks}---
Nonuniform Fourier evaluations specify a query value and phase modes that guide its encoding.
The alignment bound separates phase mismatch from coefficient error;  
an explicit intensity-contrast construction makes both vanish for a nonlinear, phase-only target.
The teacher--student experiment shows that an aligned fixed query can improve learning and tolerate moderate query error.
However, the teacher belongs to the circuit family by construction, and the exact realization covers a specific target.
Extending these results to general MNUDFT-based classifiers requires identifying which signal-dependent coefficient functions are accessible with compact task-derived encodings and how their structure affects learning from finite data.

\emph{Geometry-aligned classification of raw SAR-like data}---
SAR acquisition geometry supplies a query independently of training data and labels.
The leading azimuth-range coupling predicts useful query values across acquisitions, and supplying this prediction improves localization from limited labeled scenes.
The experiment uses small synthetic, phase-only patches and a scalar local correction, with ideal input encoding and noiseless circuit evaluation.
The next step is to test richer geometry-derived queries on larger, experimentally acquired phase histories.
Comparisons with geometry-aware classical classifiers and conventional image-based pipelines should match the available information and account for computational resources, including measurement loading and circuit-output estimation.

More broadly, many machine learning tasks based on measurement data prescribe in advance how samples must be phased and combined.
Such prescriptions can enter a parameterized quantum circuit as a fixed query, providing a way to build task-aligned inductive biases.
The absorption tests and diagnostics developed here can assess whether this choice changes the hypothesis class or function prior before any training data are seen.

\section*{Data and code availability}
The code and data generated and analyzed in this study are available in the accompanying GitHub repository at \url{https://github.com/felixfrohnertdb/sar-qml}.



\section*{Competing interests}
S.L. and M.S. are employees of Airbus Defence and Space GmbH. 
The remaining authors declare no financial or non-financial competing interests.

\begin{acknowledgments}
The authors thank Riccardo Molteni for useful discussions.
This work was funded by the European Union under Grant Agreement
No.~101080142 (EQUALITY). 
A.P.-S. acknowledges the Swiss National Science Foundation for its support through the grant TMPFP2\_234085. 
V.D. acknowledges support from the European Union’s Horizon Europe program through the ERC CoG BeMAIQuantum (Grant No. 101124342). 
This work was supported by the Dutch National Growth Fund (NGF) as part of the QDNL programme.
\end{acknowledgments}


\newpage
\bibliography{references}

\begin{thebibliography}{65}%
\makeatletter
\providecommand \@ifxundefined [1]{%
 \@ifx{#1\undefined}
}%
\providecommand \@ifnum [1]{%
 \ifnum #1\expandafter \@firstoftwo
 \else \expandafter \@secondoftwo
 \fi
}%
\providecommand \@ifx [1]{%
 \ifx #1\expandafter \@firstoftwo
 \else \expandafter \@secondoftwo
 \fi
}%
\providecommand \natexlab [1]{#1}%
\providecommand \enquote  [1]{``#1''}%
\providecommand \bibnamefont  [1]{#1}%
\providecommand \bibfnamefont [1]{#1}%
\providecommand \citenamefont [1]{#1}%
\providecommand \href@noop [0]{\@secondoftwo}%
\providecommand \href [0]{\begingroup \@sanitize@url \@href}%
\providecommand \@href[1]{\@@startlink{#1}\@@href}%
\providecommand \@@href[1]{\endgroup#1\@@endlink}%
\providecommand \@sanitize@url [0]{\catcode `\\12\catcode `\$12\catcode `\&12\catcode `\#12\catcode `\^12\catcode `\_12\catcode `\%12\relax}%
\providecommand \@@startlink[1]{}%
\providecommand \@@endlink[0]{}%
\providecommand \url  [0]{\begingroup\@sanitize@url \@url }%
\providecommand \@url [1]{\endgroup\@href {#1}{\urlprefix }}%
\providecommand \urlprefix  [0]{URL }%
\providecommand \Eprint [0]{\href }%
\providecommand \doibase [0]{https://doi.org/}%
\providecommand \selectlanguage [0]{\@gobble}%
\providecommand \bibinfo  [0]{\@secondoftwo}%
\providecommand \bibfield  [0]{\@secondoftwo}%
\providecommand \translation [1]{[#1]}%
\providecommand \BibitemOpen [0]{}%
\providecommand \bibitemStop [0]{}%
\providecommand \bibitemNoStop [0]{.\EOS\space}%
\providecommand \EOS [0]{\spacefactor3000\relax}%
\providecommand \BibitemShut  [1]{\csname bibitem#1\endcsname}%
\let\auto@bib@innerbib\@empty
\bibitem [{\citenamefont {Biamonte}\ \emph {et~al.}(2017)\citenamefont {Biamonte}, \citenamefont {Wittek}, \citenamefont {Pancotti}, \citenamefont {Rebentrost}, \citenamefont {Wiebe},\ and\ \citenamefont {Lloyd}}]{biamonte_quantum_2017}%
  \BibitemOpen
  \bibfield  {author} {\bibinfo {author} {\bibfnamefont {J.}~\bibnamefont {Biamonte}}, \bibinfo {author} {\bibfnamefont {P.}~\bibnamefont {Wittek}}, \bibinfo {author} {\bibfnamefont {N.}~\bibnamefont {Pancotti}}, \bibinfo {author} {\bibfnamefont {P.}~\bibnamefont {Rebentrost}}, \bibinfo {author} {\bibfnamefont {N.}~\bibnamefont {Wiebe}},\ and\ \bibinfo {author} {\bibfnamefont {S.}~\bibnamefont {Lloyd}},\ }\bibfield  {title} {\bibinfo {title} {Quantum machine learning},\ }\href@noop {} {\bibfield  {journal} {\bibinfo  {journal} {Nature}\ }\textbf {\bibinfo {volume} {549}},\ \bibinfo {pages} {195} (\bibinfo {year} {2017})}\BibitemShut {NoStop}%
\bibitem [{\citenamefont {Benedetti}\ \emph {et~al.}(2019)\citenamefont {Benedetti}, \citenamefont {Lloyd}, \citenamefont {Sack},\ and\ \citenamefont {Fiorentini}}]{benedettiParameterizedQuantumCircuits2019}%
  \BibitemOpen
  \bibfield  {author} {\bibinfo {author} {\bibfnamefont {M.}~\bibnamefont {Benedetti}}, \bibinfo {author} {\bibfnamefont {E.}~\bibnamefont {Lloyd}}, \bibinfo {author} {\bibfnamefont {S.}~\bibnamefont {Sack}},\ and\ \bibinfo {author} {\bibfnamefont {M.}~\bibnamefont {Fiorentini}},\ }\bibfield  {title} {\bibinfo {title} {Parameterized quantum circuits as machine learning models},\ }\href@noop {} {\bibfield  {journal} {\bibinfo  {journal} {Quantum Science and Technology}\ }\textbf {\bibinfo {volume} {4}},\ \bibinfo {pages} {043001} (\bibinfo {year} {2019})}\BibitemShut {NoStop}%
\bibitem [{\citenamefont {Cerezo}\ \emph {et~al.}(2021)\citenamefont {Cerezo}, \citenamefont {Arrasmith}, \citenamefont {Babbush}, \citenamefont {Benjamin}, \citenamefont {Endo}, \citenamefont {Fujii}, \citenamefont {McClean}, \citenamefont {Mitarai}, \citenamefont {Yuan}, \citenamefont {Cincio},\ and\ \citenamefont {Coles}}]{cerezoVariational2021}%
  \BibitemOpen
  \bibfield  {author} {\bibinfo {author} {\bibfnamefont {M.}~\bibnamefont {Cerezo}}, \bibinfo {author} {\bibfnamefont {A.}~\bibnamefont {Arrasmith}}, \bibinfo {author} {\bibfnamefont {R.}~\bibnamefont {Babbush}}, \bibinfo {author} {\bibfnamefont {S.~C.}\ \bibnamefont {Benjamin}}, \bibinfo {author} {\bibfnamefont {S.}~\bibnamefont {Endo}}, \bibinfo {author} {\bibfnamefont {K.}~\bibnamefont {Fujii}}, \bibinfo {author} {\bibfnamefont {J.~R.}\ \bibnamefont {McClean}}, \bibinfo {author} {\bibfnamefont {K.}~\bibnamefont {Mitarai}}, \bibinfo {author} {\bibfnamefont {X.}~\bibnamefont {Yuan}}, \bibinfo {author} {\bibfnamefont {L.}~\bibnamefont {Cincio}},\ and\ \bibinfo {author} {\bibfnamefont {P.~J.}\ \bibnamefont {Coles}},\ }\bibfield  {title} {\bibinfo {title} {Variational quantum algorithms},\ }\href {https://doi.org/10.1038/s42254-021-00348-9} {\bibfield  {journal} {\bibinfo  {journal} {Nature Reviews Physics}\ }\textbf {\bibinfo {volume} {3}},\ \bibinfo {pages} {625} (\bibinfo {year} {2021})}\BibitemShut
  {NoStop}%
\bibitem [{\citenamefont {Preskill}(2018)}]{preskill_quantum_2018}%
  \BibitemOpen
  \bibfield  {author} {\bibinfo {author} {\bibfnamefont {J.}~\bibnamefont {Preskill}},\ }\bibfield  {title} {\bibinfo {title} {Quantum {Computing} in the {NISQ} era and beyond},\ }\href {https://doi.org/10.22331/q-2018-08-06-79} {\bibfield  {journal} {\bibinfo  {journal} {Quantum}\ }\textbf {\bibinfo {volume} {2}},\ \bibinfo {pages} {79} (\bibinfo {year} {2018})}\BibitemShut {NoStop}%
\bibitem [{\citenamefont {Bharti}\ \emph {et~al.}(2022)\citenamefont {Bharti}, \citenamefont {Cervera-Lierta}, \citenamefont {Kyaw}, \citenamefont {Haug}, \citenamefont {Alperin-Lea}, \citenamefont {Anand}, \citenamefont {Degroote}, \citenamefont {Heimonen}, \citenamefont {Kottmann}, \citenamefont {Menke}, \citenamefont {Mok}, \citenamefont {Sim}, \citenamefont {Kwek},\ and\ \citenamefont {Aspuru-Guzik}}]{bharti2022}%
  \BibitemOpen
  \bibfield  {author} {\bibinfo {author} {\bibfnamefont {K.}~\bibnamefont {Bharti}}, \bibinfo {author} {\bibfnamefont {A.}~\bibnamefont {Cervera-Lierta}}, \bibinfo {author} {\bibfnamefont {T.~H.}\ \bibnamefont {Kyaw}}, \bibinfo {author} {\bibfnamefont {T.}~\bibnamefont {Haug}}, \bibinfo {author} {\bibfnamefont {S.}~\bibnamefont {Alperin-Lea}}, \bibinfo {author} {\bibfnamefont {A.}~\bibnamefont {Anand}}, \bibinfo {author} {\bibfnamefont {M.}~\bibnamefont {Degroote}}, \bibinfo {author} {\bibfnamefont {H.}~\bibnamefont {Heimonen}}, \bibinfo {author} {\bibfnamefont {J.~S.}\ \bibnamefont {Kottmann}}, \bibinfo {author} {\bibfnamefont {T.}~\bibnamefont {Menke}}, \bibinfo {author} {\bibfnamefont {W.-K.}\ \bibnamefont {Mok}}, \bibinfo {author} {\bibfnamefont {S.}~\bibnamefont {Sim}}, \bibinfo {author} {\bibfnamefont {L.-C.}\ \bibnamefont {Kwek}},\ and\ \bibinfo {author} {\bibfnamefont {A.}~\bibnamefont {Aspuru-Guzik}},\ }\bibfield  {title} {\bibinfo {title} {Noisy intermediate-scale quantum algorithms},\ }\href@noop
  {} {\bibfield  {journal} {\bibinfo  {journal} {Reviews of Modern Physics}\ }\textbf {\bibinfo {volume} {94}},\ \bibinfo {pages} {015004} (\bibinfo {year} {2022})}\BibitemShut {NoStop}%
\bibitem [{\citenamefont {Cerezo}\ \emph {et~al.}(2022)\citenamefont {Cerezo}, \citenamefont {Verdon}, \citenamefont {Huang}, \citenamefont {Cincio},\ and\ \citenamefont {Coles}}]{cerezo2022challenges}%
  \BibitemOpen
  \bibfield  {author} {\bibinfo {author} {\bibfnamefont {M.}~\bibnamefont {Cerezo}}, \bibinfo {author} {\bibfnamefont {G.}~\bibnamefont {Verdon}}, \bibinfo {author} {\bibfnamefont {H.-Y.}\ \bibnamefont {Huang}}, \bibinfo {author} {\bibfnamefont {L.}~\bibnamefont {Cincio}},\ and\ \bibinfo {author} {\bibfnamefont {P.~J.}\ \bibnamefont {Coles}},\ }\bibfield  {title} {\bibinfo {title} {Challenges and opportunities in quantum machine learning},\ }\href@noop {} {\bibfield  {journal} {\bibinfo  {journal} {Nature Computational Science}\ }\textbf {\bibinfo {volume} {2}},\ \bibinfo {pages} {567} (\bibinfo {year} {2022})}\BibitemShut {NoStop}%
\bibitem [{\citenamefont {Bowles}\ \emph {et~al.}(2024)\citenamefont {Bowles}, \citenamefont {Ahmed},\ and\ \citenamefont {Schuld}}]{bowles2024better}%
  \BibitemOpen
  \bibfield  {author} {\bibinfo {author} {\bibfnamefont {J.}~\bibnamefont {Bowles}}, \bibinfo {author} {\bibfnamefont {S.}~\bibnamefont {Ahmed}},\ and\ \bibinfo {author} {\bibfnamefont {M.}~\bibnamefont {Schuld}},\ }\href@noop {} {\bibinfo {title} {Better than classical? {T}he subtle art of benchmarking quantum machine learning models}} (\bibinfo {year} {2024}),\ \Eprint {https://arxiv.org/abs/2403.07059} {arXiv:2403.07059 [quant-ph]} \BibitemShut {NoStop}%
\bibitem [{\citenamefont {Huang}\ \emph {et~al.}(2021)\citenamefont {Huang}, \citenamefont {Broughton}, \citenamefont {Mohseni}, \citenamefont {Babbush}, \citenamefont {Boixo}, \citenamefont {Neven},\ and\ \citenamefont {McClean}}]{huang2021power}%
  \BibitemOpen
  \bibfield  {author} {\bibinfo {author} {\bibfnamefont {H.-Y.}\ \bibnamefont {Huang}}, \bibinfo {author} {\bibfnamefont {M.}~\bibnamefont {Broughton}}, \bibinfo {author} {\bibfnamefont {M.}~\bibnamefont {Mohseni}}, \bibinfo {author} {\bibfnamefont {R.}~\bibnamefont {Babbush}}, \bibinfo {author} {\bibfnamefont {S.}~\bibnamefont {Boixo}}, \bibinfo {author} {\bibfnamefont {H.}~\bibnamefont {Neven}},\ and\ \bibinfo {author} {\bibfnamefont {J.~R.}\ \bibnamefont {McClean}},\ }\bibfield  {title} {\bibinfo {title} {Power of data in quantum machine learning},\ }\href {https://doi.org/10.1038/s41467-021-22539-9} {\bibfield  {journal} {\bibinfo  {journal} {Nature Communications}\ }\textbf {\bibinfo {volume} {12}},\ \bibinfo {pages} {2631} (\bibinfo {year} {2021})}\BibitemShut {NoStop}%
\bibitem [{\citenamefont {K{\"u}bler}\ \emph {et~al.}(2021)\citenamefont {K{\"u}bler}, \citenamefont {Buchholz},\ and\ \citenamefont {Sch{\"o}lkopf}}]{kublerInductiveBiasQuantum2021}%
  \BibitemOpen
  \bibfield  {author} {\bibinfo {author} {\bibfnamefont {J.}~\bibnamefont {K{\"u}bler}}, \bibinfo {author} {\bibfnamefont {S.}~\bibnamefont {Buchholz}},\ and\ \bibinfo {author} {\bibfnamefont {B.}~\bibnamefont {Sch{\"o}lkopf}},\ }\bibfield  {title} {\bibinfo {title} {The inductive bias of quantum kernels},\ }\href@noop {} {\bibfield  {journal} {\bibinfo  {journal} {Advances in neural information processing systems}\ }\textbf {\bibinfo {volume} {34}},\ \bibinfo {pages} {12661} (\bibinfo {year} {2021})}\BibitemShut {NoStop}%
\bibitem [{\citenamefont {Schuld}\ and\ \citenamefont {Killoran}(2022)}]{SchuldAdvantage2022}%
  \BibitemOpen
  \bibfield  {author} {\bibinfo {author} {\bibfnamefont {M.}~\bibnamefont {Schuld}}\ and\ \bibinfo {author} {\bibfnamefont {N.}~\bibnamefont {Killoran}},\ }\bibfield  {title} {\bibinfo {title} {Is quantum advantage the right goal for quantum machine learning?},\ }\href@noop {} {\bibfield  {journal} {\bibinfo  {journal} {PRX Quantum}\ }\textbf {\bibinfo {volume} {3}},\ \bibinfo {pages} {030101} (\bibinfo {year} {2022})}\BibitemShut {NoStop}%
\bibitem [{\citenamefont {Bowles}\ \emph {et~al.}(2023)\citenamefont {Bowles}, \citenamefont {Wright}, \citenamefont {Farkas}, \citenamefont {Killoran},\ and\ \citenamefont {Schuld}}]{bowles_contextuality_2023}%
  \BibitemOpen
  \bibfield  {author} {\bibinfo {author} {\bibfnamefont {J.}~\bibnamefont {Bowles}}, \bibinfo {author} {\bibfnamefont {V.~J.}\ \bibnamefont {Wright}}, \bibinfo {author} {\bibfnamefont {M.}~\bibnamefont {Farkas}}, \bibinfo {author} {\bibfnamefont {N.}~\bibnamefont {Killoran}},\ and\ \bibinfo {author} {\bibfnamefont {M.}~\bibnamefont {Schuld}},\ }\href@noop {} {\bibinfo {title} {Contextuality and inductive bias in quantum machine learning}} (\bibinfo {year} {2023}),\ \bibinfo {note} {arXiv:2302.01365 [quant-ph]}\BibitemShut {NoStop}%
\bibitem [{\citenamefont {Gili}\ \emph {et~al.}(2024)\citenamefont {Gili}, \citenamefont {Alonso},\ and\ \citenamefont {Schuld}}]{gili2024inductive}%
  \BibitemOpen
  \bibfield  {author} {\bibinfo {author} {\bibfnamefont {K.}~\bibnamefont {Gili}}, \bibinfo {author} {\bibfnamefont {G.}~\bibnamefont {Alonso}},\ and\ \bibinfo {author} {\bibfnamefont {M.}~\bibnamefont {Schuld}},\ }\bibfield  {title} {\bibinfo {title} {An inductive bias from quantum mechanics: learning order effects with non-commuting measurements},\ }\href@noop {} {\bibfield  {journal} {\bibinfo  {journal} {Quantum Machine Intelligence}\ }\textbf {\bibinfo {volume} {6}},\ \bibinfo {pages} {67} (\bibinfo {year} {2024})}\BibitemShut {NoStop}%
\bibitem [{\citenamefont {Mitchell}(1980)}]{mitchell1980}%
  \BibitemOpen
  \bibfield  {author} {\bibinfo {author} {\bibfnamefont {T.~M.}\ \bibnamefont {Mitchell}},\ }\href@noop {} {\emph {\bibinfo {title} {The need for biases in learning generalizations}}},\ \bibinfo {type} {Tech. Rep.}\ \bibinfo {number} {CBM-TR-117}\ (\bibinfo  {institution} {Rutgers University},\ \bibinfo {year} {1980})\BibitemShut {NoStop}%
\bibitem [{\citenamefont {Battaglia}\ \emph {et~al.}(2018)\citenamefont {Battaglia}, \citenamefont {Hamrick}, \citenamefont {Bapst} \emph {et~al.}}]{battaglia2018}%
  \BibitemOpen
  \bibfield  {author} {\bibinfo {author} {\bibfnamefont {P.~W.}\ \bibnamefont {Battaglia}}, \bibinfo {author} {\bibfnamefont {J.~B.}\ \bibnamefont {Hamrick}}, \bibinfo {author} {\bibfnamefont {V.}~\bibnamefont {Bapst}}, \emph {et~al.},\ }\href@noop {} {\bibinfo {title} {Relational inductive biases, deep learning, and graph networks}} (\bibinfo {year} {2018}),\ \Eprint {https://arxiv.org/abs/1806.01261} {arXiv:1806.01261 [cs.LG]} \BibitemShut {NoStop}%
\bibitem [{\citenamefont {Lecun}\ \emph {et~al.}(1998)\citenamefont {Lecun}, \citenamefont {Bottou}, \citenamefont {Bengio},\ and\ \citenamefont {Haffner}}]{lecun_gradient-based_1998}%
  \BibitemOpen
  \bibfield  {author} {\bibinfo {author} {\bibfnamefont {Y.}~\bibnamefont {Lecun}}, \bibinfo {author} {\bibfnamefont {L.}~\bibnamefont {Bottou}}, \bibinfo {author} {\bibfnamefont {Y.}~\bibnamefont {Bengio}},\ and\ \bibinfo {author} {\bibfnamefont {P.}~\bibnamefont {Haffner}},\ }\bibfield  {title} {\bibinfo {title} {Gradient-based learning applied to document recognition},\ }\href {https://doi.org/10.1109/5.726791} {\bibfield  {journal} {\bibinfo  {journal} {Proceedings of the IEEE}\ }\textbf {\bibinfo {volume} {86}},\ \bibinfo {pages} {2278} (\bibinfo {year} {1998})}\BibitemShut {NoStop}%
\bibitem [{\citenamefont {Cohen}\ and\ \citenamefont {Welling}(2016)}]{cohen2016}%
  \BibitemOpen
  \bibfield  {author} {\bibinfo {author} {\bibfnamefont {T.}~\bibnamefont {Cohen}}\ and\ \bibinfo {author} {\bibfnamefont {M.}~\bibnamefont {Welling}},\ }\bibfield  {title} {\bibinfo {title} {Group equivariant convolutional networks},\ }in\ \href@noop {} {\emph {\bibinfo {booktitle} {Proceedings of the 33rd International Conference on Machine Learning}}},\ \bibinfo {series} {Proceedings of Machine Learning Research}, Vol.~\bibinfo {volume} {48}\ (\bibinfo {year} {2016})\ pp.\ \bibinfo {pages} {2990--2999}\BibitemShut {NoStop}%
\bibitem [{\citenamefont {Zeiler}\ and\ \citenamefont {Fergus}(2014)}]{zeiler_visualizing_2013}%
  \BibitemOpen
  \bibfield  {author} {\bibinfo {author} {\bibfnamefont {M.~D.}\ \bibnamefont {Zeiler}}\ and\ \bibinfo {author} {\bibfnamefont {R.}~\bibnamefont {Fergus}},\ }\bibfield  {title} {\bibinfo {title} {Visualizing and understanding convolutional networks},\ }in\ \href@noop {} {\emph {\bibinfo {booktitle} {European conference on computer vision}}}\ (\bibinfo {year} {2014})\ pp.\ \bibinfo {pages} {818--833}\BibitemShut {NoStop}%
\bibitem [{\citenamefont {Vaswani}\ \emph {et~al.}(2017)\citenamefont {Vaswani}, \citenamefont {Shazeer}, \citenamefont {Parmar}, \citenamefont {Uszkoreit}, \citenamefont {Jones}, \citenamefont {Gomez}, \citenamefont {Kaiser},\ and\ \citenamefont {Polosukhin}}]{vaswani_attention_2023}%
  \BibitemOpen
  \bibfield  {author} {\bibinfo {author} {\bibfnamefont {A.}~\bibnamefont {Vaswani}}, \bibinfo {author} {\bibfnamefont {N.}~\bibnamefont {Shazeer}}, \bibinfo {author} {\bibfnamefont {N.}~\bibnamefont {Parmar}}, \bibinfo {author} {\bibfnamefont {J.}~\bibnamefont {Uszkoreit}}, \bibinfo {author} {\bibfnamefont {L.}~\bibnamefont {Jones}}, \bibinfo {author} {\bibfnamefont {A.~N.}\ \bibnamefont {Gomez}}, \bibinfo {author} {\bibfnamefont {{\L}.}~\bibnamefont {Kaiser}},\ and\ \bibinfo {author} {\bibfnamefont {I.}~\bibnamefont {Polosukhin}},\ }\bibfield  {title} {\bibinfo {title} {Attention is all you need},\ }\href@noop {} {\bibfield  {journal} {\bibinfo  {journal} {Advances in neural information processing systems}\ }\textbf {\bibinfo {volume} {30}} (\bibinfo {year} {2017})}\BibitemShut {NoStop}%
\bibitem [{\citenamefont {Maur{\'i}cio}\ \emph {et~al.}(2023)\citenamefont {Maur{\'i}cio}, \citenamefont {Domingues},\ and\ \citenamefont {Bernardino}}]{mauricioComparingVisionTransformers2023}%
  \BibitemOpen
  \bibfield  {author} {\bibinfo {author} {\bibfnamefont {J.}~\bibnamefont {Maur{\'i}cio}}, \bibinfo {author} {\bibfnamefont {I.}~\bibnamefont {Domingues}},\ and\ \bibinfo {author} {\bibfnamefont {J.}~\bibnamefont {Bernardino}},\ }\bibfield  {title} {\bibinfo {title} {Comparing {{Vision Transformers}} and {{Convolutional Neural Networks}} for {{Image Classification}}: {{A Literature Review}}},\ }\href {https://doi.org/10.3390/app13095521} {\bibfield  {journal} {\bibinfo  {journal} {Applied Sciences}\ }\textbf {\bibinfo {volume} {13}},\ \bibinfo {pages} {5521} (\bibinfo {year} {2023})}\BibitemShut {NoStop}%
\bibitem [{\citenamefont {Dosovitskiy}\ \emph {et~al.}(2021)\citenamefont {Dosovitskiy}, \citenamefont {Beyer}, \citenamefont {Kolesnikov}, \citenamefont {Weissenborn}, \citenamefont {Zhai}, \citenamefont {Unterthiner}, \citenamefont {Dehghani}, \citenamefont {Minderer}, \citenamefont {Heigold}, \citenamefont {Gelly}, \citenamefont {Uszkoreit},\ and\ \citenamefont {Houlsby}}]{dosovitskiy2021}%
  \BibitemOpen
  \bibfield  {author} {\bibinfo {author} {\bibfnamefont {A.}~\bibnamefont {Dosovitskiy}}, \bibinfo {author} {\bibfnamefont {L.}~\bibnamefont {Beyer}}, \bibinfo {author} {\bibfnamefont {A.}~\bibnamefont {Kolesnikov}}, \bibinfo {author} {\bibfnamefont {D.}~\bibnamefont {Weissenborn}}, \bibinfo {author} {\bibfnamefont {X.}~\bibnamefont {Zhai}}, \bibinfo {author} {\bibfnamefont {T.}~\bibnamefont {Unterthiner}}, \bibinfo {author} {\bibfnamefont {M.}~\bibnamefont {Dehghani}}, \bibinfo {author} {\bibfnamefont {M.}~\bibnamefont {Minderer}}, \bibinfo {author} {\bibfnamefont {G.}~\bibnamefont {Heigold}}, \bibinfo {author} {\bibfnamefont {S.}~\bibnamefont {Gelly}}, \bibinfo {author} {\bibfnamefont {J.}~\bibnamefont {Uszkoreit}},\ and\ \bibinfo {author} {\bibfnamefont {N.}~\bibnamefont {Houlsby}},\ }\bibfield  {title} {\bibinfo {title} {An image is worth 16x16 words: Transformers for image recognition at scale},\ }in\ \href@noop {} {\emph {\bibinfo {booktitle} {International Conference on Learning Representations
  (ICLR)}}}\ (\bibinfo {year} {2021})\BibitemShut {NoStop}%
\bibitem [{\citenamefont {Prince}(2023)}]{prince2023understanding}%
  \BibitemOpen
  \bibfield  {author} {\bibinfo {author} {\bibfnamefont {S.~J.}\ \bibnamefont {Prince}},\ }\href {http://udlbook.com} {\emph {\bibinfo {title} {Understanding Deep Learning}}}\ (\bibinfo  {publisher} {MIT Press},\ \bibinfo {year} {2023})\BibitemShut {NoStop}%
\bibitem [{\citenamefont {Elesedy}\ and\ \citenamefont {Zaidi}(2021)}]{elesedy2021}%
  \BibitemOpen
  \bibfield  {author} {\bibinfo {author} {\bibfnamefont {B.}~\bibnamefont {Elesedy}}\ and\ \bibinfo {author} {\bibfnamefont {S.}~\bibnamefont {Zaidi}},\ }\bibfield  {title} {\bibinfo {title} {Provably strict generalisation benefit for equivariant models},\ }in\ \href@noop {} {\emph {\bibinfo {booktitle} {Proceedings of the 38th International Conference on Machine Learning}}},\ \bibinfo {series} {Proceedings of Machine Learning Research}, Vol.\ \bibinfo {volume} {139}\ (\bibinfo {year} {2021})\ pp.\ \bibinfo {pages} {2959--2969}\BibitemShut {NoStop}%
\bibitem [{\citenamefont {Bronstein}\ \emph {et~al.}(2021)\citenamefont {Bronstein}, \citenamefont {Bruna}, \citenamefont {Cohen},\ and\ \citenamefont {Veli{\v{c}}kovi{\'c}}}]{bronstein2021}%
  \BibitemOpen
  \bibfield  {author} {\bibinfo {author} {\bibfnamefont {M.~M.}\ \bibnamefont {Bronstein}}, \bibinfo {author} {\bibfnamefont {J.}~\bibnamefont {Bruna}}, \bibinfo {author} {\bibfnamefont {T.}~\bibnamefont {Cohen}},\ and\ \bibinfo {author} {\bibfnamefont {P.}~\bibnamefont {Veli{\v{c}}kovi{\'c}}},\ }\href@noop {} {\bibinfo {title} {Geometric deep learning: Grids, groups, graphs, geodesics, and gauges}} (\bibinfo {year} {2021}),\ \Eprint {https://arxiv.org/abs/2104.13478} {arXiv:2104.13478 [cs.LG]} \BibitemShut {NoStop}%
\bibitem [{\citenamefont {Maier}\ \emph {et~al.}(2019)\citenamefont {Maier}, \citenamefont {Syben}, \citenamefont {Stimpel}, \citenamefont {W{\"u}rfl}, \citenamefont {Hoffmann}, \citenamefont {Schebesch}, \citenamefont {Fu}, \citenamefont {Mill}, \citenamefont {Kling},\ and\ \citenamefont {Christiansen}}]{maier2019}%
  \BibitemOpen
  \bibfield  {author} {\bibinfo {author} {\bibfnamefont {A.~K.}\ \bibnamefont {Maier}}, \bibinfo {author} {\bibfnamefont {C.}~\bibnamefont {Syben}}, \bibinfo {author} {\bibfnamefont {B.}~\bibnamefont {Stimpel}}, \bibinfo {author} {\bibfnamefont {T.}~\bibnamefont {W{\"u}rfl}}, \bibinfo {author} {\bibfnamefont {M.}~\bibnamefont {Hoffmann}}, \bibinfo {author} {\bibfnamefont {F.}~\bibnamefont {Schebesch}}, \bibinfo {author} {\bibfnamefont {W.}~\bibnamefont {Fu}}, \bibinfo {author} {\bibfnamefont {L.}~\bibnamefont {Mill}}, \bibinfo {author} {\bibfnamefont {L.}~\bibnamefont {Kling}},\ and\ \bibinfo {author} {\bibfnamefont {S.}~\bibnamefont {Christiansen}},\ }\bibfield  {title} {\bibinfo {title} {Learning with known operators reduces maximum error bounds},\ }\href@noop {} {\bibfield  {journal} {\bibinfo  {journal} {Nature Machine Intelligence}\ }\textbf {\bibinfo {volume} {1}},\ \bibinfo {pages} {373} (\bibinfo {year} {2019})}\BibitemShut {NoStop}%
\bibitem [{\citenamefont {Monga}\ \emph {et~al.}(2021)\citenamefont {Monga}, \citenamefont {Li},\ and\ \citenamefont {Eldar}}]{monga2021}%
  \BibitemOpen
  \bibfield  {author} {\bibinfo {author} {\bibfnamefont {V.}~\bibnamefont {Monga}}, \bibinfo {author} {\bibfnamefont {Y.}~\bibnamefont {Li}},\ and\ \bibinfo {author} {\bibfnamefont {Y.~C.}\ \bibnamefont {Eldar}},\ }\bibfield  {title} {\bibinfo {title} {Algorithm unrolling: Interpretable, efficient deep learning for signal and image processing},\ }\href@noop {} {\bibfield  {journal} {\bibinfo  {journal} {{IEEE Signal Processing Magazine}}\ }\textbf {\bibinfo {volume} {38}},\ \bibinfo {pages} {18} (\bibinfo {year} {2021})}\BibitemShut {NoStop}%
\bibitem [{\citenamefont {Karniadakis}\ \emph {et~al.}(2021)\citenamefont {Karniadakis}, \citenamefont {Kevrekidis}, \citenamefont {Lu}, \citenamefont {Perdikaris}, \citenamefont {Wang},\ and\ \citenamefont {Yang}}]{karniadakis2021}%
  \BibitemOpen
  \bibfield  {author} {\bibinfo {author} {\bibfnamefont {G.~E.}\ \bibnamefont {Karniadakis}}, \bibinfo {author} {\bibfnamefont {I.~G.}\ \bibnamefont {Kevrekidis}}, \bibinfo {author} {\bibfnamefont {L.}~\bibnamefont {Lu}}, \bibinfo {author} {\bibfnamefont {P.}~\bibnamefont {Perdikaris}}, \bibinfo {author} {\bibfnamefont {S.}~\bibnamefont {Wang}},\ and\ \bibinfo {author} {\bibfnamefont {L.}~\bibnamefont {Yang}},\ }\bibfield  {title} {\bibinfo {title} {Physics-informed machine learning},\ }\href@noop {} {\bibfield  {journal} {\bibinfo  {journal} {Nature Reviews Physics}\ }\textbf {\bibinfo {volume} {3}},\ \bibinfo {pages} {422} (\bibinfo {year} {2021})}\BibitemShut {NoStop}%
\bibitem [{\citenamefont {McClean}\ \emph {et~al.}(2018)\citenamefont {McClean}, \citenamefont {Boixo}, \citenamefont {Smelyanskiy}, \citenamefont {Babbush},\ and\ \citenamefont {Neven}}]{mcclean2018barren}%
  \BibitemOpen
  \bibfield  {author} {\bibinfo {author} {\bibfnamefont {J.~R.}\ \bibnamefont {McClean}}, \bibinfo {author} {\bibfnamefont {S.}~\bibnamefont {Boixo}}, \bibinfo {author} {\bibfnamefont {V.~N.}\ \bibnamefont {Smelyanskiy}}, \bibinfo {author} {\bibfnamefont {R.}~\bibnamefont {Babbush}},\ and\ \bibinfo {author} {\bibfnamefont {H.}~\bibnamefont {Neven}},\ }\bibfield  {title} {\bibinfo {title} {Barren plateaus in quantum neural network training landscapes},\ }\href@noop {} {\bibfield  {journal} {\bibinfo  {journal} {Nature communications}\ }\textbf {\bibinfo {volume} {9}},\ \bibinfo {pages} {4812} (\bibinfo {year} {2018})}\BibitemShut {NoStop}%
\bibitem [{\citenamefont {Larocca}\ \emph {et~al.}(2025)\citenamefont {Larocca}, \citenamefont {Thanasilp}, \citenamefont {Wang}, \citenamefont {Sharma}, \citenamefont {Biamonte}, \citenamefont {Coles}, \citenamefont {Cincio}, \citenamefont {McClean}, \citenamefont {Holmes},\ and\ \citenamefont {Cerezo}}]{larocca2025barren}%
  \BibitemOpen
  \bibfield  {author} {\bibinfo {author} {\bibfnamefont {M.}~\bibnamefont {Larocca}}, \bibinfo {author} {\bibfnamefont {S.}~\bibnamefont {Thanasilp}}, \bibinfo {author} {\bibfnamefont {S.}~\bibnamefont {Wang}}, \bibinfo {author} {\bibfnamefont {K.}~\bibnamefont {Sharma}}, \bibinfo {author} {\bibfnamefont {J.}~\bibnamefont {Biamonte}}, \bibinfo {author} {\bibfnamefont {P.~J.}\ \bibnamefont {Coles}}, \bibinfo {author} {\bibfnamefont {L.}~\bibnamefont {Cincio}}, \bibinfo {author} {\bibfnamefont {J.~R.}\ \bibnamefont {McClean}}, \bibinfo {author} {\bibfnamefont {Z.}~\bibnamefont {Holmes}},\ and\ \bibinfo {author} {\bibfnamefont {M.}~\bibnamefont {Cerezo}},\ }\bibfield  {title} {\bibinfo {title} {Barren plateaus in variational quantum computing},\ }\href@noop {} {\bibfield  {journal} {\bibinfo  {journal} {Nature Reviews Physics}\ }\textbf {\bibinfo {volume} {7}},\ \bibinfo {pages} {174} (\bibinfo {year} {2025})}\BibitemShut {NoStop}%
\bibitem [{\citenamefont {Larocca}\ \emph {et~al.}(2022)\citenamefont {Larocca}, \citenamefont {Sauvage}, \citenamefont {Sbahi}, \citenamefont {Verdon}, \citenamefont {Coles},\ and\ \citenamefont {Cerezo}}]{larocca2022group}%
  \BibitemOpen
  \bibfield  {author} {\bibinfo {author} {\bibfnamefont {M.}~\bibnamefont {Larocca}}, \bibinfo {author} {\bibfnamefont {F.}~\bibnamefont {Sauvage}}, \bibinfo {author} {\bibfnamefont {F.~M.}\ \bibnamefont {Sbahi}}, \bibinfo {author} {\bibfnamefont {G.}~\bibnamefont {Verdon}}, \bibinfo {author} {\bibfnamefont {P.~J.}\ \bibnamefont {Coles}},\ and\ \bibinfo {author} {\bibfnamefont {M.}~\bibnamefont {Cerezo}},\ }\bibfield  {title} {\bibinfo {title} {Group-invariant quantum machine learning},\ }\href@noop {} {\bibfield  {journal} {\bibinfo  {journal} {PRX Quantum}\ }\textbf {\bibinfo {volume} {3}},\ \bibinfo {pages} {030341} (\bibinfo {year} {2022})}\BibitemShut {NoStop}%
\bibitem [{\citenamefont {Meyer}\ \emph {et~al.}(2023)\citenamefont {Meyer}, \citenamefont {Mularski}, \citenamefont {Gil-Fuster}, \citenamefont {Mele}, \citenamefont {Arzani}, \citenamefont {Wilms},\ and\ \citenamefont {Eisert}}]{meyer2023}%
  \BibitemOpen
  \bibfield  {author} {\bibinfo {author} {\bibfnamefont {J.~J.}\ \bibnamefont {Meyer}}, \bibinfo {author} {\bibfnamefont {M.}~\bibnamefont {Mularski}}, \bibinfo {author} {\bibfnamefont {E.}~\bibnamefont {Gil-Fuster}}, \bibinfo {author} {\bibfnamefont {A.~A.}\ \bibnamefont {Mele}}, \bibinfo {author} {\bibfnamefont {F.}~\bibnamefont {Arzani}}, \bibinfo {author} {\bibfnamefont {A.}~\bibnamefont {Wilms}},\ and\ \bibinfo {author} {\bibfnamefont {J.}~\bibnamefont {Eisert}},\ }\bibfield  {title} {\bibinfo {title} {Exploiting symmetry in variational quantum machine learning},\ }\href@noop {} {\bibfield  {journal} {\bibinfo  {journal} {PRX Quantum}\ }\textbf {\bibinfo {volume} {4}},\ \bibinfo {pages} {010328} (\bibinfo {year} {2023})}\BibitemShut {NoStop}%
\bibitem [{\citenamefont {Nguyen}\ \emph {et~al.}(2024)\citenamefont {Nguyen}, \citenamefont {Schatzki}, \citenamefont {Braccia}, \citenamefont {Ragone}, \citenamefont {Coles}, \citenamefont {Sauvage}, \citenamefont {Larocca},\ and\ \citenamefont {Cerezo}}]{nguyen2024}%
  \BibitemOpen
  \bibfield  {author} {\bibinfo {author} {\bibfnamefont {Q.~T.}\ \bibnamefont {Nguyen}}, \bibinfo {author} {\bibfnamefont {L.}~\bibnamefont {Schatzki}}, \bibinfo {author} {\bibfnamefont {P.}~\bibnamefont {Braccia}}, \bibinfo {author} {\bibfnamefont {M.}~\bibnamefont {Ragone}}, \bibinfo {author} {\bibfnamefont {P.~J.}\ \bibnamefont {Coles}}, \bibinfo {author} {\bibfnamefont {F.}~\bibnamefont {Sauvage}}, \bibinfo {author} {\bibfnamefont {M.}~\bibnamefont {Larocca}},\ and\ \bibinfo {author} {\bibfnamefont {M.}~\bibnamefont {Cerezo}},\ }\bibfield  {title} {\bibinfo {title} {Theory for equivariant quantum neural networks},\ }\href@noop {} {\bibfield  {journal} {\bibinfo  {journal} {PRX Quantum}\ }\textbf {\bibinfo {volume} {5}},\ \bibinfo {pages} {020328} (\bibinfo {year} {2024})}\BibitemShut {NoStop}%
\bibitem [{\citenamefont {Schatzki}\ \emph {et~al.}(2024)\citenamefont {Schatzki}, \citenamefont {Larocca}, \citenamefont {Nguyen}, \citenamefont {Sauvage},\ and\ \citenamefont {Cerezo}}]{schatzki2024}%
  \BibitemOpen
  \bibfield  {author} {\bibinfo {author} {\bibfnamefont {L.}~\bibnamefont {Schatzki}}, \bibinfo {author} {\bibfnamefont {M.}~\bibnamefont {Larocca}}, \bibinfo {author} {\bibfnamefont {Q.~T.}\ \bibnamefont {Nguyen}}, \bibinfo {author} {\bibfnamefont {F.}~\bibnamefont {Sauvage}},\ and\ \bibinfo {author} {\bibfnamefont {M.}~\bibnamefont {Cerezo}},\ }\bibfield  {title} {\bibinfo {title} {Theoretical guarantees for permutation-equivariant quantum neural networks},\ }\href@noop {} {\bibfield  {journal} {\bibinfo  {journal} {npj Quantum Information}\ }\textbf {\bibinfo {volume} {10}},\ \bibinfo {pages} {12} (\bibinfo {year} {2024})}\BibitemShut {NoStop}%
\bibitem [{\citenamefont {Schuld}\ \emph {et~al.}(2021)\citenamefont {Schuld}, \citenamefont {Sweke},\ and\ \citenamefont {Meyer}}]{schuldEffectDataEncoding2021}%
  \BibitemOpen
  \bibfield  {author} {\bibinfo {author} {\bibfnamefont {M.}~\bibnamefont {Schuld}}, \bibinfo {author} {\bibfnamefont {R.}~\bibnamefont {Sweke}},\ and\ \bibinfo {author} {\bibfnamefont {J.~J.}\ \bibnamefont {Meyer}},\ }\bibfield  {title} {\bibinfo {title} {The effect of data encoding on the expressive power of variational quantum machine learning models},\ }\href@noop {} {\bibfield  {journal} {\bibinfo  {journal} {Physical Review A}\ }\textbf {\bibinfo {volume} {103}},\ \bibinfo {pages} {032430} (\bibinfo {year} {2021})}\BibitemShut {NoStop}%
\bibitem [{\citenamefont {{P{\'e}rez-Salinas}}\ \emph {et~al.}(2020)\citenamefont {{P{\'e}rez-Salinas}}, \citenamefont {{Cervera-Lierta}}, \citenamefont {{Gil-Fuster}},\ and\ \citenamefont {Latorre}}]{perez-salinasDataReuploadingUniversal2020}%
  \BibitemOpen
  \bibfield  {author} {\bibinfo {author} {\bibfnamefont {A.}~\bibnamefont {{P{\'e}rez-Salinas}}}, \bibinfo {author} {\bibfnamefont {A.}~\bibnamefont {{Cervera-Lierta}}}, \bibinfo {author} {\bibfnamefont {E.}~\bibnamefont {{Gil-Fuster}}},\ and\ \bibinfo {author} {\bibfnamefont {J.~I.}\ \bibnamefont {Latorre}},\ }\bibfield  {title} {\bibinfo {title} {Data re-uploading for a universal quantum classifier},\ }\href@noop {} {\bibfield  {journal} {\bibinfo  {journal} {Quantum}\ }\textbf {\bibinfo {volume} {4}},\ \bibinfo {pages} {226} (\bibinfo {year} {2020})}\BibitemShut {NoStop}%
\bibitem [{\citenamefont {Caro}\ \emph {et~al.}(2021)\citenamefont {Caro}, \citenamefont {Gil-Fuster}, \citenamefont {Meyer}, \citenamefont {Eisert},\ and\ \citenamefont {Sweke}}]{caro2021encoding}%
  \BibitemOpen
  \bibfield  {author} {\bibinfo {author} {\bibfnamefont {M.~C.}\ \bibnamefont {Caro}}, \bibinfo {author} {\bibfnamefont {E.}~\bibnamefont {Gil-Fuster}}, \bibinfo {author} {\bibfnamefont {J.~J.}\ \bibnamefont {Meyer}}, \bibinfo {author} {\bibfnamefont {J.}~\bibnamefont {Eisert}},\ and\ \bibinfo {author} {\bibfnamefont {R.}~\bibnamefont {Sweke}},\ }\bibfield  {title} {\bibinfo {title} {Encoding-dependent generalization bounds for parametrized quantum circuits},\ }\href@noop {} {\bibfield  {journal} {\bibinfo  {journal} {Quantum}\ }\textbf {\bibinfo {volume} {5}},\ \bibinfo {pages} {582} (\bibinfo {year} {2021})}\BibitemShut {NoStop}%
\bibitem [{\citenamefont {Shin}\ \emph {et~al.}(2023)\citenamefont {Shin}, \citenamefont {Teo},\ and\ \citenamefont {Jeong}}]{shin2023}%
  \BibitemOpen
  \bibfield  {author} {\bibinfo {author} {\bibfnamefont {S.}~\bibnamefont {Shin}}, \bibinfo {author} {\bibfnamefont {Y.~S.}\ \bibnamefont {Teo}},\ and\ \bibinfo {author} {\bibfnamefont {H.}~\bibnamefont {Jeong}},\ }\bibfield  {title} {\bibinfo {title} {Exponential data encoding for quantum supervised learning},\ }\href@noop {} {\bibfield  {journal} {\bibinfo  {journal} {Physical Review A}\ }\textbf {\bibinfo {volume} {107}},\ \bibinfo {pages} {012422} (\bibinfo {year} {2023})}\BibitemShut {NoStop}%
\bibitem [{\citenamefont {Jaderberg}\ \emph {et~al.}(2024)\citenamefont {Jaderberg}, \citenamefont {Gentile}, \citenamefont {Berrada}, \citenamefont {Shishenina},\ and\ \citenamefont {Elfving}}]{jaderberg2024}%
  \BibitemOpen
  \bibfield  {author} {\bibinfo {author} {\bibfnamefont {B.}~\bibnamefont {Jaderberg}}, \bibinfo {author} {\bibfnamefont {A.~A.}\ \bibnamefont {Gentile}}, \bibinfo {author} {\bibfnamefont {Y.~A.}\ \bibnamefont {Berrada}}, \bibinfo {author} {\bibfnamefont {E.}~\bibnamefont {Shishenina}},\ and\ \bibinfo {author} {\bibfnamefont {V.~E.}\ \bibnamefont {Elfving}},\ }\bibfield  {title} {\bibinfo {title} {Let quantum neural networks choose their own frequencies},\ }\href@noop {} {\bibfield  {journal} {\bibinfo  {journal} {Physical Review A}\ }\textbf {\bibinfo {volume} {109}},\ \bibinfo {pages} {042421} (\bibinfo {year} {2024})}\BibitemShut {NoStop}%
\bibitem [{\citenamefont {Barthe}\ and\ \citenamefont {{P{\'e}rez-Salinas}}(2024)}]{bartheGradientsFrequencyProfiles2024}%
  \BibitemOpen
  \bibfield  {author} {\bibinfo {author} {\bibfnamefont {A.}~\bibnamefont {Barthe}}\ and\ \bibinfo {author} {\bibfnamefont {A.}~\bibnamefont {{P{\'e}rez-Salinas}}},\ }\bibfield  {title} {\bibinfo {title} {Gradients and frequency profiles of quantum re-uploading models},\ }\href@noop {} {\bibfield  {journal} {\bibinfo  {journal} {Quantum}\ }\textbf {\bibinfo {volume} {8}},\ \bibinfo {pages} {1523} (\bibinfo {year} {2024})}\BibitemShut {NoStop}%
\bibitem [{\citenamefont {Duffy}\ and\ \citenamefont {Jastrzebski}(2025)}]{duffy2025}%
  \BibitemOpen
  \bibfield  {author} {\bibinfo {author} {\bibfnamefont {C.}~\bibnamefont {Duffy}}\ and\ \bibinfo {author} {\bibfnamefont {M.}~\bibnamefont {Jastrzebski}},\ }\href@noop {} {\bibinfo {title} {Spectral bias in variational quantum machine learning}} (\bibinfo {year} {2025}),\ \Eprint {https://arxiv.org/abs/2506.22555} {arXiv:2506.22555 [quant-ph]} \BibitemShut {NoStop}%
\bibitem [{\citenamefont {Wang}\ \emph {et~al.}(2026)\citenamefont {Wang}, \citenamefont {Parampalli},\ and\ \citenamefont {Myers}}]{qsm2026}%
  \BibitemOpen
  \bibfield  {author} {\bibinfo {author} {\bibfnamefont {P.}~\bibnamefont {Wang}}, \bibinfo {author} {\bibfnamefont {U.}~\bibnamefont {Parampalli}},\ and\ \bibinfo {author} {\bibfnamefont {C.~R.}\ \bibnamefont {Myers}},\ }\href@noop {} {\bibinfo {title} {Quantum spectral model: Data reuploading with input-conditioned frequency support}} (\bibinfo {year} {2026}),\ \Eprint {https://arxiv.org/abs/2607.22516} {arXiv:2607.22516 [quant-ph]} \BibitemShut {NoStop}%
\bibitem [{\citenamefont {Bagchi}\ and\ \citenamefont {Mitra}(2012)}]{bagchi2012nonuniform}%
  \BibitemOpen
  \bibfield  {author} {\bibinfo {author} {\bibfnamefont {S.}~\bibnamefont {Bagchi}}\ and\ \bibinfo {author} {\bibfnamefont {S.~K.}\ \bibnamefont {Mitra}},\ }\href@noop {} {\emph {\bibinfo {title} {The nonuniform discrete {Fourier} transform and its applications in signal processing}}},\ Vol.\ \bibinfo {volume} {463}\ (\bibinfo  {publisher} {Springer Science \& Business Media},\ \bibinfo {year} {2012})\BibitemShut {NoStop}%
\bibitem [{\citenamefont {Fessler}\ and\ \citenamefont {Sutton}(2003)}]{fessler2003nonuniform}%
  \BibitemOpen
  \bibfield  {author} {\bibinfo {author} {\bibfnamefont {J.~A.}\ \bibnamefont {Fessler}}\ and\ \bibinfo {author} {\bibfnamefont {B.~P.}\ \bibnamefont {Sutton}},\ }\bibfield  {title} {\bibinfo {title} {Nonuniform fast {F}ourier transforms using min-max interpolation},\ }\href@noop {} {\bibfield  {journal} {\bibinfo  {journal} {{IEEE transactions on signal processing}}\ }\textbf {\bibinfo {volume} {51}},\ \bibinfo {pages} {560} (\bibinfo {year} {2003})}\BibitemShut {NoStop}%
\bibitem [{\citenamefont {Greengard}\ and\ \citenamefont {Lee}(2004)}]{greengard2004}%
  \BibitemOpen
  \bibfield  {author} {\bibinfo {author} {\bibfnamefont {L.}~\bibnamefont {Greengard}}\ and\ \bibinfo {author} {\bibfnamefont {J.-Y.}\ \bibnamefont {Lee}},\ }\bibfield  {title} {\bibinfo {title} {Accelerating the nonuniform fast {Fourier} transform},\ }\href@noop {} {\bibfield  {journal} {\bibinfo  {journal} {SIAM Review}\ }\textbf {\bibinfo {volume} {46}},\ \bibinfo {pages} {443} (\bibinfo {year} {2004})}\BibitemShut {NoStop}%
\bibitem [{\citenamefont {Thompson}\ \emph {et~al.}(2017)\citenamefont {Thompson}, \citenamefont {Moran},\ and\ \citenamefont {Swenson}}]{thompson2017}%
  \BibitemOpen
  \bibfield  {author} {\bibinfo {author} {\bibfnamefont {A.~R.}\ \bibnamefont {Thompson}}, \bibinfo {author} {\bibfnamefont {J.~M.}\ \bibnamefont {Moran}},\ and\ \bibinfo {author} {\bibfnamefont {G.~W.}\ \bibnamefont {Swenson}, \bibfnamefont {Jr.}},\ }\href@noop {} {\emph {\bibinfo {title} {Interferometry and Synthesis in Radio Astronomy}}},\ \bibinfo {edition} {3rd}\ ed.\ (\bibinfo  {publisher} {Springer},\ \bibinfo {address} {Cham},\ \bibinfo {year} {2017})\BibitemShut {NoStop}%
\bibitem [{\citenamefont {Moreira}\ \emph {et~al.}(2013)\citenamefont {Moreira}, \citenamefont {Prats-Iraola}, \citenamefont {Younis}, \citenamefont {Krieger}, \citenamefont {Hajnsek},\ and\ \citenamefont {Papathanassiou}}]{moreira2013tutorial}%
  \BibitemOpen
  \bibfield  {author} {\bibinfo {author} {\bibfnamefont {A.}~\bibnamefont {Moreira}}, \bibinfo {author} {\bibfnamefont {P.}~\bibnamefont {Prats-Iraola}}, \bibinfo {author} {\bibfnamefont {M.}~\bibnamefont {Younis}}, \bibinfo {author} {\bibfnamefont {G.}~\bibnamefont {Krieger}}, \bibinfo {author} {\bibfnamefont {I.}~\bibnamefont {Hajnsek}},\ and\ \bibinfo {author} {\bibfnamefont {K.~P.}\ \bibnamefont {Papathanassiou}},\ }\bibfield  {title} {\bibinfo {title} {A tutorial on synthetic aperture radar},\ }\href@noop {} {\bibfield  {journal} {\bibinfo  {journal} {IEEE Geoscience and Remote Sensing Magazine}\ }\textbf {\bibinfo {volume} {1}},\ \bibinfo {pages} {6} (\bibinfo {year} {2013})}\BibitemShut {NoStop}%
\bibitem [{\citenamefont {Cumming}\ and\ \citenamefont {Wong}(2005)}]{cumming2005digital}%
  \BibitemOpen
  \bibfield  {author} {\bibinfo {author} {\bibfnamefont {I.~G.}\ \bibnamefont {Cumming}}\ and\ \bibinfo {author} {\bibfnamefont {F.~H.}\ \bibnamefont {Wong}},\ }\href@noop {} {\emph {\bibinfo {title} {Digital Processing of Synthetic Aperture Radar Data: Algorithms and Implementation}}}\ (\bibinfo  {publisher} {Artech House},\ \bibinfo {year} {2005})\BibitemShut {NoStop}%
\bibitem [{\citenamefont {Flores}\ \emph {et~al.}(2019)\citenamefont {Flores}, \citenamefont {Herndon}, \citenamefont {Thapa},\ and\ \citenamefont {Cherrington}}]{floresSyntheticApertureRadar2019}%
  \BibitemOpen
  \bibfield  {author} {\bibinfo {author} {\bibfnamefont {A.}~\bibnamefont {Flores}}, \bibinfo {author} {\bibfnamefont {K.}~\bibnamefont {Herndon}}, \bibinfo {author} {\bibfnamefont {R.}~\bibnamefont {Thapa}},\ and\ \bibinfo {author} {\bibfnamefont {E.}~\bibnamefont {Cherrington}},\ }\href@noop {} {\emph {\bibinfo {title} {Synthetic Aperture Radar (SAR) Handbook: Comprehensive Methodologies for Forest Monitoring and Biomass Estimation}}}\ (\bibinfo  {publisher} {NASA},\ \bibinfo {year} {2019})\BibitemShut {NoStop}%
\bibitem [{\citenamefont {Garc{\'i}a}\ \emph {et~al.}(2024)\citenamefont {Garc{\'i}a}, \citenamefont {Furano}, \citenamefont {Ghiglione}, \citenamefont {Zancan}, \citenamefont {Imbembo}, \citenamefont {Ilioudis}, \citenamefont {Clemente},\ and\ \citenamefont {Trucco}}]{garciaAdvancementsOnboardProcessing2024}%
  \BibitemOpen
  \bibfield  {author} {\bibinfo {author} {\bibfnamefont {L.~P.}\ \bibnamefont {Garc{\'i}a}}, \bibinfo {author} {\bibfnamefont {G.}~\bibnamefont {Furano}}, \bibinfo {author} {\bibfnamefont {M.}~\bibnamefont {Ghiglione}}, \bibinfo {author} {\bibfnamefont {V.}~\bibnamefont {Zancan}}, \bibinfo {author} {\bibfnamefont {E.}~\bibnamefont {Imbembo}}, \bibinfo {author} {\bibfnamefont {C.}~\bibnamefont {Ilioudis}}, \bibinfo {author} {\bibfnamefont {C.}~\bibnamefont {Clemente}},\ and\ \bibinfo {author} {\bibfnamefont {P.}~\bibnamefont {Trucco}},\ }\bibfield  {title} {\bibinfo {title} {Advancements in {{Onboard Processing}} of {{Synthetic Aperture Radar}} ({{SAR}}) {{Data}}: {{Enhancing Efficiency}} and {{Real-Time Capabilities}}},\ }\href {https://doi.org/10.1109/JSTARS.2024.3406155} {\bibfield  {journal} {\bibinfo  {journal} {IEEE Journal of Selected Topics in Applied Earth Observations and Remote Sensing}\ }\textbf {\bibinfo {volume} {17}},\ \bibinfo {pages} {16625} (\bibinfo {year} {2024})}\BibitemShut {NoStop}%
\bibitem [{\citenamefont {Gaburro}\ \emph {et~al.}(2026)\citenamefont {Gaburro}, \citenamefont {Healy}, \citenamefont {Naidu},\ and\ \citenamefont {Nolan}}]{gaburro2025}%
  \BibitemOpen
  \bibfield  {author} {\bibinfo {author} {\bibfnamefont {R.}~\bibnamefont {Gaburro}}, \bibinfo {author} {\bibfnamefont {P.}~\bibnamefont {Healy}}, \bibinfo {author} {\bibfnamefont {S.}~\bibnamefont {Naidu}},\ and\ \bibinfo {author} {\bibfnamefont {C.}~\bibnamefont {Nolan}},\ }\bibfield  {title} {\bibinfo {title} {A machine learning approach for image classification in synthetic aperture {RADAR}},\ }\href {https://doi.org/10.1093/imamat/hxag009} {\bibfield  {journal} {\bibinfo  {journal} {IMA Journal of Applied Mathematics}\ }\textbf {\bibinfo {volume} {91}},\ \bibinfo {pages} {15} (\bibinfo {year} {2026})},\ \bibinfo {note} {arXiv:2508.04234}\BibitemShut {NoStop}%
\bibitem [{\citenamefont {Naik}\ \emph {et~al.}(2024)\citenamefont {Naik}, \citenamefont {Vaughn}, \citenamefont {Uehara}, \citenamefont {Spanias},\ and\ \citenamefont {Jaskie}}]{naik2024}%
  \BibitemOpen
  \bibfield  {author} {\bibinfo {author} {\bibfnamefont {S.}~\bibnamefont {Naik}}, \bibinfo {author} {\bibfnamefont {N.}~\bibnamefont {Vaughn}}, \bibinfo {author} {\bibfnamefont {G.}~\bibnamefont {Uehara}}, \bibinfo {author} {\bibfnamefont {A.}~\bibnamefont {Spanias}},\ and\ \bibinfo {author} {\bibfnamefont {K.}~\bibnamefont {Jaskie}},\ }\bibfield  {title} {\bibinfo {title} {Quantum classification for synthetic aperture radar},\ }in\ \href@noop {} {\emph {\bibinfo {booktitle} {Automatic Target Recognition XXXIV}}},\ \bibinfo {series} {Proceedings of SPIE}, Vol.\ \bibinfo {volume} {13039},\ \bibinfo {editor} {edited by\ \bibinfo {editor} {\bibfnamefont {K.}~\bibnamefont {Chen}}, \bibinfo {editor} {\bibfnamefont {R.~I.}\ \bibnamefont {Hammoud}},\ and\ \bibinfo {editor} {\bibfnamefont {T.~L.}\ \bibnamefont {Overman}}}\ (\bibinfo {year} {2024})\ p.\ \bibinfo {pages} {130390H}\BibitemShut {NoStop}%
\bibitem [{\citenamefont {Tanner}\ \emph {et~al.}(2025)\citenamefont {Tanner}, \citenamefont {Davies}, \citenamefont {Elahi}, \citenamefont {Myers}, \citenamefont {Huynh}, \citenamefont {Liu}, \citenamefont {Reynolds},\ and\ \citenamefont {Wang}}]{maritimeqkm2025}%
  \BibitemOpen
  \bibfield  {author} {\bibinfo {author} {\bibfnamefont {J.}~\bibnamefont {Tanner}}, \bibinfo {author} {\bibfnamefont {N.}~\bibnamefont {Davies}}, \bibinfo {author} {\bibfnamefont {P.~J.}\ \bibnamefont {Elahi}}, \bibinfo {author} {\bibfnamefont {C.~R.}\ \bibnamefont {Myers}}, \bibinfo {author} {\bibfnamefont {D.}~\bibnamefont {Huynh}}, \bibinfo {author} {\bibfnamefont {W.}~\bibnamefont {Liu}}, \bibinfo {author} {\bibfnamefont {M.}~\bibnamefont {Reynolds}},\ and\ \bibinfo {author} {\bibfnamefont {J.}~\bibnamefont {Wang}},\ }\href@noop {} {\bibinfo {title} {Maritime object classification with {SAR} imagery using quantum kernel methods}} (\bibinfo {year} {2025}),\ \Eprint {https://arxiv.org/abs/2512.11367} {arXiv:2512.11367 [quant-ph]} \BibitemShut {NoStop}%
\bibitem [{\citenamefont {Giovagnoli}\ \emph {et~al.}(2025)\citenamefont {Giovagnoli}, \citenamefont {Huber},\ and\ \citenamefont {Krieger}}]{giovagnoli2025}%
  \BibitemOpen
  \bibfield  {author} {\bibinfo {author} {\bibfnamefont {A.}~\bibnamefont {Giovagnoli}}, \bibinfo {author} {\bibfnamefont {S.}~\bibnamefont {Huber}},\ and\ \bibinfo {author} {\bibfnamefont {G.}~\bibnamefont {Krieger}},\ }\href@noop {} {\bibinfo {title} {A quantum range-{Doppler} algorithm for synthetic aperture radar image formation}} (\bibinfo {year} {2025}),\ \Eprint {https://arxiv.org/abs/2504.20811} {arXiv:2504.20811 [quant-ph]} \BibitemShut {NoStop}%
\bibitem [{\citenamefont {Al~Salahat}\ \emph {et~al.}(2025)\citenamefont {Al~Salahat}, \citenamefont {El~Moussawi},\ and\ \citenamefont {Ghandour}}]{alsalahat2025}%
  \BibitemOpen
  \bibfield  {author} {\bibinfo {author} {\bibfnamefont {K.}~\bibnamefont {Al~Salahat}}, \bibinfo {author} {\bibfnamefont {M.}~\bibnamefont {El~Moussawi}},\ and\ \bibinfo {author} {\bibfnamefont {A.~J.}\ \bibnamefont {Ghandour}},\ }\href@noop {} {\bibinfo {title} {Quantum meets {SAR}: A novel range-{Doppler} algorithm for next-gen earth observation}} (\bibinfo {year} {2025}),\ \Eprint {https://arxiv.org/abs/2504.01832} {arXiv:2504.01832 [quant-ph]} \BibitemShut {NoStop}%
\bibitem [{\citenamefont {P{\'e}rez-Salinas}\ \emph {et~al.}(2025)\citenamefont {P{\'e}rez-Salinas}, \citenamefont {Yaghubi~Rad}, \citenamefont {Barthe},\ and\ \citenamefont {Dunjko}}]{perez-salinasUniversalApproximationContinuous2025}%
  \BibitemOpen
  \bibfield  {author} {\bibinfo {author} {\bibfnamefont {A.}~\bibnamefont {P{\'e}rez-Salinas}}, \bibinfo {author} {\bibfnamefont {M.}~\bibnamefont {Yaghubi~Rad}}, \bibinfo {author} {\bibfnamefont {A.}~\bibnamefont {Barthe}},\ and\ \bibinfo {author} {\bibfnamefont {V.}~\bibnamefont {Dunjko}},\ }\bibfield  {title} {\bibinfo {title} {Universal approximation of continuous functions with minimal quantum circuits},\ }\href@noop {} {\bibfield  {journal} {\bibinfo  {journal} {Physical Review Research}\ }\textbf {\bibinfo {volume} {7}},\ \bibinfo {pages} {043282} (\bibinfo {year} {2025})}\BibitemShut {NoStop}%
\bibitem [{\citenamefont {Manzano}\ \emph {et~al.}(2025)\citenamefont {Manzano}, \citenamefont {Dechant}, \citenamefont {Tura},\ and\ \citenamefont {Dunjko}}]{manzanoApproximationGeneralizationCapacities2025}%
  \BibitemOpen
  \bibfield  {author} {\bibinfo {author} {\bibfnamefont {A.}~\bibnamefont {Manzano}}, \bibinfo {author} {\bibfnamefont {D.}~\bibnamefont {Dechant}}, \bibinfo {author} {\bibfnamefont {J.}~\bibnamefont {Tura}},\ and\ \bibinfo {author} {\bibfnamefont {V.}~\bibnamefont {Dunjko}},\ }\bibfield  {title} {\bibinfo {title} {Approximation and {{Generalization Capacities}} of {{Parametrized Quantum Circuits}} for {{Functions}} in {{Sobolev Spaces}}},\ }\href@noop {} {\bibfield  {journal} {\bibinfo  {journal} {Quantum}\ }\textbf {\bibinfo {volume} {9}},\ \bibinfo {pages} {1658} (\bibinfo {year} {2025})}\BibitemShut {NoStop}%
\bibitem [{\citenamefont {Nikias}\ and\ \citenamefont {Mendel}(1993)}]{nikias1993signal}%
  \BibitemOpen
  \bibfield  {author} {\bibinfo {author} {\bibfnamefont {C.~L.}\ \bibnamefont {Nikias}}\ and\ \bibinfo {author} {\bibfnamefont {J.~M.}\ \bibnamefont {Mendel}},\ }\bibfield  {title} {\bibinfo {title} {Signal processing with higher-order spectra},\ }\href@noop {} {\bibfield  {journal} {\bibinfo  {journal} {{IEEE Signal processing magazine}}\ }\textbf {\bibinfo {volume} {10}},\ \bibinfo {pages} {10} (\bibinfo {year} {1993})}\BibitemShut {NoStop}%
\bibitem [{\citenamefont {Boashash}\ \emph {et~al.}(2015)\citenamefont {Boashash}, \citenamefont {Khan},\ and\ \citenamefont {Ben-Jabeur}}]{boashash2015time}%
  \BibitemOpen
  \bibfield  {author} {\bibinfo {author} {\bibfnamefont {B.}~\bibnamefont {Boashash}}, \bibinfo {author} {\bibfnamefont {N.~A.}\ \bibnamefont {Khan}},\ and\ \bibinfo {author} {\bibfnamefont {T.}~\bibnamefont {Ben-Jabeur}},\ }\bibfield  {title} {\bibinfo {title} {Time--frequency features for pattern recognition using high-resolution {TFDs}: A tutorial review},\ }\href@noop {} {\bibfield  {journal} {\bibinfo  {journal} {Digital Signal Processing}\ }\textbf {\bibinfo {volume} {40}},\ \bibinfo {pages} {1} (\bibinfo {year} {2015})}\BibitemShut {NoStop}%
\bibitem [{\citenamefont {Curlander}\ and\ \citenamefont {McDonough}(1991)}]{curlander1991synthetic}%
  \BibitemOpen
  \bibfield  {author} {\bibinfo {author} {\bibfnamefont {J.~C.}\ \bibnamefont {Curlander}}\ and\ \bibinfo {author} {\bibfnamefont {R.~N.}\ \bibnamefont {McDonough}},\ }\href@noop {} {\emph {\bibinfo {title} {Synthetic Aperture Radar: Systems and Signal Processing}}}\ (\bibinfo  {publisher} {Wiley},\ \bibinfo {year} {1991})\BibitemShut {NoStop}%
\bibitem [{\citenamefont {Bamler}\ and\ \citenamefont {Hartl}(1998)}]{bamlerSyntheticApertureRadar1998}%
  \BibitemOpen
  \bibfield  {author} {\bibinfo {author} {\bibfnamefont {R.}~\bibnamefont {Bamler}}\ and\ \bibinfo {author} {\bibfnamefont {P.}~\bibnamefont {Hartl}},\ }\bibfield  {title} {\bibinfo {title} {Synthetic aperture radar interferometry},\ }\href@noop {} {\bibfield  {journal} {\bibinfo  {journal} {Inverse Problems}\ }\textbf {\bibinfo {volume} {14}},\ \bibinfo {pages} {R1} (\bibinfo {year} {1998})}\BibitemShut {NoStop}%
\bibitem [{\citenamefont {Cossu}\ \emph {et~al.}(2009)\citenamefont {Cossu}, \citenamefont {Schoepfer}, \citenamefont {Bally},\ and\ \citenamefont {Fusco}}]{cossuRealtimeSARbasedProcessing2009}%
  \BibitemOpen
  \bibfield  {author} {\bibinfo {author} {\bibfnamefont {R.}~\bibnamefont {Cossu}}, \bibinfo {author} {\bibfnamefont {E.}~\bibnamefont {Schoepfer}}, \bibinfo {author} {\bibfnamefont {P.}~\bibnamefont {Bally}},\ and\ \bibinfo {author} {\bibfnamefont {L.}~\bibnamefont {Fusco}},\ }\bibfield  {title} {\bibinfo {title} {Near real-time {{SAR-based}} processing to support flood monitoring},\ }\href {https://doi.org/10.1007/s11554-009-0114-4} {\bibfield  {journal} {\bibinfo  {journal} {Journal of Real-Time Image Processing}\ }\textbf {\bibinfo {volume} {4}},\ \bibinfo {pages} {205} (\bibinfo {year} {2009})}\BibitemShut {NoStop}%
\bibitem [{\citenamefont {Rudin}(1976)}]{rudin_principles_1976}%
  \BibitemOpen
  \bibfield  {author} {\bibinfo {author} {\bibfnamefont {W.}~\bibnamefont {Rudin}},\ }\href@noop {} {\emph {\bibinfo {title} {Principles of Mathematical Analysis}}},\ \bibinfo {edition} {3rd}\ ed.\ (\bibinfo  {publisher} {McGraw--Hill},\ \bibinfo {year} {1976})\BibitemShut {NoStop}%
\bibitem [{\citenamefont {Oppenheim}\ and\ \citenamefont {Schafer}(2010)}]{oppenheim1999discrete}%
  \BibitemOpen
  \bibfield  {author} {\bibinfo {author} {\bibfnamefont {A.~V.}\ \bibnamefont {Oppenheim}}\ and\ \bibinfo {author} {\bibfnamefont {R.~W.}\ \bibnamefont {Schafer}},\ }\href@noop {} {\emph {\bibinfo {title} {Discrete-Time Signal Processing}}},\ \bibinfo {edition} {3rd}\ ed.\ (\bibinfo  {publisher} {Pearson},\ \bibinfo {year} {2010})\BibitemShut {NoStop}%
\bibitem [{\citenamefont {Cooley}\ and\ \citenamefont {Tukey}(1965)}]{cooley1965algorithm}%
  \BibitemOpen
  \bibfield  {author} {\bibinfo {author} {\bibfnamefont {J.~W.}\ \bibnamefont {Cooley}}\ and\ \bibinfo {author} {\bibfnamefont {J.~W.}\ \bibnamefont {Tukey}},\ }\bibfield  {title} {\bibinfo {title} {An algorithm for the machine calculation of complex {Fourier} series},\ }\href@noop {} {\bibfield  {journal} {\bibinfo  {journal} {Mathematics of computation}\ }\textbf {\bibinfo {volume} {19}},\ \bibinfo {pages} {297} (\bibinfo {year} {1965})}\BibitemShut {NoStop}%
\bibitem [{\citenamefont {Briggs}\ and\ \citenamefont {Henson}(1995)}]{briggs1995dft}%
  \BibitemOpen
  \bibfield  {author} {\bibinfo {author} {\bibfnamefont {W.~L.}\ \bibnamefont {Briggs}}\ and\ \bibinfo {author} {\bibfnamefont {V.~E.}\ \bibnamefont {Henson}},\ }\href@noop {} {\emph {\bibinfo {title} {The DFT: an owner's manual for the discrete {Fourier} transform}}}\ (\bibinfo  {publisher} {SIAM},\ \bibinfo {year} {1995})\BibitemShut {NoStop}%
\bibitem [{\citenamefont {Dutt}\ and\ \citenamefont {Rokhlin}(1993)}]{dutt1993fast}%
  \BibitemOpen
  \bibfield  {author} {\bibinfo {author} {\bibfnamefont {A.}~\bibnamefont {Dutt}}\ and\ \bibinfo {author} {\bibfnamefont {V.}~\bibnamefont {Rokhlin}},\ }\bibfield  {title} {\bibinfo {title} {Fast {F}ourier transforms for nonequispaced data},\ }\href@noop {} {\bibfield  {journal} {\bibinfo  {journal} {SIAM Journal on Scientific computing}\ }\textbf {\bibinfo {volume} {14}},\ \bibinfo {pages} {1368} (\bibinfo {year} {1993})}\BibitemShut {NoStop}%
\end{thebibliography}%

\setcounter{figure}{0}
\renewcommand{\thefigure}{A\arabic{figure}}
\renewcommand{\theHfigure}{appendix.\arabic{figure}}

\setcounter{table}{0}
\renewcommand{\thetable}{A\arabic{table}}

\clearpage

\phantom{}


\appendix

\subsection*{Notation}
\begin{table}[h]
\centering
\footnotesize
\renewcommand{\arraystretch}{1.05}
\setlength{\tabcolsep}{3pt}
\caption{\textbf{Key notation.}
Superscripts $(x)$ and $(f)$ distinguish signal and query quantities;
a star identifies task-prescribed or target quantities.}
\label{tab:notation}
\begin{tabularx}{\columnwidth}{@{}p{0.34\columnwidth}X@{}}
\toprule
Symbol & Meaning \\
\midrule

$\boldsymbol x$
& Sample-dependent signal; represented by real features
when supplied to the circuit. \\

$\boldsymbol f$, $\boldsymbol f^\star$
& Query input and its task-prescribed value,
shared across samples. \\

$\boldsymbol\theta$
& Trainable circuit parameters. \\

$M$, $N$, $n_q$
& Signal-feature dimension, query dimension,
and number of qubits. \\

$L$
& Total encoding layers; alternating DDR uses
$L/2$ paired signal--query blocks. \\

$v(\boldsymbol x,\boldsymbol f;\boldsymbol\theta)$
& Circuit output: the expectation value of
the observable $O$. \\

$E_\ell^{(x)}$, $E_\ell^{(f)}$
& Signal- and query-encoding gates in layer $\ell$. \\

$\mathcal E_x$, $\mathcal E_f$
& Signal- and query-encoding families. \\

$\Omega^{(x)}$, $\Omega^{(f)}$
& Encoding-generated Fourier modes in the
signal and query variables. \\

$\mathbf C(\boldsymbol\theta)$
& Joint Fourier coefficient matrix, indexed
by signal and query modes. \\

$\boldsymbol e(\boldsymbol x)$,
$\boldsymbol w(\boldsymbol f)$
& Vectors of signal and query phases,
$\exp(i\boldsymbol\omega^{(x)}\cdot\boldsymbol x)$
and $\exp(i\boldsymbol\omega^{(f)}\cdot\boldsymbol f)$. \\

$\boldsymbol a(\boldsymbol\theta,\boldsymbol f)$
& Signal coefficient vector at fixed query:
$\boldsymbol a=\mathbf C\boldsymbol w$. \\

$\mathcal H_{\mathcal E_f}(\boldsymbol f)$
& Representable functions of the signal
at fixed query encoding and value. \\

$\mu$, $\mathbb P_{\mathcal E_f,\boldsymbol f}$
& Parameter distribution and the function
prior it induces. \\

$\Delta$, $S$, $D_\Sigma$
& Diagnostics of fixed-parameter shape change,
nonseparability, and function-prior change. \\

$\boldsymbol p_{\boldsymbol d}$,
$\boldsymbol f_{\boldsymbol k}$
& Sampling position and MNUDFT
evaluation-frequency vector. \\

$D$, $K$, $n$
& Numbers of signal samples and Fourier
evaluations, and their spatial dimension;
$N=Kn$. \\

$X_{\boldsymbol k}
(\boldsymbol x,\boldsymbol f_{\boldsymbol k})$
& MNUDFT coefficient at evaluation
frequency $\boldsymbol f_{\boldsymbol k}$. \\

$s^\star$, $G$
& Target score and the map producing it
from prescribed MNUDFT coefficients. \\

$\Omega_\star^{(f)}$
& Target Fourier modes in the query variable,
generated by the spectral task. \\

$F_{\boldsymbol\omega}$,
$F^\star_{\boldsymbol\nu}$
& Signal-dependent coefficient functions
in the circuit and target query expansions. \\

$\varepsilon_{\mathrm{ph}}$,
$\varepsilon_{\mathrm{coef}}$
& Phase- and coefficient-mismatch
contributions to the score-error bound. \\

$N_{\mathrm{train}}$
& Number of labeled training examples. \\

$f^\star_{\mathrm{geom}}$
& Scalar query derived from acquisition
geometry, specifying the leading azimuth--range
phase coupling in the SAR-like task. \\

\bottomrule
\end{tabularx}
\end{table}

\subsection*{Implementation}
\label{appx:qml-details}

Circuits are evaluated using noiseless statevector simulation.
Fixed random seeds control data generation, parameter initialization, and statistical resampling.
Within each learning comparison, models share data splits and paired parameter initializations and use matched optimization settings and training budgets to isolate the effect of query alignment.
Experiment-specific settings are given in the appendices; software, differentiation methods, and reproducible implementations are documented in the accompanying GitHub repository.

\section{Details on query-controlled inductive bias}

\subsection{Single data reuploading circuits}\label{appx:simple-pqc}
In the following, we derive the expectation value of a data reuploading parameterized quantum circuit using the method discussed in Ref.~\cite{schuldEffectDataEncoding2021}.
We consider parameterized quantum circuits~\cite{benedettiParameterizedQuantumCircuits2019,schuldEffectDataEncoding2021,perez-salinasDataReuploadingUniversal2020}, where the model output is given by the expectation value of an observable:
\begin{equation}
    v(\boldsymbol{x};\boldsymbol{\theta}) = \bra{0} U^\dagger(\boldsymbol{x};\boldsymbol{\theta}) O U(\boldsymbol{x};\boldsymbol{\theta}) \ket{0}, \label{eq:circ-exp-appx}
\end{equation}
with input vector $\boldsymbol{x} \in \mathbb{R}^M$. 
The circuit, which is illustrated in Fig.~\ref{fig:encoding1}, consists of $L$ layers with alternating data encoding and parameterized unitaries, 
\begin{equation}
    U(\boldsymbol{x};\boldsymbol{\theta}) =  V_{L}(\boldsymbol{\theta}) E_{L}(\boldsymbol{x}) \cdots V_1(\boldsymbol{\theta}) E_1(\boldsymbol{x}) , \label{eq:data-reupl-appx}
\end{equation}  
where $V_\ell(\boldsymbol{\theta})$ are parameterized unitaries and
$E_\ell(\boldsymbol{x})$ encodes the classical input as:
\begin{equation}
    E_\ell(\boldsymbol{x}) = \exp\left(-i \sum_{m=1}^{M} x_m  H_{\ell,m}\right),
\end{equation}
with $\{H_{\ell,m}\}$ denoting data-encoding Hamiltonians (e.g., Pauli operators).
Repeating $E_\ell(\boldsymbol{x})$ across $L$ layers is called data reuploading~\cite{perez-salinasDataReuploadingUniversal2020}.
We diagonalize the encoding Hamiltonians jointly so that each $H_{\ell,m}$ becomes diagonal: 
$\bar{H}_{\ell,m} = \mathrm{diag}(\lambda_{\ell,m,1}, \ldots, \lambda_{\ell,m,b})$, where $b=2^{n_q}$ denotes the dimension of the Hilbert space of the $n_q$ qubits on which the Hamiltonians act.
After absorbing any change-of-basis unitaries into the parameterized layers $V_{\ell}(\boldsymbol{\theta})$, we can assume, without loss of generality, that:
\begin{equation}
    E_\ell(\boldsymbol{x}) = \exp\left(-i \sum_{m=1}^{M} x_m    \bar{H}_{\ell,m} \right),
\end{equation}
where $\bar{H}_{\ell,m}$ is a diagonal matrix of eigenvalues for $H_{\ell,m}$.
The output amplitude at index $s$ of the quantum state can be written as:
\begin{align}
    \bra{s}U(\boldsymbol{x};\boldsymbol{\theta})\ket{0} &= \sum_{\boldsymbol{j} } \alpha_{s, \boldsymbol{j}}(\boldsymbol{\theta})  \exp\left(-i \sum_{m=1}^{M} x_m  \Lambda_{\boldsymbol{j},m}\right) \nonumber \\
    &=\sum_{\boldsymbol{j} } \alpha_{s, \boldsymbol{j}}(\boldsymbol{\theta})  \exp\left(-i \boldsymbol{x}  \cdot \boldsymbol{\Lambda}_{\boldsymbol{j}}\right) \label{eq:single-amplitude}
\end{align}
where $[b]:=\{1,\ldots,b\}$ and $\boldsymbol{j}=(j_1,\ldots,j_L)\in[b]^L$ indexes layer-wise eigenvalue paths.
Furthermore, $\Lambda_{\boldsymbol{j},m} = \sum_{\ell=1}^{L} \lambda_{\ell,m, j_\ell}$ is the accumulated eigenvalue in direction $m$, and $\alpha_{s,\boldsymbol{j}}(\boldsymbol{\theta})$ summarizes the contribution from parameterized unitaries and absorbed terms from the diagonalization of $H_{\ell,m}$.
The full expression for the expectation value (with observable $O$ and the Hermitian conjugate of $U(\boldsymbol{x};\boldsymbol{\theta})$) in Eq.~\eqref{eq:circ-exp} is then:
\begin{equation}
    v(\boldsymbol{x};\boldsymbol{\theta}) = \sum_{\boldsymbol{j}, \boldsymbol{h} \in [b]^L} \alpha_{\boldsymbol{j}, \boldsymbol{h}}(\boldsymbol{\theta})  \exp\left(i  \boldsymbol{\omega}_{\boldsymbol{j},\boldsymbol{h}}  \cdot  \boldsymbol{x} \right),
\end{equation}
 where $\boldsymbol{\omega}_{\boldsymbol{j},\boldsymbol{h}} = \boldsymbol{\Lambda}_{\boldsymbol{h}} - \boldsymbol{\Lambda}_{\boldsymbol{j}} \in \mathbb{R}^M$ denotes a frequency vector and $\alpha_{\boldsymbol{j},\boldsymbol{h}}(\boldsymbol{\theta})$ absorbs the contributions of parameterized unitaries and the observable.
By grouping terms with identical frequencies $\boldsymbol{\omega} \in \Omega$, the model output takes the form:
\begin{equation}
v(\boldsymbol{x};\boldsymbol{\theta}) = \sum_{\boldsymbol{\omega} \in \Omega} c_{\boldsymbol{\omega}}(\boldsymbol{\theta})   \exp\left(i \boldsymbol{\omega} \cdot \boldsymbol{x}\right), \label{eq:pqcGTP-appx}
\end{equation}
where $\Omega$ denotes the \textit{frequency spectrum} of the model, i.e., the set of all frequency vectors $\boldsymbol{\omega} \in \mathbb{R}^M$ that can appear in the expansion, $\{c_{\boldsymbol{\omega}}(\boldsymbol{\theta})\}$ are the
complex-valued coefficients determined by the circuit parameters, and
$\exp(i\boldsymbol{\omega}\cdot\boldsymbol{x})$ are the corresponding
basis functions.
This expression represents the expectation value $v(\boldsymbol{x};\boldsymbol{\theta})$ as a generalized trigonometric polynomial (GTP) in the input features $\boldsymbol{x}$~\cite{schuldEffectDataEncoding2021,manzanoApproximationGeneralizationCapacities2025}. 
The set $\Omega$ is fully determined by the eigenvalue spectra of the data-encoding Hamiltonians ${H_{\ell,m}}$: 
\begin{equation}
    \Omega = \left\{ \boldsymbol{\Lambda}_{\boldsymbol{h}} - \boldsymbol{\Lambda}_{\boldsymbol{j}} \middle| \boldsymbol{j}, \boldsymbol{h} \in [b]^L \right\}. \label{eq:spectrum-appx}
\end{equation}

Each frequency $\boldsymbol{\omega} \in \Omega$ arises as a sum of eigenvalue differences across circuit layers. 
In contrast, the parameterized unitaries (and any absorbed basis changes from diagonalization) affect only the coefficients $c_{\boldsymbol{\omega}}(\boldsymbol{\theta})$, thereby shaping the relative weight and phase of each spectral component.

\subsection{Double data reuploading circuits}
\label{appx:expanding-data-rep}

We now study a natural extension of the expectation value derivation to parameterized quantum circuits featuring two inputs, $\boldsymbol{x} \in \mathbb{R}^M$ and $\boldsymbol{f} \in \mathbb{R}^N$.
This setting is motivated by the structure of the MNUDFT in Eq.~\eqref{eq:multi-nudft}, which depends jointly on signal values and query frequencies.
In order to mirror this separation, we encode inputs in alternating layers, a scheme we denote as double data reuploading, illustrated in Fig.~\ref{fig:encoding2}.
Specifically, let the circuit consist of an even number $L$ of layers, where each odd-numbered layer applies a data-encoding unitary dependent on $\boldsymbol{x}$, and each even-numbered layer applies a data-encoding unitary dependent on $\boldsymbol{f}$.
Between each pair of data-encoding layers, an arbitrary parameterized unitary is applied. 
The full circuit takes the form:
\begin{equation}
    U(\boldsymbol{x}, \boldsymbol{f};\boldsymbol{\theta}) = \prod_{\ell=1}^{L/2} \left[ 
    V^{(x)}_{\ell}(\boldsymbol{\theta})
    E_\ell^{(x)}(\boldsymbol{x})
    V^{(f)}_{\ell}(\boldsymbol{\theta})
    E_\ell^{(f)}(\boldsymbol{f})       \right], \label{eq:unitaryDouble}
\end{equation}
where the signal encodings are
\begin{equation}
E^{(x)}_\ell(\bm x)=\exp\Bigl(-i\sum_{m=1}^{M}x_mH^{(x)}_{\ell,m}\Bigr),
\end{equation}
and the query encodings $E^{(f)}_\ell(\bm f)$ are defined analogously with generators $H^{(f)}_{\ell,n}$, $n=1,\dots,N$. 
We assume that the generators within each encoding layer commute pairwise, so that they admit a joint eigenbasis; commuting Pauli strings are a common choice. 
The basis change may differ between layers but is input independent and enters the Fourier coefficients together with the trainable blocks~\cite{schuldEffectDataEncoding2021}.
If both inputs use the same generators and $\bm f=\bm x$, the circuit reduces to a single data reuploading circuit of the form of Eq.~\eqref{eq:data-reupl}.

Again, we are interested in computing the expectation value of an observable $O$ with respect to the output of this PQC:
\begin{equation}
    v(\boldsymbol{x}, \boldsymbol{f};\boldsymbol{\theta}) = \bra{0} U^\dagger(\boldsymbol{x}, \boldsymbol{f};\boldsymbol{\theta}) O U(\boldsymbol{x}, \boldsymbol{f};\boldsymbol{\theta}) \ket{0}.
\end{equation}

To this end, we compute the output amplitude at index $s$ of the quantum state, following the approach of Eq.~\eqref{eq:single-amplitude}:

\begin{align}
    \bra{s}U(\boldsymbol{x},\boldsymbol{f};\boldsymbol{\theta})\ket{0} &= \sum_{\boldsymbol{j}^{(x)}, \boldsymbol{j}^{(f)} } \alpha_{s, \boldsymbol{j}^{(x)}, \boldsymbol{j}^{(f)}} (\boldsymbol{\theta}) \nonumber \\
      &\cdot\exp\left(-i \left[\boldsymbol{x} \cdot \boldsymbol{\Lambda}^{(x)}_{\boldsymbol{j}^{(x)}} + \boldsymbol{f}  \cdot \boldsymbol{\Lambda}^{(f)}_{\boldsymbol{j}^{(f)}}\right]\right) 
\end{align}
where $\boldsymbol{j}^{(\cdot)} = (j^{(\cdot)}_1, \dots, j^{(\cdot)}_{L/2})\in [b]^{L/2}$ denotes the set of any $L/2$ integers between $1$ and $b$ and indexes eigenvalue sequences.
Furthermore, $\boldsymbol{\Lambda}^{(\cdot)}_{\boldsymbol{j}^{(\cdot)}}$ is composed of $\Lambda_{\boldsymbol{j}^{(\cdot)},m}^{(\cdot)} = \sum_{\ell=1}^{L/2} \lambda^{(\cdot)}_{\ell,m,j_\ell^{(\cdot)}}$, which are the accumulated eigenvalues in direction $m$.
Finally, $\alpha_{s,\boldsymbol{j}^{(x)},\boldsymbol{j}^{(f)}}(\boldsymbol{\theta})$ summarizes the contributions from parameterized unitaries and absorbed terms from the diagonalization of the data encoding.
The full expression for the expectation value of the circuit is then:

\begin{align}
    v(\boldsymbol{x}, \boldsymbol{f};\boldsymbol{\theta}) &= \sum_{\boldsymbol{\omega}^{(x)}, \boldsymbol{\omega}^{(f)}} c_{\boldsymbol{\omega}^{(x)}, \boldsymbol{\omega}^{(f)}}(\boldsymbol{\theta}) \nonumber \\
    &\cdot \exp \left({i \left[\boldsymbol{\omega}^{(x)} \cdot \boldsymbol{x} + \boldsymbol{\omega}^{(f)} \cdot \boldsymbol{f}\right]}\right),
\end{align}
for some set of frequencies $\boldsymbol{\omega}^{(x)} \in \Omega^{(x)} \subset \mathbb{R}^{M}$, $\boldsymbol{\omega}^{(f)} \in \Omega^{(f)} \subset \mathbb{R}^{N}$, and complex coefficients $ c_{\boldsymbol{\omega}^{(x)}, \boldsymbol{\omega}^{(f)}}(\boldsymbol{\theta})$.

Finally, grouping terms with the same $\boldsymbol{\omega}^{(f)}$ frequency yields
\begin{equation}
    v(\boldsymbol{x}, \boldsymbol{f};\boldsymbol{\theta}) 
    = \sum_{\boldsymbol{\omega}^{(f)}} 
    F_{\boldsymbol{\omega}^{(f)}}(\boldsymbol{x};\boldsymbol{\theta})  
    \exp \left(i \boldsymbol{\omega}^{(f)} \cdot \boldsymbol{f}\right),
    \label{eq:finalPqc-appx}
\end{equation}
with the parameterized coefficient function 
\begin{equation}
    F_{\boldsymbol{\omega}^{(f)}}(\boldsymbol{x};\boldsymbol{\theta}) 
    = \sum_{\boldsymbol{\omega}^{(x)}} 
    c_{\boldsymbol{\omega}^{(x)}, \boldsymbol{\omega}^{(f)}} (\boldsymbol{\theta})  
    \exp \left(i \boldsymbol{\omega}^{(x)} \cdot \boldsymbol{x}\right).
\end{equation}
Equation~\eqref{eq:finalPqc-appx} shows that the expectation value takes the form of a GTP in $\boldsymbol{f}$, with coefficients that are themselves GTPs in $\boldsymbol{x}$.

\subsection{Proofs of propositions}

\subsubsection{Proof of spectral separation}\label{appx:proof-spectral-sep}

\begin{proof}[Proof of Proposition~\ref{thm:frequency-independence}]

Let $\mathcal L_x$ and $\mathcal L_f$ denote the disjoint sets of signal- and query-encoding layers. The signal and query frequency sets are
\begin{equation}
\Omega^{(x)}=\bigl\{\Lambda^{(x)}_{\bm h}-\Lambda^{(x)}_{\bm j}:\ \bm h,\bm j\in[b]^{|\mathcal L_x|}\bigr\}, \label{eq:appx-signal-frequencies}
\end{equation}
and analogously for $\Omega^{(f)}$. Because $\mathcal L_x$ and $\mathcal L_f$
are disjoint, any choice of signal path indices can be combined with
any choice of query path indices to form a full path. 
The same holds
independently for the second path in a path pair. 
Every pair in
$\Omega^{(x)}\times\Omega^{(f)}$ therefore occurs as a formal
difference of full path eigenvalues, and every such full-path
difference belongs to this product. 
This proves
Eq.~\eqref{eq:factorize}.

For fixed $\boldsymbol f$, grouping
Eq.~\eqref{eq:joint-ddr-gtp} by signal frequency expresses every
$v(\,\cdot\,,\boldsymbol f;\boldsymbol\theta)$ as a linear
combination of the basis functions indexed by $\Omega^{(x)}$.
Hence
$\mathcal H_{\mathcal E_f}(\boldsymbol f)\subseteq\mathcal V_x$.
\end{proof}

The factorization concerns the frequency set generated by the
eigenvalue paths. 
It does not assert that every pair has a nonzero
coefficient for a restricted trainable architecture, initial state,
and observable.

\subsubsection{Proof of the modulation bound}\label{appx:proof-modu-bound}

\begin{proof}[Proof of Proposition~\ref{prop:query-modulation}]
By Eq.~\eqref{eq:encoding-dependent-coefficients},
\begin{align}
\operatorname{span}_{\mathbb C}
\left\{
\boldsymbol a(\boldsymbol\theta,\boldsymbol f):
\boldsymbol f\in\mathcal F
\right\}
&=
\operatorname{span}_{\mathbb C}
\left\{
\mathbf C(\boldsymbol\theta)\boldsymbol w(\boldsymbol f):
\boldsymbol f\in\mathcal F
\right\}
\nonumber\\
&=
\mathbf C(\boldsymbol\theta)\mathcal W.
\label{eq:appx-modulation-image}
\end{align}
The dimension of this image cannot exceed either the rank of
$\mathbf C(\boldsymbol\theta)$ or the dimension of its domain
$\mathcal W$, proving Eq.~\eqref{eq:rank-bound}.
If $\mathcal W=\mathbb C^{|\Omega^{(f)}|}$, its image is the full column space of $\mathbf C(\boldsymbol\theta)$, so the dimension equals the matrix rank.
\end{proof}

For distinct query frequencies, the full-span condition holds when
$\mathcal F$ contains a nonempty open subset of $\mathbb R^N$.
Indeed, a linear functional annihilating every
$\boldsymbol w(\boldsymbol f)$ on such a set would define a finite sum of distinct exponentials that vanishes on an open set. 
Linear independence of these exponentials forces every coefficient of that functional to vanish. 
The bound also applies to restricted or discrete query domains, for which $\mathcal W$ may be smaller.

\subsubsection{Proof of query absorption}\label{appx:proof-q-absor}

\begin{proof}[Proof of Proposition~\ref{prop:absorption}]
Fix $\boldsymbol f$. At query layer $\ell$, the effective factor
$B_\ell(\boldsymbol\theta_\ell)
E_\ell^{(f)}(\boldsymbol f)$ ranges over
$\mathcal B_\ell E_\ell^{(f)}(\boldsymbol f)$.
Under Eq.~\eqref{eq:absorption-condition}, this set is
$\mathcal B_\ell$ for every query value. 
Because the adjacent blocks are parameterized independently, their effective factors can be chosen independently after this replacement. 
The set of realizable complete circuits, and therefore its image in the space of functions of $\boldsymbol x$, is independent of $\boldsymbol f$. 
This proves Eq.~\eqref{eq:absorbed-hypothesis-class}.

For the prior, right-multiplication invariance implies that the random factors
$B_\ell E_\ell^{(f)}(\boldsymbol f)$ and $B_\ell$ have the same distribution. 
Independence of these blocks from one another and from the remaining parameters makes the joint distribution of all effective circuit factors independent of
$\boldsymbol f$. 
Applying the circuit and observable to these factors is a deterministic map into functions of $\boldsymbol x$; its pushforward distribution is consequently independent of
$\boldsymbol f$. This proves
Eq.~\eqref{eq:absorbed-function-prior}.
\end{proof}

\subsection{Comparison with concatenated-input encodings}
\label{appx:concat}

A natural alternative to DDR is to feed the concatenated input $\boldsymbol{z} = \boldsymbol{x} \mathbin\Vert \boldsymbol{f} \in \mathbb{R}^{M+N}$ to a standard single-reuploading PQC. Following the same derivation as in
Sec.~\ref{sec:simple-pqc}, the expectation value takes the form
\begin{equation}
    v(\boldsymbol{z};\boldsymbol{\theta})
    = \sum_{\boldsymbol{\omega} \in \Omega}
      c_{\boldsymbol{\omega}}(\boldsymbol{\theta})\,
      \exp\!\left(i\,\boldsymbol{\omega} \cdot \boldsymbol{z}\right),
      \ \Omega \subset \mathbb{R}^{M+N}.
\end{equation}
Since $\boldsymbol{z} = \boldsymbol{x} \mathbin\Vert \boldsymbol{f}$, every frequency vector inherits the same block structure by concatenation:
$\boldsymbol{\omega} = \boldsymbol{\omega}^{(x)} \mathbin\Vert \boldsymbol{\omega}^{(f)}$
with $\boldsymbol{\omega}^{(x)} \in \mathbb{R}^M$ and
$\boldsymbol{\omega}^{(f)} \in \mathbb{R}^N$, so that
$\boldsymbol{\omega}\cdot\boldsymbol{z}
 = \boldsymbol{\omega}^{(x)}\cdot\boldsymbol{x}
 + \boldsymbol{\omega}^{(f)}\cdot\boldsymbol{f}$,
and the expansion becomes
\begin{equation}
    v(\boldsymbol{x},\boldsymbol{f};\boldsymbol{\theta})
    = \sum_{\boldsymbol{\omega}^{(x)},\boldsymbol{\omega}^{(f)}}
      c_{\boldsymbol{\omega}^{(x)},\boldsymbol{\omega}^{(f)}}(\boldsymbol{\theta})\,
      \exp\!\left(i\bigl[
        \boldsymbol{\omega}^{(x)}\cdot\boldsymbol{x}
        +\boldsymbol{\omega}^{(f)}\cdot\boldsymbol{f}
      \bigr]\right),
\end{equation}
which matches the functional form of Eq.~\eqref{eq:joint-ddr-gtp}.
The structural difference lies in how the frequencies are generated.
Each frequency arises as an eigenvalue path difference
\begin{equation}
    \boldsymbol{\omega}
    = \boldsymbol{\Lambda}_{\boldsymbol{h}} - \boldsymbol{\Lambda}_{\boldsymbol{j}},
    \quad
    \boldsymbol{\Lambda}_{\boldsymbol{j}}
    = \Bigl(\sum_{\ell=1}^{L}\lambda_{\ell,m,\,j_\ell}\Bigr)_{m=1}^{M+N},
\end{equation}
where the path index $j_\ell \in [b]$ is \emph{shared} across all $M+N$ coordinates. 
Consequently, the $\boldsymbol{x}$- and $\boldsymbol{f}$-blocks of 
$\boldsymbol{\omega}$ are determined by the same path choice
$(h_\ell, j_\ell)_{\ell=1}^{L}$:
\begin{align}
    \boldsymbol{\omega}^{(x)}
    &= \Bigl(\sum_{\ell}\bigl(\lambda_{\ell,m,h_\ell}
            -\lambda_{\ell,m,j_\ell}\bigr)\Bigr)_{m=1}^{M}, \\
    \boldsymbol{\omega}^{(f)}
    &= \Bigl(\sum_{\ell}\bigl(\lambda_{\ell,m,h_\ell}
            -\lambda_{\ell,m,j_\ell}\bigr)\Bigr)_{m=M+1}^{M+N}.
\end{align}
Fixing a path simultaneously fixes both frequency components.
For a given $\omega^{(x)}\in\Omega^{(x)}$, the compatible query frequencies are determined by the joint eigenvalue paths.
Consequently, $\Omega\subseteq\Omega^{(x)}\times\Omega^{(f)}$, with strict inclusion possible:  
concatenated encodings do not guarantee Cartesian factorization.
Disjoint-register generators provide a sufficient condition for Cartesian factorization. 
For shared-support generators, factorization depends on the joint eigenspectrum and is not guaranteed.

For example, the single-qubit encoding
\begin{equation}
E(x,f)=\exp\!\bigl[-\tfrac{i}{2}(x+f)Z\bigr]
\end{equation}
has joint frequency set $\{(-1,-1),(0,0),(1,1)\}$, although its projections onto the $x$ and $f$ axes are each $\{-1,0,1\}$.
Disjoint layers allow the path indices for the two inputs to be selected independently, yielding Eq.~\eqref{eq:factorize}. 
In either construction, the frequencies
that actually appear in the output can be a subset of the formal set when their coefficients vanish.

\subsection{Details on query experiments}
\label{appx:controlled-query}

Both rows of Fig.~\ref{fig:encoding-value-interpolation} use the same circuit family, parameter priors, Fourier extraction, and diagnostics.
They differ only in what is varied: the top row changes the query encoding $\mathcal E_f$ at fixed $\boldsymbol f$
(Appx.~\ref{appx:encoding-interpolation}), and the bottom row changes
$\boldsymbol f$ at fixed encoding (Appx.~\ref{appx:vary-query}).
Both comparisons hold the signal- and query-frequency sets fixed, so any observed effect arises from the coefficients alone. 
No parameter optimization or training is performed.

\paragraph*{Circuit family.}
We consider two qubits initialized in $\ket{++}$ and $L/2=2$ blocks. 
In each block, the signal is encoded as
\begin{equation}
E^{(x)}_\ell(\boldsymbol x)=R_Z(x_0)\otimes R_Z(x_1),
\qquad
R_P(\phi)=\exp(-i\phi P/2),
\label{eq:controlled-signal-encoding}
\end{equation}
and the complete circuit is
\begin{equation}
U(\boldsymbol x,\boldsymbol f;\boldsymbol\theta)
=
\prod_{\ell=1}^{2}
V^{(f)}_\ell(\boldsymbol\theta)\,
E^{(f)}_\ell(\boldsymbol f)\,
V^{(x)}_\ell(\boldsymbol\theta)\,
E^{(x)}_\ell(\boldsymbol x),
\label{eq:controlled-unitary}
\end{equation}
with factors of larger $\ell$ to the left. The output is the
correlator
\begin{equation}
v(\boldsymbol x,\boldsymbol f;\boldsymbol\theta)
=
\bra{++}U^\dagger Z_0Z_1\,U\ket{++}.
\label{eq:controlled-output}
\end{equation}
All trainable blocks are products of single-qubit unitaries, each parameterized as $R_Z(\phi)R_Y(\beta)R_Z(\gamma)$, giving $24$ parameters in total.

\paragraph*{Query encodings.}
Each query layer encodes both components of $\boldsymbol f=(f_0,f_1)$,
\begin{equation}
E^{(f)}_\ell(\alpha;\boldsymbol f)
=
\prod_{n=0}^{1}
\exp\!\left[-i f_n H^{(f)}_{\ell,n}(\alpha)\right],
\label{eq:encoding-interpolation-query}
\end{equation}
with generators interpolated as
\begin{equation}
H^{(f)}_{\ell,n}(\alpha)
=\frac{s_\ell}{2}\bigl[\cos(\alpha)A_n+\sin(\alpha)B_n\bigr],
\qquad (s_1,s_2)=(1,2),
\label{eq:controlled-generators}
\end{equation}
$\alpha\in[0,\pi/2]$, and common starting point
$(A_0,A_1)=(X_0,X_1)$. 
The two paths differ in their endpoints:
\begin{align}
\text{local:}\quad &(B_0,B_1)=(Y_0,Y_1),
\label{eq:encoding-local-path}\\
\text{nonlocal:}\quad &(B_0,B_1)=(Y_0Z_1,\,Z_0Y_1).
\label{eq:encoding-nonlocal-path}
\end{align}
The two generators within a layer commute on both paths, so the ordering in Eq.~\eqref{eq:encoding-interpolation-query} is immaterial.
The distinct nonlocal targets retain separate dependence on $f_0$ and $f_1$; choosing $Z_0Z_1$ for both would collapse the query layer to a single rotation in $f_0+f_1$. 
The query-value experiment uses
the two fixed encodings at the ends of these paths: the common local start $(X_0,X_1)$ at $\alpha=0$, and the nonlocal endpoint
$(Y_0Z_1,Z_0Y_1)$ at $\alpha=\pi/2$.

\paragraph*{Preserved frequency sets.}
For each $n$, $A_n$ and $B_n$ are anticommuting Pauli strings, so
$[H^{(f)}_{\ell,n}(\alpha)]^2=(s_\ell/2)^2 I$. 
The commuting pair
$H^{(f)}_{\ell,0},H^{(f)}_{\ell,1}$ has joint eigenvalues
$(s_\ell/2)(\pm1,\pm1)$, each occurring once, independently of the path and of $\alpha$. 
Each layer therefore contributes the frequency differences $\{-s_\ell,0,s_\ell\}$ per query coordinate, and the two layers give
\begin{equation}
\Omega^{(f)}=\{-3,-2,\ldots,3\}^{2}.
\label{eq:controlled-f-modes}
\end{equation}
The two repetitions of Eq.~\eqref{eq:controlled-signal-encoding} give
\begin{equation}
\Omega^{(x)}=\{-2,-1,0,1,2\}^{2}.
\label{eq:controlled-x-modes}
\end{equation}
Both sets are identical for every encoding considered below.

\paragraph*{Hypothesis classes and priors.}
Every local encoding, including the common starting point, consists of product query encoding gates. 
These are absorbed by the adjacent product blocks, each of whose single-qubit factors ranges over all of
$\mathrm{SU}(2)$. 
By Proposition~\ref{prop:absorption}, the hypothesis class is therefore the same for every local encoding and every query value. The nonlocal query encoding gates are generically nonproduct and cannot be absorbed into an adjacent product block.

We use two parameter priors. 
Under the Haar prior, each single-qubit factor is sampled independently from the Haar measure on
$\mathrm{SU}(2)$; right invariance then makes absorption preserve the function prior as well. 
Under the narrow prior, the Euler angles are sampled independently from $\text{normal}(0,0.35^2)$. 
This untruncated Gaussian has full support, so the reachable set is unchanged, but it is not right invariant, and absorption can change the function prior. 
The three series in Fig.~\ref{fig:encoding-value-interpolation} are
\begin{enumerate}
\item[(I)] local encoding, Haar prior: neither the hypothesis class nor
the function prior changes;
\item[(II)] local encoding, narrow prior: the hypothesis class is
fixed, but the function prior can change;
\item[(III)] nonlocal encoding, Haar prior: both can change.
\end{enumerate}

\paragraph*{Fourier extraction.}
The output is sampled on the grid $x_{n,a}=2\pi a/N_{\rm g}$, $a=0,\dots,N_{\rm g}-1$, with $N_{\rm g}=8$, and the signal coefficients are
\begin{align}
A_{\omega_0,\omega_1}
=\frac{1}{N_{\rm g}^2}\sum_{a,b=0}^{N_{\rm g}-1}
&v(x_{0,a},x_{1,b},\boldsymbol f;\boldsymbol\theta)\nonumber \\
&\times \exp(-i(\omega_0x_{0,a}+\omega_1x_{1,b})).
\label{eq:controlled-coefficients}
\end{align}
Retaining $\omega_0,\omega_1\in\{-2,\ldots,2\}$ gives the
$5\times5$ matrix $\mathbf A$ of the main text, with
$\boldsymbol a=\operatorname{vec}\mathbf A$; $N_{\rm g}=8$ resolves these modes without aliasing.

\paragraph*{Diagnostics.}
We use $\Delta$, $S$, and $D_\Sigma$ as defined in Eq.~\eqref{eq:diagnostics}, with the baseline at $\alpha=0$ or $t=0$. 
The shape change $\Delta$ is evaluated for one fixed parameter realization per series.

At the baseline, the initial state, circuit, and observable all factorize across the two qubits, so every function in the baseline class is separable across $x_0|x_1$ and has a rank-one coefficient matrix.

Hence $S$ vanishes on the baseline class, and $S>0$ exhibits a function outside it, implying $\mathcal H_{\mathcal E_f}(\bm f)\neq\mathcal H_0$, where $\mathcal H_0$ and $P_0$ denote the hypothesis class and function prior at the baseline ($\alpha=0$ or $t=0$).
The maximum over $\boldsymbol\theta$ is appropriate because class membership is an existence statement; averaging would instead probe the prior. 
We estimate it over $1000$ parameter samples, so finite sampling can only underestimate $S$.

Because $\boldsymbol\Sigma$ is a functional of the function prior, $D_\Sigma>0$ implies $\mathbb P\neq\mathbb P_0$, whereas $D_\Sigma=0$ does not establish equality of the full priors. 
We estimate $\boldsymbol\Sigma$ from $5000$ independent parameter samples, using a ten-block jackknife for bias correction and standard errors. 
Series~(I), whose invariance follows analytically from Haar right invariance, defines the finite-sample null; its upper $95\%$ percentile is shown as the shaded band.

\subsubsection{Varying the query encoding}
\label{appx:encoding-interpolation}

The top row of Fig.~\ref{fig:encoding-value-interpolation} fixes
$\boldsymbol f=(\pi/4,\pi/3)$ and evaluates $181$ equally spaced
values of $\alpha\in[0,\pi/2]$ along both paths
(Fig.~\ref{fig:encoding-interpolation-circuit}).
Here $\alpha$ labels circuit architectures rather than inputs. At this $\boldsymbol f$, the nonlocal query encoding gates are nonproduct for generic $\alpha>0$. 
As a check on Eq.~\eqref{eq:controlled-f-modes}, we additionally evaluate the joint Fourier expansion over $(f_0,f_1)$ at five values of $\alpha$: 
for generic parameters, all $49$ query modes are occupied on both paths, with no weight outside this set.

\begin{figure}[t!]
    \centering
    \begin{adjustbox}{width=\columnwidth}
    \begin{quantikz}
    \lstick{$\ket{+}_0$}
        & \gate{R_Z(x_0)} \gategroup[
          wires=2,
          steps=2,
          style={dashed,rounded corners,fill=myblue!80},
          background,
          label style={
              label position=above,
              anchor=south,
              yshift=0.0cm
          }
      ]{}
        & \gate{V_1^{(x)}} 
        & \gate[2]{E_1^{(f)}(\alpha;\boldsymbol f)}\gategroup[
          wires=2,
          steps=2,
          style={dashed,rounded corners,fill=myred!80},
          background,
          label style={
              label position=above,
              anchor=south,
              yshift=0.0cm
          }
      ]{}
        & \gate{V_1^{(f)}}
        & \gate{R_Z(x_0)}\gategroup[
          wires=2,
          steps=2,
          style={dashed,rounded corners,fill=myblue!80},
          background,
          label style={
              label position=above,
              anchor=south,
              yshift=0.0cm
          }
      ]{}
        & \gate{V_2^{(x)}}
        & \gate[2]{E_2^{(f)}(\alpha;\boldsymbol f)}
        \gategroup[
          wires=2,
          steps=2,
          style={dashed,rounded corners,fill=myred!80},
          background,
          label style={
              label position=above,
              anchor=south,
              yshift=0.0cm
          }
      ]{}
        & \gate{V_2^{(f)}}
        & \meter{} \\
    \lstick{$\ket{+}_1$}
        & \gate{R_Z(x_1)}
        & \gate{V_1^{(x)}}
        & \qw
        & \gate{V_1^{(f)}}
        & \gate{R_Z(x_1)}
        & \gate{V_2^{(x)}}
        & \qw
        & \gate{V_2^{(f)}}
        & \meter{}
    \end{quantikz}
    \end{adjustbox}
    \caption{
    \textbf{Circuit used in the query-encoding interpolation experiment.}
    Each of the two DDR blocks encodes
    $\boldsymbol x=(x_0,x_1)$, applies product trainable blocks
    $V^{(x)}_\ell$ and $V^{(f)}_\ell$, and encodes both components of the fixed query $\boldsymbol f=(\pi/4,\pi/3)$.
    The composite query layer $E_\ell^{(f)}$ contains the two commuting rotations defined in Eq.~\eqref{eq:encoding-interpolation-query}.
    The local path interpolates from $(X_0,X_1)$ to $(Y_0,Y_1)$, whereas the nonlocal path interpolates from
    $(X_0,X_1)$ to $(Y_0Z_1,Z_0Y_1)$, as defined in Eqs.~\eqref{eq:encoding-local-path} and \eqref{eq:encoding-nonlocal-path}.
    The output is the correlator $\langle Z_0Z_1\rangle$.
    }
    \label{fig:encoding-interpolation-circuit}
\end{figure}
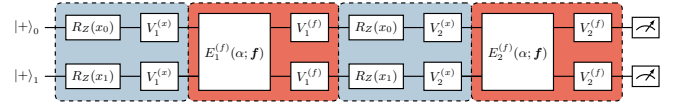

\emph{Results.} 
The shape change becomes large on all three series, with endpoint values $\Delta=0.938$, $0.998$, and $0.970$ for (I)--(III). 
A fixed parameter vector thus produces a substantially different function even when class and prior are invariant. The class witness is exactly zero along (I) and (II), where the full circuit remains a product, and reaches $S=0.66$ on (III). 
The prior distance leaves the Haar null for both (II) and (III), with endpoint values $D_\Sigma=0.37$ and $0.068$. 
Series~(II) is the essential separation:
its prior changes despite the exact class invariance of the local path.

\subsubsection{Varying the query value}
\label{appx:vary-query}

The bottom row of Fig.~\ref{fig:encoding-value-interpolation} fixes the
encoding at one end of the paths above,
\begin{equation}
E^{(f)}_\ell(t)\equiv E^{(f)}_\ell\bigl(\alpha_P;\boldsymbol f(t)\bigr),
\qquad \alpha_{\mathrm{loc}}=0,\ \ \alpha_{\mathrm{nl}}=\pi/2,
\label{eq:query-value-encoding}
\end{equation}
with $\alpha_{\rm P}\in\{\alpha_{\rm loc},\alpha_{\rm nl}\}$ for the local and nonlocal encodings, that is, generators $(X_0,X_1)$ for series (I) and (II) and
$(Y_0Z_1,Z_0Y_1)$ for series (III), and sweeps the diagonal path
\begin{equation}
\boldsymbol f(t)=t \bm 1,\qquad 0\le t\le 2\pi
\label{eq:query-value-ray}
\end{equation}
(Fig.~\ref{fig:query-circuit-architecture}). 
The horizontal axis is
therefore the path parameter $t$, not a scalar query. 
The two Haar series use the same parameter realizations, giving a paired comparison.

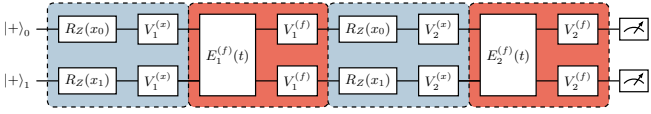
\begin{figure}[t!]
    \centering
    \begin{adjustbox}{width=\columnwidth}
    \begin{quantikz}
    \lstick{$\ket{+}_0$}
        & \gate{R_Z(x_0)}\gategroup[
          wires=2,
          steps=2,
          style={dashed,rounded corners,fill=myblue!80},
          background,
          label style={
              label position=above,
              anchor=south,
              yshift=0.0cm
          }
      ]{}
        & \gate{V_1^{(x)}}
        & \gate[2]{E_1^{(f)}(t)}\gategroup[
          wires=2,
          steps=2,
          style={dashed,rounded corners,fill=myred!80},
          background,
          label style={
              label position=above,
              anchor=south,
              yshift=0.0cm
          }
      ]{}
        & \gate{V_1^{(f)}}
        & \gate{R_Z(x_0)}\gategroup[
          wires=2,
          steps=2,
          style={dashed,rounded corners,fill=myblue!80},
          background,
          label style={
              label position=above,
              anchor=south,
              yshift=0.0cm
          }
      ]{}
        & \gate{V_2^{(x)}}
        & \gate[2]{E_2^{(f)}(t)}\gategroup[
          wires=2,
          steps=2,
          style={dashed,rounded corners,fill=myred!80},
          background,
          label style={
              label position=above,
              anchor=south,
              yshift=0.0cm
          }
      ]{}
        & \gate{V_2^{(f)}}
        & \meter{} \\
    \lstick{$\ket{+}_1$}
        & \gate{R_Z(x_1)}
        & \gate{V_1^{(x)}}
        & \qw
        & \gate{V_1^{(f)}}
        & \gate{R_Z(x_1)}
        & \gate{V_2^{(x)}}
        & \qw
        & \gate{V_2^{(f)}}
        & \meter{}
    \end{quantikz}
    \end{adjustbox}
    \caption{
    \textbf{Circuit used in the query-value experiment.}
    Each of the two DDR blocks encodes
    $\boldsymbol{x}=(x_0,x_1)$, applies product trainable blocks
    $V^{(x)}_\ell$ and $V^{(f)}_\ell$, and encodes the query along the diagonal path $\boldsymbol f(t)=t(1,1)$, with layer weights
    $(s_1,s_2)=(1,2)$.
    The composite query layer
    $E_\ell^{(f)}(t)\equiv
    E_\ell^{(f)}[\boldsymbol f(t)]$ contains the two commuting rotations defined in Eq.~\eqref{eq:query-value-encoding}.
    The local encoding uses generators $(X_0,X_1)$ and is absorbable into the trainable blocks, whereas the nonlocal encoding uses
    $(Y_0Z_1,Z_0Y_1)$ and is generically nonproduct.
    Series~(I) and~(II) use the same local circuit with different parameter priors; series~(III) uses the nonlocal encoding.
    The output is the correlator $\langle Z_0Z_1\rangle$.
    }
    \label{fig:query-circuit-architecture}
\end{figure}

\emph{Frequencies.} 
For any direction $\boldsymbol d$,
$\operatorname{spec}(d_0X_0+d_1X_1)
=\operatorname{spec}(d_0Y_0Z_1+d_1Z_0Y_1)=\{\pm d_0\pm d_1\}$, so both encodings share $\Omega^{(f)}$ in Eq.~\eqref{eq:controlled-f-modes}. Along the diagonal, this reduces to
$\Omega^{(t)}=\{-6,\ldots,6\}$. Numerically, all $49$ joint query modes are occupied for both encodings with no weight outside Eq.~\eqref{eq:controlled-f-modes}, and the coefficient matrices across the $13$ path modes both have numerical rank $13$. 
Neither encoding therefore accesses additional Fourier modes.

\emph{Extraction.} The output is sampled on an $8\times8\times14$
uniform grid over $(x_0,x_1,t)$, giving joint coefficients
$C_{\omega_0,\omega_1,\nu}(\boldsymbol\theta)$ of dimensions
$5\times5\times13$. The signal coefficients at query value $t$ are
\begin{equation}
A_{\omega_0,\omega_1}(t;\boldsymbol\theta)
=\sum_{\nu=-6}^{6}C_{\omega_0,\omega_1,\nu}(\boldsymbol\theta)\,\exp (i\nu t),
\label{eq:query-value-coefficients}
\end{equation}
the diagonal restriction of $\boldsymbol a=\mathbf C\boldsymbol w$ in Eq.~\eqref{eq:encoding-dependent-coefficients}. The diagnostics are evaluated at $181$ equally spaced values of $t$.

\emph{Product recurrences.}
At $t=0$ and $2\pi$, every query layer is proportional to the identity.
At $t=\pi$, the first query layer equals $-X_0X_1$ and the second the identity, up to global phases, for both encodings, so the local and nonlocal Haar series coincide there. 
These points give analytic zeros of $S$, and at $t=0,2\pi$ of all three diagnostics.

\emph{Results.} 
The shape change $\Delta(t)$ becomes large for every series, including (I). 
Both local series remain exactly separable, so $S(t)=0$ independently of the prior, whereas the nonlocal series reaches $S=0.695$, with exact zeros at $t=0,\pi,2\pi$. 
The Haar local series stays within its null threshold of $4.0\times10^{-3}$ for all $t$. 
The narrow-prior local series reaches $D_\Sigma=0.997$ despite exact class invariance (II), and the nonlocal series reaches $D_\Sigma=0.109$ (III), returning to the null at $t=\pi$.

Across both experiments, all three cases occur at fixed circuit size, parameter count, and signal- and query-frequency sets. 
Which case occurs is set by absorbability and prior invariance: the query acts as a pure reparameterization, reweights an unchanged class, or makes new functions accessible within the same Fourier space.

\section{Details on spectral task alignment}

\subsection{Real representation of complex signals}
\label{appx:real}

Angle-based PQC encodings require real inputs, whereas the signals in Sec.~\ref{sec:mnudft} may be complex.
We therefore represent
$x_{\boldsymbol d}
=
x_{\boldsymbol d}^{\mathrm R}
+
i x_{\boldsymbol d}^{\mathrm I}$
by
\begin{equation}
\widetilde{\boldsymbol x}
=
\left(
\operatorname{Re}\boldsymbol x,
\operatorname{Im}\boldsymbol x
\right)
\in\mathbb R^{2D},
\qquad
D=|\mathcal I|.
\label{eq:complex-real-map}
\end{equation}
This representation is bijective and therefore preserves all information in the complex signal.

To express the MNUDFT in these coordinates, define
\begin{equation}
\phi_{\boldsymbol k,\boldsymbol d}
=
2\pi\boldsymbol f_{\boldsymbol k}
\cdot\boldsymbol p_{\boldsymbol d}.
\end{equation}
Then
\begin{equation}
\begin{pmatrix}
\operatorname{Re}X_{\boldsymbol k}\\
\operatorname{Im}X_{\boldsymbol k}
\end{pmatrix}
=
\sum_{\boldsymbol d\in\mathcal I}
\begin{pmatrix}
\cos\phi_{\boldsymbol k,\boldsymbol d}
&
\sin\phi_{\boldsymbol k,\boldsymbol d}
\\
-\sin\phi_{\boldsymbol k,\boldsymbol d}
&
\cos\phi_{\boldsymbol k,\boldsymbol d}
\end{pmatrix}
\begin{pmatrix}
x_{\boldsymbol d}^{\mathrm R}\\
x_{\boldsymbol d}^{\mathrm I}
\end{pmatrix}.
\label{eq:real-mnudft}
\end{equation}
Thus, the real and imaginary parts of every MNUDFT coefficient are real linear functions of $\widetilde{\boldsymbol x}$.
Any decision map
$G:\mathbb C^K\rightarrow\mathbb R$ therefore has an equivalent representation in the corresponding real coordinates.

Moreover, both $x_{\boldsymbol d}$ and $x_{\boldsymbol d}^{*}$ are linear in
$(x_{\boldsymbol d}^{\mathrm R},x_{\boldsymbol d}^{\mathrm I})$.
Hence, polynomial decision maps in the MNUDFT coefficients and their complex conjugates remain polynomial in the real input components.
The real representation therefore preserves both the learning task and the query-dependent Fourier structure established in Sec.~\ref{sec:correspondence}.

\subsection{Proof of the MNUDFT score expansion}
\label{appx:mnudft-structure}

\begin{proof}[Proof of Proposition~\ref{prop:mnudft-structure}]
Using Eq.~\eqref{eq:lifted-phase-identity}, each MNUDFT evaluation can be written as
\begin{equation}
X_{\boldsymbol k}(\boldsymbol x,\boldsymbol f_{\boldsymbol k})
=
\sum_{\boldsymbol d\in\mathcal I}
x_{\boldsymbol d}
\exp(-2\pi i
\boldsymbol q_{\boldsymbol k,\boldsymbol d}\cdot\boldsymbol f).
\label{eq:appx-lifted-mnudft}
\end{equation}
Write a polynomial $G$ as a finite sum of monomials
\begin{equation}
G(\boldsymbol z,\boldsymbol z^*)
=
\sum_{\boldsymbol\alpha,\boldsymbol\beta}
g_{\boldsymbol\alpha,\boldsymbol\beta}
\prod_{\boldsymbol k\in\mathcal K}
z_{\boldsymbol k}^{\alpha_{\boldsymbol k}}
(z_{\boldsymbol k}^*)^{\beta_{\boldsymbol k}},
\label{eq:appx-score-polynomial}
\end{equation}
where the exponents are nonnegative integers. 
Substituting Eq.~\eqref{eq:appx-lifted-mnudft} and expanding the products, each factor $X_{\boldsymbol k}$ contributes a phase with frequency
$-2\pi\boldsymbol q_{\boldsymbol k,\boldsymbol d}$, while each
conjugate factor contributes
$+2\pi\boldsymbol q_{\boldsymbol k,\boldsymbol d}$.
Each term therefore has a query frequency of the form
\begin{equation}
\boldsymbol\omega
=
2\pi\left(
\sum_{\text{conjugate factors}}
\boldsymbol q_{\boldsymbol k,\boldsymbol d}
-
\sum_{\text{unconjugated factors}}
\boldsymbol q_{\boldsymbol k,\boldsymbol d}
\right)
\label{eq:appx-target-frequency}
\end{equation}
and an amplitude given by a polynomial in
$x_{\boldsymbol d}$ and $x_{\boldsymbol d}^*$.
Grouping terms with equal $\boldsymbol\omega$ yields Eq.~\eqref{eq:score-fourier-f}. For fixed $\boldsymbol x$, a real-valued score satisfies
$s^\star(\boldsymbol x\mid\boldsymbol f)
=s^\star(\boldsymbol x\mid\boldsymbol f)^*$.
Linear independence of distinct exponentials as functions of
$\boldsymbol f$ then gives
$F^\star_{-\boldsymbol\omega}
=(F^\star_{\boldsymbol\omega})^*$.

For the continuous case, let $C\subset\mathbb C^K$ be a compact set containing the MNUDFT coefficient vectors on the domain of interest. 
Polynomials in the coordinates $z_{\boldsymbol k}$ and their conjugates contain the constants, separate points of
$C$, and are closed under complex conjugation. By the
Stone--Weierstrass theorem~\cite{rudin_principles_1976}, for every
$\varepsilon>0$ there is such a polynomial $G_\varepsilon$
with
\begin{equation}
\sup_{\boldsymbol z\in C}
|G(\boldsymbol z)-G_\varepsilon(\boldsymbol z)|
<\varepsilon.
\end{equation}
Composing $G_\varepsilon$ with the MNUDFT evaluations gives a finite expansion of the form Eq.~\eqref{eq:score-fourier-f}.
The same bound holds uniformly over all signal and query pairs whose
MNUDFT coefficient vector lies in $C$.
\end{proof}

\subsection{Details of the spectral alignment experiment}
\label{appx:spectral-alignment}

Figure~\ref{fig:alignment} shows that task-spectrum alignment can be achieved together with a change in the function prior.
We consider a four-qubit circuit initialized in
$\ket{+}^{\otimes 4}$, with
$\boldsymbol x=(x_0,\ldots,x_3)$ encoded by
$\bigotimes_n R_Y(x_n)$ in each of two signal layers, so that each qubit carries one signal sample and one query coordinate.
Each trainable block consists of independently parameterized $R_Y$ rotations followed by a CNOT ring, and the circuit output is
$\langle Z_0\rangle$.
In contrast to Fig.~\ref{fig:encoding-value-interpolation}, the adjacent trainable blocks absorb $Y_n$ query rotations by shifts of their $R_Y$ angles, whereas generic $Z_n$ query rotations cannot be absorbed into these adjacent blocks.

The reference task is the MNUDFT intensity score
\begin{equation}
s(\boldsymbol x,\boldsymbol f)
=\sum_{n=0}^{3}\Bigl|\sum_{d=0}^{3}a_d\,\exp(-2\pi i p_d f_n)\Bigr|^2,
\end{equation}
for the phase-only signal $a_d=\exp(ix_d)$ sampled at the nonuniform positions $\boldsymbol p=(0,0.10,0.45,0.80)$, evaluated at the query
$\boldsymbol f^\star=(0.30,0.70,1.10,1.60)$.
Because the per-query intensities enter additively, every query-frequency mode involves a single coordinate $f_n$, and the spectrum is generated by the pairwise position differences,
$\Omega_\star^{(f)}=\{-2\pi(p_d-p_{d'})\boldsymbol e_n\}\cup\{\boldsymbol 0\}$, giving $41$ distinct modes with coordinate extent
$R=2\pi(p_{\max}-p_{\min})\simeq5.03$.

Panel~a evaluates the spectrum construction of Proposition~\ref{thm:realizability}.
For generator scales
$s_{\ell,n}=R_n2^{-(\ell+1)}$, where $R_n$ is the task extent in query coordinate $n$, the worst-case distance from
$\Omega_\star^{(f)}$ to the accessible spectrum
$\Omega^{(f)}$ halves with every query layer and remains below
$\sqrt{N}R\,2^{-L_f}$.
For comparison, uniform generators with
$s_{\ell,n}=R_n/L_f$ converge only polynomially.
The dyadic construction therefore provides an exponentially improving approximation of the task spectrum.

Panel~b tests whether this aligned branch produces a model-level effect.
At $L_f=2$, we evaluate the spectral second moment along
$\boldsymbol f=t\boldsymbol f^\star$ and compare the aligned $Z_n$ encoding with a $Y_n$ control.
The two encodings have identical generator eigenvalues and therefore the same $\Omega^{(f)}$, but the control belongs to the adjacent
$R_Y$ trainable family and is absorbed by a shift of its rotation angles.
This absorption holds to numerical precision, whereas the same shift does not absorb the aligned encoding.
Accordingly, the normalized second-moment distance remains at the sampling floor for the control but reaches
$0.345\pm0.010$ for the aligned encoding, showing that it changes the parameter-induced function prior.
The estimates use $1600$ shared parameter realizations and a jackknife correction for finite-sample bias; the shaded regions show
$95\%$ confidence intervals.

This experiment isolates the two architectural ingredients of task-alignment: the query encoding reproduces the task-relevant frequency structure and remains functionally active within the chosen trainable architecture.
This experiment establishes spectral alignment and a change in the function prior, but does not establish realizability of the target coefficient functions. Proposition~\ref{prop:exact-intensity-contrast} provides a separate exact construction, while the teacher--student experiments test learning within a realizable circuit family.

\begin{figure}[t!]
    \centering
    \begin{adjustbox}{width=\columnwidth}
    \begin{quantikz}
    \lstick{$\ket{+}_0$}
        & \gate{R_Y(x_0)}
          \gategroup[
              wires=4,steps=2,
              style={dashed,rounded corners,fill=myblue!80},
              background
          ]{}
        & \gate[4]{V_1^{(x)}}
        & \gate{R_Z(2s_{1,0}f_0)}
          \gategroup[
              wires=4,steps=2,
              style={dashed,rounded corners,fill=myred!80},
              background
          ]{}
        & \gate[4]{V_1^{(f)}}
        & \gate{R_Y(x_0)}
          \gategroup[
              wires=4,steps=2,
              style={dashed,rounded corners,fill=myblue!80},
              background
          ]{}
        & \gate[4]{V_2^{(x)}}
        & \gate{R_Z(2s_{2,0}f_0)}
          \gategroup[
              wires=4,steps=2,
              style={dashed,rounded corners,fill=myred!80},
              background
          ]{}
        & \gate[4]{V_2^{(f)}}
        & \meter{} \\
    \lstick{$\ket{+}_1$}
        & \gate{R_Y(x_1)}
        & \qw
        & \gate{R_Z(2s_{1,1}f_1)}
        & \qw
        & \gate{R_Y(x_1)}
        & \qw
        & \gate{R_Z(2s_{2,1}f_1)}
        & \qw
        & \qw \\
    \lstick{$\ket{+}_2$}
        & \gate{R_Y(x_2)}
        & \qw
        & \gate{R_Z(2s_{1,2}f_2)}
        & \qw
        & \gate{R_Y(x_2)}
        & \qw
        & \gate{R_Z(2s_{2,2}f_2)}
        & \qw
        & \qw \\
    \lstick{$\ket{+}_3$}
        & \gate{R_Y(x_3)}
        & \qw
        & \gate{R_Z(2s_{1,3}f_3)}
        & \qw
        & \gate{R_Y(x_3)}
        & \qw
        & \gate{R_Z(2s_{2,3}f_3)}
        & \qw
        & \qw
    \end{quantikz}
    \end{adjustbox}
    \caption{
    \textbf{Four-qubit DDR circuit used in the spectral-alignment experiment.}
    The signal is uploaded twice through local $R_Y(x_n)$ rotations.
    Each trainable block $V_\ell^{(x)}$ or $V_\ell^{(f)}$ consists of independent $R_Y$ rotations followed by a CNOT ring
    $0\to1\to2\to3\to0$.
    The query rotations implement
    $H_{\ell,n}^{(f)}=s_{\ell,n}Z_n$, with dyadic scales
    $s_{\ell,n}=R_n2^{-(\ell+1)}$.
    The isospectral control replaces each query $R_Z$ rotation by an $R_Y$ rotation with the same angle.
    The output is $\langle Z_0\rangle$.
    }
    \label{fig:encoding3}
\end{figure}
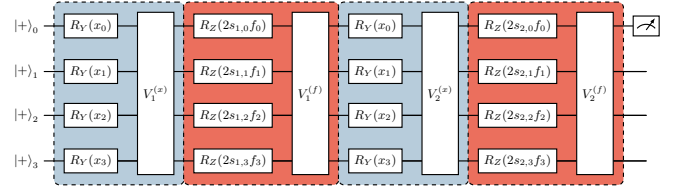

Figure~\ref{fig:encoding3} shows the four-qubit circuit used in this experiment. 
Data are uploaded by $R_Y$ angle encoding, with separate encoding stages for the signal $\bm x$ and the query $\bm f$. 
Each trainable block consists of independently parameterized $R_Y$ rotations followed by a periodic ring of nearest-neighbor CNOT gates.

\subsection{Derivation of the alignment bound}
\label{appx:realizability}    

For an injective matching
$\pi:\Omega_\star^{(f)}\to\Omega^{(f)}$, write the target and circuit
at the fixed task query as
\begin{equation}
s^\star(\boldsymbol x\mid\boldsymbol f^\star)
=
\sum_{\boldsymbol\nu\in\Omega_\star^{(f)}}
F^\star_{\boldsymbol\nu}(\boldsymbol x)
\exp (i\boldsymbol\nu\cdot\boldsymbol f^\star),
\label{eq:target-nested-alignment}
\end{equation}
and
\begin{equation}
v(\boldsymbol x,\boldsymbol f^\star;\boldsymbol\theta)
=
\sum_{\boldsymbol\nu\in\Omega_\star^{(f)}}
F_{\pi(\boldsymbol\nu)}(\boldsymbol x;\boldsymbol\theta)
\exp( i\pi(\boldsymbol\nu)\cdot\boldsymbol f^\star)
+r_{\boldsymbol\theta}(\boldsymbol x),
\label{eq:matched-circuit-expansion}
\end{equation}
where $r_{\boldsymbol\theta}$ collects the unused circuit modes.
Adding and subtracting
$F^\star_{\boldsymbol\nu}(\boldsymbol x)
\exp (i\pi(\boldsymbol\nu)\cdot\boldsymbol f^\star )$ for each target mode, then using
$\lvert \exp(ia)-\exp (ib)\rvert\leq\lvert a-b\rvert$
and Cauchy--Schwarz, gives Eqs.~\eqref{eq:alignment-error-decomposition}--\eqref{eq:coefficient-error} after taking the supremum over $\mathcal X$.

\subsection{Task-spectrum realizability}

Phase alignment, $\varepsilon_{\rm ph}\to0$, requires every task-defined mode in $\Omega^{(f)}_\star$ to be accessible, at least approximately, through the query encoding.
These modes are generated by the lifted sampling positions
$-2\pi\boldsymbol{q}_{\boldsymbol{k},\boldsymbol{d}}$ and their sums and differences, and are therefore determined by the sampling geometry and the decision map $G$.
Proposition~\ref{thm:frequency-independence} allows the query encoding to be designed for this spectrum without changing
$\Omega^{(x)}$.
The following Proposition gives an explicit choice of query-encoding generators that approximates any finite
$\Omega_{\star}^{(f)}$.

\begin{proposition}[Task-spectrum realizability]
\label{thm:realizability}
Let $\Omega_{\star}^{(f)}\subset\mathbb{R}^{N}$ be finite, with coordinate-wise extent
\begin{equation}
    R_n
    =
    \max_{\boldsymbol{\nu}\in\Omega_{\star}^{(f)}}|\nu_n|,
    \qquad n=1,\ldots,N.
\end{equation}
Consider $L_f$ query-encoding layers acting on $n_q=N$ qubits, one per query coordinate, with
\begin{equation}
    H_{\ell,n}^{(f)}
    =
    R_n2^{-(\ell+1)}Z_n,
    \qquad
    \ell=1,\ldots,L_f.
    \label{eq:realizable-generator-scales}
\end{equation}
Then, for every $\boldsymbol{\nu}\in\Omega_{\star}^{(f)}$, the accessible spectrum contains a mode
$\boldsymbol{\omega}^{(f)}\in\Omega^{(f)}$ satisfying
\begin{equation}
    \left\|
    \boldsymbol{\nu}-\boldsymbol{\omega}^{(f)}
    \right\|_2
    \leq
    2^{-L_f}\|\boldsymbol R\|_2,
    \qquad
    \boldsymbol R=(R_1,\ldots,R_N).
    \label{eq:task-spectrum-error}
\end{equation}
Hence, the complete target spectrum is approximated to accuracy
$\varepsilon$ whenever
$2^{-L_f}\|\boldsymbol R\|_2\leq\varepsilon$.
\end{proposition}

\begin{proof}
For an eigenvalue path $\boldsymbol j$, the accumulated query eigenvalue in coordinate $n$ is
\begin{equation}
    \Lambda_{\boldsymbol j,n}^{(f)}
    =
    \sum_{\ell=1}^{L_f}
    R_n2^{-(\ell+1)}z_{j_\ell,n},
    \qquad
    z_{j_\ell,n}\in\{-1,1\}.
\end{equation}
Consequently, every choice
$m_{\ell,n}\in\{-2,0,2\}$ is realized by a pair of paths
$\boldsymbol j,\boldsymbol h$, giving
\begin{equation}
    \omega_n^{(f)}
    =
    \Lambda_{\boldsymbol h,n}^{(f)}
    -
    \Lambda_{\boldsymbol j,n}^{(f)}
    =
    R_n\sum_{\ell=1}^{L_f}
    m_{\ell,n}2^{-(\ell+1)}.
\end{equation}

For $R_n>0$, choose binary digits $\beta_{\ell,n}\in\{0,1\}$ such that
\begin{equation}
    \left|
    \frac{|\nu_n|}{R_n}
    -
    \sum_{\ell=1}^{L_f}\beta_{\ell,n}2^{-\ell}
    \right|
    \leq 2^{-L_f},
\end{equation}
and set
$m_{\ell,n}=2\,\operatorname{sgn}(\nu_n)\beta_{\ell,n}$.
If $R_n=0$, then $\nu_n=0$ and we set $m_{\ell,n}=0$.
The resulting accessible mode obeys
\begin{equation}
    |\nu_n-\omega_n^{(f)}|
    \leq R_n2^{-L_f}.
\end{equation}
Taking the Euclidean norm over all coordinates yields Eq.~\eqref{eq:task-spectrum-error}.
\end{proof}

Proposition~\ref{thm:realizability} connects the query generators
to the target spectrum $\Omega_\star^{(f)}$, determined by the
sampling positions and the decision map.
For every target mode $\boldsymbol\nu\in\Omega_\star^{(f)}$,
choosing the query depth so that
$2^{-L_f}\|\boldsymbol R\|_2\leq\varepsilon$
guarantees a circuit mode
$\boldsymbol\omega^{(f)}\in\Omega^{(f)}$ satisfying
$
\left\|
\boldsymbol\omega^{(f)}-\boldsymbol\nu
\right\|_2
\leq\varepsilon$.
At the task-defined query $\boldsymbol f^\star$, the corresponding
DDR phase therefore satisfies
\begin{equation}
\left|
\exp\left(i\boldsymbol\omega^{(f)}\cdot\boldsymbol f^\star\right)
-
\exp\left(i\boldsymbol\nu\cdot\boldsymbol f^\star\right)
\right|
\leq
\varepsilon\|\boldsymbol f^\star\|_2.
\label{eq:mnudft-phase-alignment}
\end{equation}
For distinct target modes, the assignment is injective whenever
$2\varepsilon$ is smaller than their minimum pairwise separation.
Thus, the task supplies both the evaluation point
$\boldsymbol f^\star$ and, through its sampling positions and decision map, the target modes that guide the choice of query-encoding generators.

The construction establishes existence rather than resource optimality: 
it uses one qubit per query coordinate and covers the coordinate-wise extent of the complete target spectrum.
More compact encodings may instead exploit the particular structure of
$\Omega_{\star}^{(f)}$.

\begin{figure}[ht!]
    \centering
    \includegraphics[width=\linewidth]{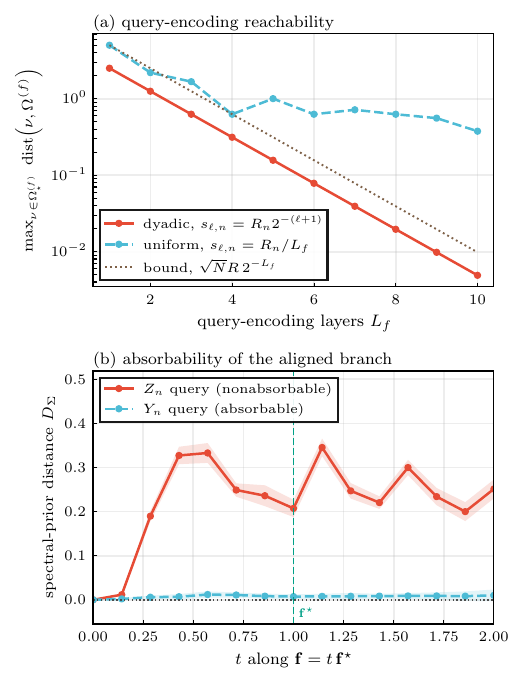}
    \caption{
    \textbf{A query encoding can match the task spectrum while changing the function prior.}
    We test these two requirements in a four-qubit, two-layer DDR circuit for the MNUDFT intensity task described in the text.
    \textbf{(a)}
    Worst-case distance from the task modes
    $\Omega_\star^{(f)}$ to the accessible encoding spectrum
    $\Omega^{(f)}$ as the number $L_f$ of query layers increases.
    The dyadic generator scales of Proposition~\ref{thm:realizability} halve the distance with every layer and satisfy the bound $\sqrt{N}R\,2^{-L_f}$, whereas uniformly spaced scales converge only polynomially.
    \textbf{(b)}
    Normalized distance $D_\Sigma$ between the spectral second moments at
    $\boldsymbol f=t\boldsymbol f^\star$ and
    $\boldsymbol f=\boldsymbol 0$ for $L_f=2$.
    The task-aligned $Z_n$ encoding is compared with a $Y_n$ control having the same generator eigenvalues and therefore the same $\Omega^{(f)}$.
    The control is absorbed by the adjacent $R_Y$ trainable blocks and
    remains at the sampling floor, whereas the aligned encoding produces a clear change in the parameter-induced function prior.
    Thus, matching the task spectrum and changing the function prior are compatible, while identical frequency support alone does not guarantee the latter.
    Shaded regions denote jackknife $95\%$ confidence intervals.
    Circuit, task, and estimator details are given in Appx.~\ref{appx:spectral-alignment}.}
    \label{fig:alignment}
\end{figure}

Figure~\ref{fig:alignment} shows that spectral alignment and a change in the function prior can be achieved together.
For a four-qubit MNUDFT intensity task with nonuniform sampling, panel~a illustrates Proposition~\ref{thm:realizability}: dyadic generator scales halve the mode error with each query layer, remaining below the bound $\sqrt{N}R\,2^{-L_f}$, whereas uniformly spaced scales converge only polynomially. 
Panel~b compares the aligned encoding with an absorbable control that has the same generator eigenvalues and hence the same $\Omega^{(f)}$. 
Along $\boldsymbol f=t\boldsymbol f^\star$, the aligned encoding changes the spectral second moment of the function prior, while the control remains invariant. 
These tests establish spectral support and a model-level query effect; they leave open whether the target’s signal-dependent coefficients are realizable.
We address that requirement next for a specific nonlinear score.

\subsection{Exact realization of an intensity contrast}
\label{appx:target-realization}

\begin{proposition}[Exact realization of an MNUDFT intensity contrast]
\label{prop:exact-intensity-contrast}
Let $D=2^m$ and let
$\boldsymbol p_d\in\mathbb R^n$, $d=0,\ldots,D-1$, be arbitrary sampling positions.
For the phase-only signal
\begin{equation}
x_d(\boldsymbol\varphi)
=
\exp(i\varphi_d),
\quad
\boldsymbol\varphi\in[0,2\pi)^D,
\end{equation}
define two MNUDFT coefficients
\begin{equation}
X_k(\boldsymbol\varphi,\boldsymbol f_k)
=
\sum_{d=0}^{D-1}
\exp(i\varphi_d)
\exp\left(
-2\pi i\,
\boldsymbol f_k\cdot\boldsymbol p_d
\right),
\label{eq:two-query-mnudft}
\end{equation}
with $k\in\{0,1\}$.
Then there exists an $(m+1)$-qubit DDR circuit with a hypothesis class that contains the normalized intensity-contrast score
\begin{equation}
s^\star(\boldsymbol\varphi)
=
\frac{
|X_0(\boldsymbol\varphi,\boldsymbol f_0^\star)|^2
-
|X_1(\boldsymbol\varphi,\boldsymbol f_1^\star)|^2
}{
2D^2
}.
\label{eq:target-intensity-contrast}
\end{equation}
More strongly, at a parameter value $\boldsymbol\theta^\star$, the circuit and target agree for every query pair:
\begin{equation}
v(\boldsymbol\varphi,\boldsymbol f;\boldsymbol\theta^\star)
=
\frac{
|X_0(\boldsymbol\varphi,\boldsymbol f_0)|^2
-
|X_1(\boldsymbol\varphi,\boldsymbol f_1)|^2
}{
2D^2
},
\label{eq:exact-promoted-equality}
\end{equation}
where
$\boldsymbol f=\boldsymbol f_0\mathbin\Vert\boldsymbol f_1$.
Consequently,
$\varepsilon_{\mathrm{ph}}
=
\varepsilon_{\mathrm{coef}}(\boldsymbol\theta^\star)
=0$.
\end{proposition}

\begin{proof}
Use one query-index qubit and an $m$-qubit data register with computational basis
$\{\ket d\}_{d=0}^{D-1}$.
Define
\begin{equation}
\Pi_d=\ket d\!\bra d,
\quad
\ket{+_D}
=
\frac{1}{\sqrt D}
\sum_{d=0}^{D-1}\ket d.
\end{equation}
A trainable preparation block produces
\begin{equation}
\ket{\chi(\theta)}
=
\cos\left(\frac{\theta}{2}\right)\ket0
+
\sin\left(\frac{\theta}{2}\right)\ket1
\end{equation}
on the query qubit, while the data register is initialized in
$\ket{+_D}$.
Encode the signal phases through
\begin{equation}
E^{(x)}(\boldsymbol\varphi)
=
I\otimes
\exp\left(
i\sum_d\varphi_d\Pi_d
\right)
\label{eq:exact-signal-encoding}
\end{equation}
and the query vectors through
\begin{equation}
E^{(f)}(\boldsymbol f)
=
\exp\left[
-2\pi i
\sum_{k=0}^{1}
\ket k\!\bra k
\otimes
\sum_d
\left(
\boldsymbol f_k\cdot\boldsymbol p_d
\right)
\Pi_d
\right].
\label{eq:exact-query-encoding}
\end{equation}
Both are valid commuting encodings of the form considered in
Sec.~\ref{sec:architecture}.
For example, the query-encoding Hamiltonian associated with component
$r$ of $\boldsymbol f_k$ is
\begin{equation}
H^{(f,k,r)}
=
2\pi
\ket k\!\bra k
\otimes
\sum_d p_{d,r}\Pi_d.
\end{equation}

After applying the two encoding layers, the state is
\begin{align}
\ket{\psi}
={}&
\frac{\cos(\theta/2)}{\sqrt D}
\sum_d
\exp(i\varphi_d-2\pi i\boldsymbol f_0\cdot\boldsymbol p_d)
\ket{0,d}
\nonumber\\
&+
\frac{\sin(\theta/2)}{\sqrt D}
\sum_d
\exp(i\varphi_d-2\pi i\boldsymbol f_1\cdot\boldsymbol p_d)
\ket{1,d}.
\label{eq:exact-intensity-state}
\end{align}
Consider the observable
\begin{equation}
O
=
Z\otimes
\ket{+_D}\!\bra{+_D}.
\label{eq:intensity-contrast-observable}
\end{equation}
Projecting the data register onto $\ket{+_D}$ coherently sums the phase-encoded samples. Its expectation value is therefore
\begin{equation}
v(\boldsymbol\varphi,\boldsymbol f;\theta)
=
\frac{\cos^2(\theta/2)}{D^2}
|X_0|^2
-
\frac{\sin^2(\theta/2)}{D^2}
|X_1|^2.
\label{eq:parameterized-intensity-contrast}
\end{equation}
Choosing
$\theta^\star=\pi/2$
gives Eq.~\eqref{eq:exact-promoted-equality}.
To see the agreement at the level of the nested expansion, expand the two intensities:
\begin{equation}
\begin{aligned}
v(\boldsymbol\varphi,\boldsymbol f;\theta^\star)
=
\frac{1}{2D^2}
&\sum_{k=0}^{1}\eta_k
\sum_{d,d'}
\exp\left[
i(\varphi_d-\varphi_{d'})
\right]
\\
&\times
\exp\left[
-2\pi i\,
\boldsymbol f_k\cdot
(\boldsymbol p_d-\boldsymbol p_{d'})
\right],
\end{aligned}
\label{eq:exact-intensity-nested-expansion}
\end{equation}
where $\eta_0=1$ and $\eta_1=-1$.
The target query modes are therefore
\begin{equation}
\boldsymbol\nu_{kdd'}
=
-2\pi
\boldsymbol e_k
\otimes
(\boldsymbol p_d-\boldsymbol p_{d'}).
\end{equation}
After collecting coincident modes, their coefficient functions are
\begin{equation}
F^\star_{\boldsymbol\nu}(\boldsymbol\varphi)
=
\frac{1}{2D^2}
\sum_{\substack{k,d,d'\\
\boldsymbol\nu_{kdd'}=\boldsymbol\nu}}
\eta_k\exp[i(\varphi_d-\varphi_{d'})].
\end{equation}
At $\bm \theta^*$, Eq.~\eqref{eq:exact-intensity-nested-expansion} shows the output contains only the target query modes, with the required coefficient functions; contributions from unused circuit modes vanish.

Thus, both errors in Eq.~\eqref{eq:alignment-error-decomposition} vanish.
\end{proof}

Proposition~\ref{prop:exact-intensity-contrast} makes the correspondence explicit.
The nonlinear decision map
\begin{equation}
G(X_0,X_1)
=
\frac{|X_0|^2-|X_1|^2}{2D^2}
\end{equation}
is realized exactly by a concrete DDR family, and a binary spectral task can be obtained by thresholding its sign. 
In this construction, the sample magnitudes are fixed, and the class information is carried by their relative phases and their coherent response at two prescribed nonuniform queries.
This construction is an exact containment result within the DDR family,
rather than a separation from concatenated encodings, since its commuting
encoding layers can be combined algebraically.
Nevertheless, within the specified trainable family, query changes that alter either intensity function change the hypothesis class.

This exact realization does not imply that a general DDR circuit computes an MNUDFT or that arbitrary spectral decision maps belong to its hypothesis class. 
It provides an explicit nontrivial case in which the task modes and their signal-dependent coefficients are simultaneously accessible.
The diagonal encodings may also require a number of elementary gates that grows with $D$, so the construction is an expressivity result rather than an efficiency claim. 
For general tasks and restricted circuit ansätze, the coefficient error in Eq.~\eqref{eq:coefficient-error} remains a learning problem.

\subsection{Details on teacher--student experiments }\label{appx:teacher-details}

Teachers and students share the two-qubit circuit
$U(\bm x,\bm f;\bm\theta)=\prod_{\ell=1}^{2}V^{(f)}_\ell\,E^{(f)}_\ell(\bm f)\,V_\ell^{(x)}\,E_\ell^{(x)}(\bm x)$,
with signal encoding $E_\ell^{(x)}=R_Z(x_0)\otimes R_Z(x_1)$, nonlocal query encoding gates
$E_\ell^{(f)}(\bm f)=R_{Z_0Y_1}(s_\ell f_1)\,R_{Y_0Z_1}(s_\ell f_0)$ with layer weights
$(s_1,s_2)=(1,2)$, initial state $|{+}{+}\rangle$ and observable $Z_0Z_1$.
The query encoding gates are those of series (III) in Fig.~\ref{fig:encoding-value-interpolation}. 
The trainable blocks are reduced to single-qubit rotations,
$V_\ell^{(x)}=R_Y(\theta_{\ell,0})\otimes R_Y(\theta_{\ell,1})$ and likewise for $V^{(f)}_\ell$, which gives $8$ trainable parameters. Because these blocks are products, the zero-query model is separable in $(x_0,x_1)$. 
We compare three fixed queries: 
the teacher query $\bm f^\star=(\pi/4,\pi/3)$, the mismatched query $2\bm f^\star$, and $\bm f=0$.

\textit{Teachers.} 
Teacher parameters are drawn from the student initialization distribution, $\theta_i\sim U[-\pi,\pi)$ independently. 
A draw is retained only if (i) the fraction of positive scores $v(\bm x,\bm f^\star;\bm\theta_T)>0$ on $2\times10^4$ probe inputs, drawn uniformly from $[0,2\pi)^2$, lies in $(0.4,0.6)$, and  (ii) the nonseparability fraction $1-\sigma_1^2/\sum_i\sigma_i^2$ of its $5\times5$ Fourier coefficient matrix at $\bm f^\star$ exceeds $0.2$. 
The coefficients are computed exactly by a two-dimensional FFT of $v$ on a $5\times5$ grid.
Condition (ii) ensures that the separable zero-query model cannot represent the continuous teacher score.
For each of $10$ teacher seeds we keep the first draw that passes; this took between $1$ and $5$ draws. 
The retained teachers have positive fractions between $0.42$ and $0.54$ and witnesses between $0.25$ and $0.60$. 
The probe inputs are used only for this selection. 
Labels are $y=1$ if $v(\bm x,\bm f^\star;\bm\theta_T)>0$ and $y=0$ otherwise.

\begin{figure*}[t!]
    \centering
    \includegraphics[width=\linewidth]{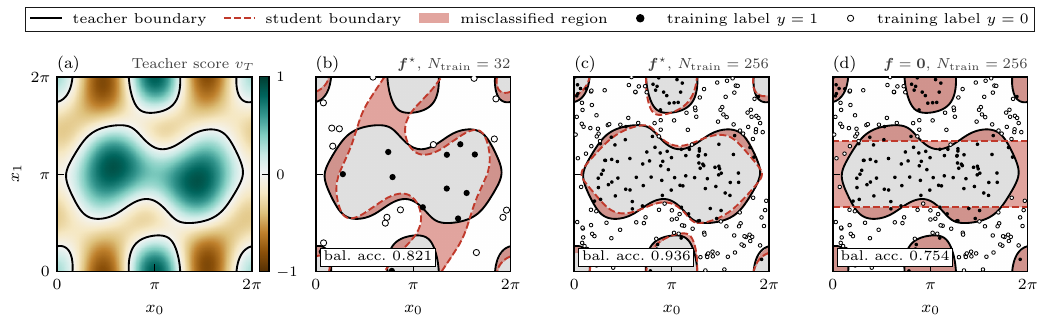}
    \caption{\textbf{Decision boundaries for one teacher:} (a)~Teacher score $v_T$ and its
    zero level set. (b--d)~Students trained on $N_{\rm train}$ labeled points
    (dots) with the correct query $\boldsymbol f^\star$ (b,c) or without a query,
    $\boldsymbol f=\boldsymbol 0$ (d). Solid line: teacher boundary; dashed line:
    student boundary; shading: misclassified region. In (c), $56\%$ of
    misclassified test points satisfy $|v_T|<0.05$, a region that contains
    $8.7\%$ of all test inputs. The separable student in (d) is restricted to
    axis-aligned boundaries.}
    \label{fig:decision-boundary}
\end{figure*}

\textit{Students.} 
A student predicts $p(y=1\mid\bm x)=\mathrm{sigmoid}[\beta\,v(\bm x,\bm f;\bm\theta)]$ with $\beta=5$ and minimizes the full-batch binary cross-entropy. 
Each fit starts from $n_{\rm init}$ initializations drawn from the prior above. 
These initializations depend only on the teacher, fold, and training-set size, so they are shared across all compared queries. 
Every start runs Adam with learning rate $0.03$, cosine-annealed to $0.003$, and gradients clipped to norm $10$. 
The $n_{\rm pol}$ starts with the lowest training loss are then polished with L-BFGS-B, and we keep the parameters with the lowest training loss. 
No test data enter any selection step. 
For Fig.~\ref{fig:teacher-student}a, we use $n_{\rm init}=16$, $2000$ Adam steps and $n_{\rm pol}=3$ with at most $500$ iterations. 
For each teacher we draw $8$ independent training-set draws. 
Each set has a pool of $256$ inputs drawn uniformly from
$[0,2\pi)^2$, and the training set of size
$N_{\rm train}\in\{8,16,\dots,256\}$ is the first $N_{\rm train}$ points of that pool. 
All folds and arms share one test set of $10^4$ inputs per teacher.
Balanced test accuracies are averaged over folds within each teacher.

\textit{Query scans.} 
Figure~\ref{fig:query-sensitivity} uses the same ten teachers with $N_{\rm train}=256$, $4$ folds (the same training sets as the first four folds of Fig.~\ref{fig:teacher-student}a), $2\times10^4$ test inputs per teacher, $n_{\rm init}=12$, $1500$ Adam steps, and $n_{\rm pol}=2$ with at most $400$ iterations. 
The student is retrained from the same initializations at every query. The scale scan in
Fig.~\ref{fig:query-sensitivity}a uses $\bm f=\lambda\bm f^\star$ with $\lambda\in\{0,\tfrac1{16},\dots,2\}$, giving $33$ query values. 
The local scan in Fig.~\ref{fig:query-sensitivity}b uses $\bm f=\bm f^\star+\delta(\cos\phi,\sin\phi)$ with $\phi\in\{0,\pi/4,\dots,7\pi/4\}$ and $11$ values of $\delta$ log-spaced between $10^{-4}$ and $0.5$. 
For each teacher, we take the minimum accuracy over directions, average these minima over teachers, and take the running minimum over $\delta$.

\textit{Realizability diagnostic.} 
Fig.~\ref{fig:teacher-student}b fits the continuous teacher score on an independent uniform $15\times15$ grid on $[0,2\pi)^2$. 
The fit minimizes the squared error with a trust-region least-squares solver and an exact Jacobian, using at most $1000$ function evaluations per start. 
We run $32$ starts per teacher and arm: the teacher's own parameters and $31$ prior draws, shared across arms. 
We report the relative mean-square error $\sum_i[v_S(\bm x_i)-v_T(\bm x_i)]^2/\sum_i v_T(\bm x_i)^2$ of the best start.
Including the teacher's parameters makes the $\bm f^\star$ arm an implementation check, which returns zero error. 
For the other arms the value is a best-found approximation error, not a lower bound.

\textit{Statistics and reproducibility.} 
Bands are 95\% percentile intervals from $5000$ bootstrap resamples over teachers, with folds averaged within each teacher and the pairing across arms and queries preserved. 
Circuits are simulated exactly in double precision with JAX.

\textit{Visualization of the learning task.}
Figure~\ref{fig:decision-boundary} shows the decision boundaries for one teacher and fold. 
The teacher's zero level set is not a simple split of the input torus; it consists of a central lobe and several pockets that wrap around the periodic boundaries. 
With $N_{\rm train}=32$ labels the student connects regions that the teacher separates, in gaps that no training point constrains. 
At $N_{\rm train}=256$ the residual errors lie in a thin band along the teacher boundary. 
The zero-query student, trained on the same points, cannot follow the boundary: 
its score factorizes as $g(x_0)h(x_1)$, so its decision boundary consists of straight lines.

\section{Details on geometry-aligned SAR classification}
\label{appx:sar-pipeline}

\subsection{Background: nonuniform Fourier transforms}\label{appx:mnudft}

Fourier transforms are a central component of signal processing, enabling the decomposition of data into frequency components~\cite{oppenheim1999discrete,cooley1965algorithm,briggs1995dft}.
The discrete Fourier transform (DFT) evaluates a complex-valued signal $\boldsymbol{x} \in \mathbb{C}^{D}$, with entries $x_d \in \mathbb{C}$, at $D$ uniformly spaced query frequencies $f_k = k-1$, yielding $X_k\in \mathbb{C}$ with $k \in \{1,\ldots,D\}$:
\begin{equation}
    \mathrm{DFT}(f_k,\boldsymbol{x}) \coloneqq  \sum_{d=1}^{D} x_d \exp\!\left(-2\pi i \, f_k \, p_d \right),
\end{equation}
where the sample positions $p_d = \tfrac{d-1}{D}$ are likewise uniformly spaced for $d \in \{1,\ldots,D\}$ and are defined by the sampling geometry.
In this formulation, $x_d$ acts as the coefficient weighting each basis function.
This illustrates that the DFT assumes regular sampling in both the spatial and frequency domains.
However, in many practical applications, such as radar imaging, medical diagnostics, or astronomy, data is acquired at irregular spatial positions or evaluated at nonuniform frequencies~\cite{dutt1993fast,fessler2003nonuniform}.
To account for this, the nonuniform discrete Fourier transform (NUDFT) generalizes the DFT by allowing both the sampling positions $p_d \in [0,1]$ for $d \in \{1,\ldots,D\}$ and the query frequencies $f_k \in \mathbb{R}$ for $k \in \{1,\ldots,K\}$ to be arbitrary:
\begin{align}
    \mathrm{NUDFT}(f_k,\boldsymbol{x}) \coloneqq  \sum_{d=1}^{D} x_d \exp\!\left(-2\pi i \, f_k \, p_d\right).
    \label{eq:nudft}
\end{align}
This extends to the multidimensional setting~\cite{bagchi2012nonuniform}, where we consider a tensor $\boldsymbol{x} \in \mathbb{C}^{D_1 \times \cdots \times D_n}$ with multi-index $\boldsymbol{d} \in \mathcal{I} \coloneqq \{1,\dots,D_1\} \times \cdots \times \{1,\dots,D_n\}$. 
Each entry $x_{\boldsymbol{d}} \in \mathbb{C}$ is associated with a sampling location $\boldsymbol{p}_{\boldsymbol{d}} \in [0,1]^n$.
The multidimensional NUDFT (MNUDFT) evaluates the signal $\boldsymbol{x}$ at $K$ query frequency vectors $\boldsymbol{f}_{\boldsymbol{k}} \in \mathbb{R}^n$, indexed by $\boldsymbol{k} \in \mathcal{K} \coloneqq \{1,\dots,K_1\} \times \cdots \times \{1,\dots,K_n\}$, as
\begin{equation}
    \mathrm{MNUDFT}(\boldsymbol{f}_{\boldsymbol{k}},\boldsymbol{x}) \coloneqq X_{\boldsymbol{k}}
    = \sum_{\boldsymbol{d}\,\in\,\mathcal{I}} x_{\boldsymbol{d}} \exp\!\left(-2\pi i\, \boldsymbol{f}_{\boldsymbol{k}} \cdot \boldsymbol{p}_{\boldsymbol{d}}\right),
    \label{eq:multi-nudft-appx}
\end{equation}
with $\cdot$ denoting the inner product.

\subsection{Nonuniform queries in \texorpdfstring{$\omega$-$k$}{omega-k} focusing}

Here $\boldsymbol x=(x_{ad})$ denotes baseband echoes sampled at azimuth slow times $\eta_a$ and fast-time delays $\tau_d$.
For the reference stripmap geometry, the Stolt relation maps the measurement-frequency coordinates
$(f_\eta,f_\tau)$ to an image-frequency range coordinate
$f'_\tau$~\cite{cumming2005digital}:
\begin{equation}
f_{\rm c}+f'_\tau
=
\sqrt{
(f_{\rm c}+f_\tau)^2
-
\left(\frac{c f_\eta}{2v_r}\right)^2
},
\label{eq:stolt-forward}
\end{equation}
where $f_{\rm c}$ is the carrier frequency, $v_r$ the effective platform velocity, and $c$ the speed of light.

Image formation requires spectral values on a regular output grid
$(f_{\eta,m},f'_{\tau,\ell})$. 
Inverting the Stolt relation gives the measurement-frequency location required at each grid point:
\begin{equation}
f_{\tau,m\ell}^{\mathrm q}
=
\sqrt{
(f_{\rm c}+f'_{\tau,\ell})^2
+
\left(\frac{c f_{\eta,m}}{2v_r}\right)^2
}
-f_{\rm c}.
\label{eq:inverse-stolt}
\end{equation}
We retain points for which this frequency lies within the acquired support. 
Although $f'_{\tau,\ell}$ is regularly spaced,
$f_{\tau,m\ell}^{\mathrm q}$ generally depends on both output indices. 
The required evaluation points
\begin{equation}
\boldsymbol f^\star_{m\ell}
=
(f_{\eta,m},f_{\tau,m\ell}^{\mathrm q})
\label{eq:sar-positions-queries}
\end{equation}
therefore form a non-Cartesian query set fixed by the acquisition and image grid.

For finite measurements, the value required at each query point can be evaluated directly:
\begin{equation}
Z_{m\ell}
=
\sum_{a,d}x_{ad}
\exp\!\left[
-2\pi i\left(
f_{\eta,m}\eta_a+
f_{\tau,m\ell}^{\mathrm q}\tau_d
\right)
\right].
\label{eq:sar-mnudft}
\end{equation}
This is the MNUDFT of the echoes sampled at positions
$(\eta_a,\tau_d)$. 
Conventional Stolt interpolation instead approximates the same finite-data spectral value from evaluations on a regular measurement-frequency grid. 
Applying the fixed filtering and spectral weights $W_{m\ell}$ to $Z_{m\ell}$, then taking the inverse transform on the regular output grid, gives the idealized operator used in Eq.~\eqref{eq:sar-operator}:
\begin{equation}
\mathcal R_{\mathrm{SAR}}
=
\mathcal F^{-1}_{\mathrm{img}}\,
\mathcal M\,\mathcal N_{\boldsymbol f^\star}.
\label{eq:sar-operator-appx}
\end{equation}
For any decision map on the resulting image, the fixed operations after the MNUDFT can be included in that map:
\begin{equation}
G_{\mathrm{SAR}}\circ\mathcal R_{\mathrm{SAR}}
=
\underbrace{
\left(G_{\mathrm{SAR}}\circ
\mathcal F^{-1}_{\mathrm{img}}\circ\mathcal M\right)
}_{\widetilde G_{\mathrm{SAR}}}
\circ\mathcal N_{\boldsymbol f^\star}.
\label{eq:sar-composite-classifier}
\end{equation}
This is the spectral-task factorization in Eq.~\eqref{eq:sar-decision}. 
Since the operations inside
$\widetilde G_{\mathrm{SAR}}$ preceding $G_{\mathrm{SAR}}$ are linear, a polynomial decision map in the focused amplitudes and their conjugates remains polynomial in the MNUDFT values and their conjugates.

\subsection{Local geometric phase and the circuit query}\label{appx:derivation-geom}

The spectral coordinates can be expressed as wavenumbers
$k_x=2\pi f_\eta/v_r$ and
$k_r=4\pi(f_{\rm c}+f_\tau)/c$ under the Fourier convention above.
The Stolt relation is then the dispersion relation
$k_y=\sqrt{k_r^2-k_x^2}$ of
Eq.~\eqref{eq:sar-dispersion}. Around a patch center
$(k_{x0},k_{r0})$, terms in $R_0k_y$ that depend on only
$k_x$ or only $k_r$ give the separable phases
$\phi_a$ and $\psi_d$ in Eq.~\eqref{eq:sar-local-expansion}. 
The first term coupling the two patch coordinates comes from
\begin{equation}
\frac{\partial^2 k_y}
{\partial k_x\,\partial k_r}
=
\frac{k_xk_r}{(k_r^2-k_x^2)^{3/2}}.
\label{eq:appx-sar-mixed-derivative}
\end{equation}
With patch spacings $\Delta k_x,\Delta k_r$, its phase coefficient is
\begin{equation}
f^\star_{\mathrm{geom}}
=
R_0\Delta k_x\Delta k_r
\frac{k_{x0}k_{r0}}
{(k_{r0}^2-k_{x0}^2)^{3/2}},
\end{equation}
as stated in Eq.~\eqref{eq:sar-geometry-query}. 
This coefficient is dimensionless. 
The circuit applies its negative as a phase correction; higher-order mixed terms remain in the generated data.

For completeness, the query encoding gate used in the six-qubit experiment has the stated cross-register structure. 
Write the three bits of
$a$ as $a=4a_2+2a_1+a_0$, and similarly for $d$. 
Since the eigenvalues of $Z$ on bit $a_r$ are $1-2a_r$,
\begin{equation}
\widetilde A
=
-\sum_{r=0}^{2}2^{r-1}Z_{A,r},
\qquad
\widetilde D
=
-\sum_{s=0}^{2}2^{s-1}Z_{D,s}.
\end{equation}
Thus $\widetilde A\otimes\widetilde D$ contains nine cross-register $Z_{A,r}Z_{D,s}$ terms. They commute, so each query upload decomposes into nine $R_{ZZ}$ rotations.

The operator-Schmidt rank of $\exp[-i(f/2)\widetilde A\otimes\widetilde D]$ is the rank of the $8\times8$ matrix
\begin{equation}
M_{ad}=\exp\bigl[-i(f/2)\,\tilde a\,\tilde d\bigr].
\end{equation}
Up to nonzero row and column factors, this is a Vandermonde matrix with nodes $\exp[-i(f/2)a]$, $a=0,\dots,7$.
It has rank eight when these nodes are distinct, as they are for the geometry-derived query values used here. 
The gate is then nonproduct across $A|D$ and cannot be absorbed into an adjacent trainable block of the form $U_A\otimes U_D$.
The full circuit is illustrated in Fig.~\ref{fig:sar-circuit}.

\begin{figure}[t!]
    \centering
    \begin{adjustbox}{width=\columnwidth}
    \begin{quantikz}
    \lstick{$\ket{+}^{\otimes3}_A$}
        & \qwbundle{3}
        & \gate{W_1^{A}}\gategroup[
          wires=2,
          steps=2,
          style={dashed,rounded corners,fill=myblue!80},
          background
      ]{}
        & \gate[2]{E^{(x)}(\boldsymbol x)}
        & \gate{V_1^{A}}\gategroup[
          wires=2,
          steps=2,
          style={dashed,rounded corners,fill=myred!80},
          background
      ]{}
        & \gate[2]{E^{(f)}(f)}
        & \gate{W_2^{A}}\gategroup[
          wires=2,
          steps=2,
          style={dashed,rounded corners,fill=myblue!80},
          background
      ]{}
        & \gate[2]{E^{(x)}(\boldsymbol x)}
        & \gate{V_2^{A}}\gategroup[
          wires=2,
          steps=2,
          style={dashed,rounded corners,fill=myred!80},
          background
      ]{}
        & \gate[2]{E^{(f)}(f)}
        & \gate{F_8}
        & \meter{} \\
    \lstick{$\ket{+}^{\otimes3}_D$}
        & \qwbundle{3}
        & \gate{W_1^{D}}
        & \qw
        & \gate{V_1^{D}}
        & \qw
        & \gate{W_2^{D}}
        & \qw
        & \gate{V_2^{D}}
        & \qw
        & \qw
        & \qw
    \end{quantikz}
    \end{adjustbox}
    \caption{
    \textbf{Circuit used in the SAR-like classification experiment.}
    The three-qubit registers $A$ and $D$ carry the binary digits of the azimuth index $a$ and the range index $d$ of the $8\times8$ patch.
    Each of the two blocks applies product trainable blocks
    $W_\ell^A\otimes W_\ell^D$ and $V_\ell^A\otimes V_\ell^D$, each consisting of $R_ZR_Y$ rotations on every qubit followed by a ring of CZ gates within the register, which gives $48$ trainable parameters.
    The signal is reuploaded as the phase-only diagonal
    $E^{(x)}(\boldsymbol x)=\mathrm{diag}_{ad}\,\exp(i\arg(x_{ad})/L_x)$, and the
    fixed query as
    $E^{(f)}(f)=\exp[-i(f/L_f)\widetilde A\widetilde D]$, which compiles into nine cross-register $R_{ZZ}$ rotations. 
    The readout applies the quantum Fourier transform $F_8$ to $A$ and measures it, giving the class distribution $p(j)$; $D$ is traced out.
    }
    \label{fig:sar-circuit}
\end{figure}
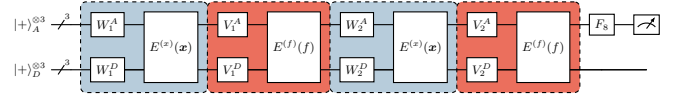

\subsection{SAR-like data generation and training}
\label{app:sar}

\textit{Acquisition.}
Table~\ref{tab:sar_acq} gives the acquisition parameters.
All acquisitions use the $8\times8$ sampling grid of Sec.~\ref{sec:sar-experiment}, with
$\Delta k_x=1.25$\,rad\,m$^{-1}$,
$k_{r0}=4\pi f_{\rm c}/c$, and
$\Delta k_r=4\pi\delta f/c$.
The baseline acquisition is used in Figs.~\ref{fig:sar-geometry-classification}a,c.

\begin{table*}[t]
\centering
\begin{tabular}{lccccc}
\hline
Acquisition & $f_{\rm c}$ (GHz) & $\delta f$ (MHz)
& $k_{x0}$ (rad\,m$^{-1}$) & $R_0$ (km)
& $f^\star_{\rm geom}$ \\
\hline
High carrier        & 14 & 5 & 125 & 3.7 & 0.377 \\
Low squint          & 10 & 5 &  80 & 3.7 & 0.467 \\
Carrier + bandwidth & 14 & 9 & 125 & 3.7 & 0.679 \\
Baseline            & 10 & 5 & 125 & 3.7 & 0.793 \\
Far range           & 10 & 5 & 125 & 4.8 & 1.029 \\
High squint         & 10 & 5 & 180 & 3.7 & 1.348 \\
High bandwidth      & 10 & 9 & 125 & 3.7 & 1.427 \\
Squint + bandwidth  & 10 & 7 & 180 & 3.7 & 1.887 \\
\hline
\end{tabular}
\caption{Acquisitions used in
Fig.~\ref{fig:sar-geometry-classification}b.
The query is calculated from Eq.~\eqref{eq:sar-geometry-query}.}
\label{tab:sar_acq}
\end{table*}

\textit{Scene generation.}
We evaluate $k_y=\sqrt{k_r^2-k_x^2}$ on the patch and define its dimensionless normalized nonseparable component,
\[
h_{ad}
=
\frac{k_y-\langle k_y\rangle_a-\langle k_y\rangle_d
+\langle k_y\rangle_{ad}}
{\kappa_{xr}\Delta k_x\Delta k_r},
\]
where $\kappa_{xr}$ is the mixed derivative in
Eq.~\eqref{eq:sar-geometry-query}, and averages are taken over the indicated indices.
The generated phase is $f^\star_{\rm geom}h_{ad}$; 
its deviation from the bilinear approximation obeys
$|f^\star_{\rm geom}(h_{ad}-\tilde a\tilde d)|<0.35$\,rad for every acquisition.
A scene of class $\ell\in\{0,\ldots,7\}$ is
\[
\begin{aligned}
&x_{ad}=s_a\Bigl[
\rho\,\exp(-i(ua+vd)+if^\star_{\rm geom}h_{ad})\\
&+\sum_{m=1}^{4}
c_m \exp(-i(u_ma+v_md)+ir_mf^\star_{\rm geom}h_{ad})
+n_{ad}\Bigr].
\end{aligned}
\]
The target and clutter parameters are sampled according to
\[
\begin{aligned}
u&=\frac{2\pi}{8}(\ell+0.48\xi),
&\xi&\sim U[-1,1],\\
\rho&=e^{i\psi},
&v,\psi&\sim U[0,2\pi),\\
u_m,v_m&\sim U[0,2\pi),
&r_m&\sim U[0.45,1.55],\\
c_m&\sim\mathcal{CN}(0,1.25^2/4),
&n_{ad}&\sim\mathcal{CN}(0,P/10),
\end{aligned}
\]
where $P$ is the mean target-plus-clutter power in the patch.
The noise variance scales with the patch power, maintaining a nominal SNR of $10$\,dB relative to the target-plus-clutter return.
The clutter range ratios $r_m$ give the additional scatterers different geometric phase curvatures, so the target-aligned correction generally leaves them with residual mixed phases.
Clutter and noise are added to the complex samples before phase extraction and therefore alter the phase-only input received by the classifier.
Separable propagation phases are represented by the random linear ramps.
The fixed channel phase is
\[
s_a=\exp\!\left[
-\frac{i}{2}\sum_{m=1}^{3}\phi_m(1-2b_m(a))
\right],
\quad
\boldsymbol\phi=(1.5,-1.7,1.0),
\]
where $b_m(a)$ is the $m$-th binary digit, most significant first.
It is shared across scenes and independent of the label.
All datasets are class-balanced, and the circuit receives only
$\arg x_{ad}$.

\textit{Circuit and optimization.}
The circuit of Fig.~\ref{fig:sar-circuit} starts in $\ket{+}^{\otimes6}$, with three-qubit registers $A$ and $D$
representing $a$ and $d$ in most-significant-bit-first order.
Each of its two blocks applies a trainable unitary, the signal encoding, a second trainable unitary, and the query encoding,
\[
\begin{aligned}
E_\ell^{(x)}
&=\operatorname{diag}_{a,d}
\!\left[\exp(i\arg(x_{ad})/2)\right],\\
E_\ell^{(f)}(f)
&=\operatorname{diag}_{a,d}
\!\left[\exp(-if\tilde a\tilde d/2)\right].
\end{aligned}
\]
The final query gate is followed directly by the readout.
Each trainable unitary is a product across $A|D$, with
$R_Z(\beta_q)R_Y(\alpha_q)$ on every qubit followed by a CZ ring within each register, giving $48$ parameters in total.
The readout applies $(F_8)_{ja}=\exp(2\pi ija/8)/\sqrt8$ on $A$ and marginalizes over $D$ to obtain $p(j)$; the prediction is $\arg\max_jp(j)$.

We minimize full-batch cross-entropy using $700$ Adam steps, with learning rate cosine-annealed from $0.04$ to $0.004$ and gradient norms clipped to $10$.
Each fit uses one zero initialization and five initializations with $\theta_i\sim\mathrm{normal} (0,0.15^2)$.
We retain the lowest training loss across all steps and starts.
Datasets and parameter initializations are paired across compared queries and models.
The trainable-query model additionally optimizes
$f=1.25\pi\tanh q$, shared by both query layers;
$q$ starts at zero in the zero initialization and at
$U[-1,1]$ in the other five.

\textit{Evaluation.}
Table~\ref{tab:sar_datasets} summarizes the datasets.
Panels~a,b retrain the circuit at $f/f^\star_{\rm geom}\in\{-1.5,-1.25,\ldots,2\}$.
For panel~b, $\hat f$ minimizes the validation cross-entropy averaged over draws. 
An interior minimum with locally convex neighboring values is refined by a parabola through the three grid points.

\begin{table}[h]
\centering
\begin{tabular}{lcccc}
\hline
Panel & Draws & Training & Validation & Test \\
\hline
a & 12 & 512 & -- & 4096 \\
b & 4  & 512 & 1024 & -- \\
c & 12 & 32--512 & -- & 4096 \\
\hline
\end{tabular}
\caption{Dataset sizes for
Fig.~\ref{fig:sar-geometry-classification}.
For panel~b, draws are performed separately for each acquisition.}
\label{tab:sar_datasets}
\end{table}

Panel~c uses nested, class-balanced training sets with $N_{\rm train}\in\{2^5,\ldots,2^9\}$ from each $512$-scene pool. 
It compares the aligned, wrong-sign, zero, and trainable queries.
Reported accuracies are balanced test accuracies.
Confidence intervals are $95\%$ percentile intervals from $5000$ bootstrap resamples over paired draws; $\hat f$ is re-estimated in every resample.
Random sampling uses NumPy, and circuits are evaluated with double-precision noiseless statevector simulation in JAX.

\end{document}